%% file: OnShellScatteringAmplitudes.tex
\documentclass[a4paper,10pt]{article}
\input{Header.tex}

\input{Graphs.tex}

\title{Scattering amplitudes in YM and GR\\
       as minimal model brackets\\
       and their recursive characterization}

\author{%
Andrea N\"utzi, ETH Zurich\thanks{andrea.nuetzi@math.ethz.ch}\\
Michael Reiterer, ETH Zurich\thanks{michael.reiterer@protonmail.com,
present affiliation: The Hebrew University of Jerusalem}}

\date{}

\begin{document}

% Titlepage
\maketitle
\begin{abstract}
\input{Abstract.tex}
\end{abstract}

% keywords

% toc
\tableofcontents

% Sections
\input{Introduction.tex}

\input{Preliminaries.tex}
\section*{Part I}
\addcontentsline{toc}{section}{\rule{0pt}{14pt}Part I}
This concerns the construction of tree scattering amplitudes using Feynman graphs.
This contains, if at a rather formal level,
the link between the classical partial differential field equations and amplitudes.
\input{MinimalModel.tex}
\input{YangMillsGeneralRelativity.tex}
\input{OptimalHomotopies.tex}
\section*{Part II}
\addcontentsline{toc}{section}{\rule{0pt}{14pt}Part II}
This concerns characterizing tree scattering amplitudes without reference to Feynman graphs,
in algebro-geometric language.
\input{VarietyAdmissibleMomenta.tex}

\input{HelicitySheafOneParticle.tex}
\input{HelicitySheaf.tex}
\input{RecursiveCharacterization.tex}

% Appendices
\appendix
\input{ModuleToSheaf.tex}

\input{FiberProduct.tex}

% References

\input{References.tex}
\end{document}

%% file: Header.tex
\usepackage{indentfirst}
\usepackage{amsmath,amssymb,amsthm}
\usepackage{calligra,mathrsfs}
\usepackage{graphicx,color}
\usepackage[all,cmtip]{xy}
\usepackage{bbm}
\usepackage{hyperref}
\usepackage{mathtools}
\usepackage{verbatim}
\usepackage{stmaryrd}
\usepackage{tikz}
\usetikzlibrary{matrix}

\usepackage{bm}

\definecolor{lcolor}{rgb}{0.6,0.3,0.3}
\hypersetup{colorlinks=true,linkcolor=lcolor,urlcolor=lcolor,citecolor=lcolor}

\definecolor{qcolor}{rgb}{0.19,0.55,0.91}
\definecolor{hcolor}{rgb}{0.9,0.2,0.5}
\definecolor{scolor}{rgb}{0.9,0.5,0.2}

\newcommand{\tc}[2]{#2}

\newcommand{\spl}[1]{\xrightarrow{\textnormal{\;split\;}#1}}
\newcommand{\p}{\partial}
\newcommand{\Z}{\mathbbm{Z}}
\newcommand{\R}{\mathbbm{R}}
\newcommand{\C}{\mathbbm{C}}
\newcommand{\one}{\mathbbm{1}}

\newcommand{\cone}{\tc{hcolor}{\textnormal{Cone}}}
\newcommand{\npar}{{\tc{scolor}{m}}}

\newcommand{\witness}{\tc{qcolor}{u}}

\newcommand{\eps}{\epsilon}
\newcommand{\recalleps}{{$\eps = (\begin{smallmatrix} 0 & 1 \\ -1 & 0 \end{smallmatrix})$}}

\newcommand{\nsheaf}{\mathcal{N}}

\newcommand{\step}{\vskip 3mm}

\newcommand{\fixy}{{a}}
\newcommand{\fixx}{{b}}
\newcommand{\fixxy}{{c}}
\newcommand{\fixY}{{A}}
\newcommand{\fixX}{{B}}
\newcommand{\fixXY}{{C}}
\newcommand{\ff}[1]{\mathcal{#1}}

\newcommand{\dgca}{\textnormal{dgca}}
\newcommand{\dgLa}{\textnormal{dgLa}}

\newcommand{\vecx}{{\tc{scolor}{V}}}

\newcommand{\ax}{{\mathfrak a}}

\newcommand{\gx}{{\mathfrak g}}
\newcommand{\hx}{{\mathfrak h}}
\newcommand{\mx}{{\mathfrak m}}
\newcommand{\nx}{{\mathfrak n}}
\newcommand{\ym}{{\mathfrak u}}
\newcommand{\gr}{{\mathfrak v}}

\newcommand{\mext}{\mathcal{M}}
\newcommand{\cext}{\mathcal{C}}

\newcommand{\kr}[1]{\C[k_{1\ldots #1}]}
\newcommand{\vwr}[1]{\C[v_{1\ldots #1},w_{1\ldots #1}]}

\newcommand{\cc}{\tc{qcolor}{\Gamma}}

\newcommand{\ux}{{\mathfrak u}}
\newcommand{\px}{{\mathfrak p}}
\newcommand{\qx}{{\mathfrak q}}

\newcommand{\so}{\tc{scolor}{{\mathfrak{so}}_{1,3}}}
\renewcommand{\sl}{{\mathfrak{sl}}_{2}}

\newcommand{\repaux}[1]{\tc{qcolor}{#1}}
\newcommand{\rep}[2]{\repaux{(}#1,#2\repaux{)}}

\DeclareMathOperator{\sheafHom}{\mathscr{Hom}\text{\kern -3pt {\calligra\large om}}\,}

\DeclareMathOperator{\Frac}{Frac}
\DeclareMathOperator{\image}{im}

\DeclareMathOperator{\End}{End}
\DeclareMathOperator{\Hom}{Hom}
\DeclareMathOperator{\Res}{Res}
\DeclareMathOperator{\MC}{MC}
\DeclareMathOperator{\tr}{tr}

\DeclareMathOperator{\divv}{div}

\newcommand{\fuse}[1]{\otimes_{#1}}

\renewcommand{\subset}{\subseteq}
\renewcommand{\supset}{\supseteq}

\renewcommand{\O}{\mathcal{O}}

\newcommand{\zl}{\tc{hcolor}{V}}

\newcommand{\M}{M} 
\newcommand{\wt}[1]{\smash{\widetilde{#1}}}

\newcommand{\sM}{\wt{\M}} % sheaf associated to module

\newcommand{\JX}{I'}

\newcommand{\Q}{\tc{qcolor}{Q}}
\newcommand{\hoff}{\tc{hcolor}{H}}

\newcommand{\dtot}{d_{\textnormal{tot}}}

\newcommand{\an}[1]{{\tc{qcolor}{\C^{#1}}}} 
\newcommand{\vn}{{\tc{scolor}{s}}}

\newtheoremstyle{mytheoremstyle}
{14pt}% space above
{14pt}% space below
{\itshape}% body font
{20pt}% indent amount
{\bfseries}% theorem head font
{.}% punctuation after theorem head
{.5em}% space after theorem head
{}% theorem head spec

\newtheoremstyle{myremarkstyle}
{10pt}% space above
{10pt}% space below
{}% body font
{20pt}% indent amount
{\itshape}% theorem head font
{.}% punctuation after theorem head
{.5em}% space after theorem head
{}% theorem head spec

\newtheoremstyle{myproofstyle}
{12pt}% space above
{12pt}% space below
{}% body font
{20pt}% indent amount
{\itshape}% theorem head font
{.}% punctuation after theorem head
{.5em}% space after theorem head
{}% theorem head spec

\theoremstyle{mytheoremstyle}
\newtheorem{theorem}{Theorem}[section]
\newtheorem{definition}[theorem]{Definition}

\newtheorem{lemma}[theorem]{Lemma}
\newtheorem{prop}[theorem]{Proposition}

\theoremstyle{myremarkstyle}
\newtheorem{remark}[theorem]{Remark}
\newtheorem{warning}[theorem]{Warning}
\newtheorem*{example}{Example}

\theoremstyle{myproofstyle}
\newtheorem*{PROOF}{Proof}

\newcommand{\hprimeaux}[1]{\tc{scolor}{h}_{\Q #1}}
\newcommand{\hprimeJ}{\tc{scolor}{h}_{\Q_J}}
\newcommand{\dprimeaux}[1]{\tc{scolor}{d}_{\Q #1}}
\newcommand{\hprime}{\hprimeaux{}}
\newcommand{\dprime}{\dprimeaux{}}

\providecommand{\keywords}[1]{{\footnotesize \textbf{\textit{Keywords---}} #1}}

%% file: Graphs.tex
\usetikzlibrary{decorations.markings}

\newcommand{\xshift}{2}
\newcommand{\yshift}{3}
\newcommand{\rcircle}{0.07}
\newcommand{\xshort}{4}
\newcommand{\dx}{1.25}
\newcommand{\dy}{1}

\newcommand{\haux}[5]{\draw [shorten >=\xshort,shorten <=\xshort,#3] (#1) -- (#2) node[midway,#4] {#5};}
\newcommand{\hl}[3]{\haux{#1}{#2}{}{anchor=east,xshift=-\xshift,yshift=\yshift}{#3}}
\newcommand{\hr}[3]{\haux{#1}{#2}{}{anchor=west,xshift=\xshift,yshift=\yshift}{#3}}
\newcommand{\il}[3]{\haux{#1}{#2}{dashed}{anchor=east,xshift=-\xshift,yshift=\yshift}{#3}}
\newcommand{\ir}[3]{\haux{#1}{#2}{dashed}{anchor=west,xshift=\xshift,yshift=\yshift}{#3}}

\newcommand{\pr}[3]{\haux{#1}{#2}{dashed}{anchor=west,xshift=\xshift,yshift=\yshift}{#3}}

\newcommand{\braux}[5]{\draw [black,fill=#1] (#2) circle (#3*\rcircle) node [#4] {#5};}
\newcommand{\br}[3]{\braux{qcolor}{#1}{1}{#2}{#3}}

\newcommand{\wrap}[2]{%
    \begin{tikzpicture}[xscale=#1,baseline=(current  bounding  box.center)]
     #2\end{tikzpicture}}

\newcommand{\graphbr}{\wrap{1}{
  \coordinate (a) at (-\dx/2,0);
  \coordinate (b) at (\dx/2,0);
  \coordinate (ab) at (0,0.7*\dy);
  \coordinate (out) at (0,1.6*\dy);
  \hl{a}{ab}{$k_1$}
  \hr{b}{ab}{$k_2$}
  \hr{ab}{out}{$k_1+k_2$}
  \br{ab}{}{}
}}

\newcommand{\graphpos}[1]{\wrap{1}{
  \coordinate (a) at (-2*\dx,0);
  \coordinate (b) at (-\dx,0);
  \coordinate (c) at (0,0);
  \coordinate (d) at (0.7*\dx,0);
  \coordinate (e) at (1.7*\dx,0);
  \coordinate (ab) at (-1.5*\dx,\dy);
  \coordinate (abc) at (-\dx,2*\dy);
  \coordinate (de) at (\dx,1.5*\dy);
  \coordinate (abcde) at (0,3*\dy);
  \coordinate (out) at (0,4*\dy);
  #1
}}

\newcommand{\graphx}[9]{\graphpos{
  \il{a}{ab}{#1}
  \ir{b}{ab}{#2}
  \hl{ab}{abc}{#3}
  \ir{c}{abc}{#4}
  \hl{abc}{abcde}{#5}
  \il{d}{de}{#6}
  \ir{e}{de}{#7}
  \hr{de}{abcde}{#8}
  \pr{abcde}{out}{#9}
  \br{ab}{}{}
  \br{abc}{}{}
  \br{abcde}{}{}
  \br{de}{}{}
}}

\newcommand{\graphintro}{\graphx{$i$}{$i$}{$\hoff$}{$i$}{$\hoff$}{$i$}{$i$}{$\hoff$}{$p$}}
\newcommand{\graphtech}{\graphx{$i_{k_1}$}{$i_{k_2}$}{$h_{k_1+k_2}$}{$i_{k_3}$}{$h_{k_1+k_2+k_3}$}{$i_{k_4}$}{$i_{k_5}$}{$h_{k_4+k_5}$}{$p_{k_1+k_2+k_3+k_4+k_5}$}}
\newcommand{\graphmomentum}{\graphx{$i_{k_1}$}{$i_{k_2}$}{$\hoff_{k_1+k_2}$}{$i_{k_3}$}{$\hoff_{k_1+k_2+k_3}$}{$i_{k_4}$}{$i_{k_5}$}{$\hoff_{k_4+k_5}$}{$p_{k_1+k_2+k_3+k_4+k_5}$}}

\newcommand{\graphfactoraux}[1]{\wrap{0.9}{
  \coordinate (a) at (-1.3*\dx,0);
  \coordinate (b) at (-0.7*\dx,0);
  \coordinate (c) at (-0.2*\dx,0);
  \coordinate (d) at (0.6*\dx,0);
  \coordinate (e) at (1.4*\dx,0);
  \coordinate (f) at (1.8*\dx,0);
  \coordinate (g) at (2.2*\dx,0);
  \coordinate (ab) at (-0.9*\dx,1.2*\dy);
  \coordinate (de) at (0.8*\dx,0.5*\dy);
  \coordinate (fg) at (1.85*\dx,\dy);
  \coordinate (cde) at (0.3*\dx,1*\dy);
  \coordinate (u) at (0.2*\dx,1.6*\dy);
  \coordinate (v) at (0.1*\dx,2.4*\dy);
  \coordinate (abcde) at (0*\dx,3*\dy);
  \coordinate (abcdefg) at (0.8*\dx,3.4*\dy);
  \coordinate (out) at (1*\dx,4.3*\dy);
  \il{a}{ab}{}
  \il{b}{ab}{}
  \il{c}{cde}{}
  \il{d}{de}{}
  \il{e}{de}{}
  \il{f}{fg}{}
  \il{g}{fg}{}
  \hl{fg}{abcdefg}{}
  \hl{ab}{abcde}{}
  \hl{de}{cde}{}
  #1
  \hl{abcde}{abcdefg}{}
  \pr{abcdefg}{out}{}
  \br{ab}{}{}
  \br{de}{}{}
  \br{fg}{}{}
  \br{cde}{}{}
  \br{abcde}{}{}
  \br{abcdefg}{}{}
}}

\newcommand{\graphfactor}{\graphfactoraux{
  \hr{cde}{abcde}{$\hoff_{k_J}$}
}}

\newcommand{\graphfactored}{\graphfactoraux{
  \pr{cde}{u}{$p$}
  \ir{v}{abcde}{$i$}
  \haux{u}{v}{}{anchor=west,xshift=\xshift,yshift=\yshift}{${\hprimeJ}$}
  \braux{hcolor}{u}{1.3}{}{}
  \braux{hcolor}{v}{1.3}{}{}
}}

%% file: Abstract.tex
Attached to both Yang-Mills and General Relativity about
Minkowski spacetime are distinguished
gauge independent objects known as the on-shell tree scattering amplitudes.
We reinterpret
and rigorously construct them as $L_\infty$ minimal model brackets.
This is based on formulating
YM and GR as differential graded Lie algebras.
Their minimal model brackets
are then given by a sum of trivalent (cubic) Feynman tree graphs.
The amplitudes are gauge independent when all internal lines are off-shell,
not merely up to $L_\infty$ isomorphism,
and we include a homological algebra proof of this fact.
Using the homological perturbation lemma,
we construct homotopies (propagators) that are optimal
in bringing out the factorization of the residues of the amplitudes.
Using a variant of Hartogs extension for singular varieties,
we give a rigorous account of a recursive characterization of the amplitudes
via their residues independent of their original definition in terms of Feynman graphs
(this does neither involve so-called BCFW shifts
 nor conditions at infinity under such shifts).
Roughly, the amplitude with $N$ legs is the unique section
of a sheaf on a variety of $N$ complex momenta
whose residues along a finite list of irreducible codimension one subvarieties
(prime divisors)
factor into amplitudes with less than $N$ legs.
The sheaf is a direct sum of rank one sheaves
labeled by helicity signs.
To emphasize that amplitudes are robust objects,
we give a succinct list of properties that suffice for a
dgLa so as to produce the YM and GR amplitudes respectively.

%% file: Introduction.tex
\section{Introduction}

Gauge theories in physics are redundant
in their description of physical processes;
to extract actual physical information one must quotient by a large gauge group.
Tree scattering amplitudes, specifically the on-shell tree scattering amplitudes
in Yang-Mills theory (YM) and General Relativity (GR),
are interesting because they are gauge independent,
devoid of any redundancy.
These objects are extensively studied in the physics literature,
see for example
the Parke-Taylor formulas for YM amplitudes in \cite{parketaylor}
exposing significant cancellations among Feynman graphs
for certain helicity configurations\footnote{%
The Parke-Taylor formulas are usually given for the color-ordered or partial YM amplitudes,
from which one can obtain the YM amplitudes.
Color-ordering is routine terminology in the literature on YM amplitudes,
but we do not use it in this paper.
}\textsuperscript{,}\footnote{%
A related instance of cancellations, there are helicity configurations
where the amplitudes are identically zero \cite{parketaylor0}.
The Parke-Taylor formulas are for the simplest nonzero components
and are known as maximally helicity violating (MHV) amplitudes.};
the discussion of the factorization of residues along
collinear and multiparticle singularities in \cite{mpx}
and earlier $S$-matrix literature;
the Britto-Cachazo-Feng-Witten (BCFW) recursion relations in \cite{bcfw,arkanikaplan};
the monograph \cite{arkanietal} and a textbook \cite{elyu}.
Tree scattering amplitudes are so named because they can be defined
as a sum of Feynman tree graphs,
and they are associated with classical physics,
whereas graphs with loops are associated with quantum physics.
\step
An interesting open question for mathematics,
not answered in this paper but an important motivation,
is how these gauge independent amplitudes are rigorously related
to solutions
to the partial differential field equations;
for GR
these are near-Minkowski solutions to the Einstein equations for vacuum, $\textnormal{Ricci} = 0$\footnote{%
Near-Minkowski solutions were rigorously constructed in \cite{ck}.}.
One would like to see rigorously how the
tree scattering amplitudes relate past incoming to future outgoing asymptotic data
for actual solutions to the
semilinear (YM) respectively quasilinear (GR)
hyperbolic partial differential field equations,
not an easy task.
This would clarify to what extent these amplitudes `describe
the nonlinear interaction of gravitational waves' as modeled by GR.
\step
This paper is about the tree scattering amplitudes themselves,
developed rigorously from their definition as $L_\infty$ minimal model
brackets implemented using Feynman graphs
to their recursive characterization independent of Feynman graphs.
It includes details often taken for granted in the literature
on amplitudes,
and departs conceptually from it
because the development is consistently homological and 
geometrical. Here is a summary:
\newcommand{\dc}[1]{\textnormal{\emph{#1}}}
\begin{description}
\item[\dc{Part I comprising Sections
\ref{sec:mm}, \ref{sec:ymgr}, \ref{sec:opthom}.}]
To prove gauge independence and to treat YM and GR in a unified way,
we define the tree scattering amplitudes in a non-standard manner
as the $L_\infty$ minimal model brackets of a differential graded Lie algebra (dgLa)
with an additional `momentum' grading.
In the mathematics literature on the homotopy transfer theorem, e.g.~\cite{lv},
it is shown that the minimal model brackets of a dgLa can be
defined as a sum of trivalent (cubic) Feynman tree graphs
 of the caricatural form
\[
\graphintro
\]
where each node stands for an application of the dgLa bracket $[-,-]$ whose inputs are the two lines
coming in on the bottom and whose output is the line leaving on the top,
whereas the lines (edges) are decorated by a contraction for the dgLa differential $d$
encoding a choice of gauge\footnote{%
In physics terminology, the on-shell quasi-isomorphisms $i$ and $p$ choose polarization vectors,
the off-shell homotopy $\hoff$ is the propagator.
The differential $d$ is, as a Fourier multiplier,
a matrix with entries polynomial in the momentum $k \in \C^4$,
with $d^2 = 0$,
with homology on-shell but no homology off-shell.
The shell is the light cone in momentum space $\C^4$.}.
We prove gauge independence in this homological language,
namely that individual trees depend on the gauge but the sum of trees does not.
Only then do we show that YM and GR can actually be formulated as the
Maurer-Cartan equations $du + \tfrac12 [u,u] = 0$ in a suitable dgLa\footnote{%
The unknown $u$ stands for, roughly,
a Lie algebra valued connection one-form and curvature two-form in YM,
respectively an orthonormal frame and connection in GR.
These are formulations where the partial differential field equations are first order.
}.
That these dgLa yield the partial differential field equations
and gauge transformations of YM and GR
suffices for our purpose;
we omit a more direct and pedantic proof that the tree amplitudes as we define
them coincide with those in the physics literature,
defined in terms of the (non-trivalent) Feynman tree graphs
associated to the YM and GR Lagrangians.
(The details of the construction of these dgLa do not permeate this paper.
It is an important by-product of the recursive characterization in Part II
 that all dgLa with the right homology, nontrivial bracket,
 Lorentz invariance and suitable homogeneity,
 yield the same amplitudes.)
The homotopy $\hoff$ is not unique and is the direct homological analog
of the propagator in the physics literature,
a matrix whose entries are rational functions of the momentum $k \in \C^4$
that is necessarily singular when $k$ is on the light cone.
We construct `optimal homotopies' $\hoff$ using the homological perturbation lemma,
whose residue at the light cone factor in a specific way
(namely with an $i$ factoring out on the left and a $p$ factoring out on the right)
and which at once implies that the residues of the corresponding poles of the scattering 
amplitude factor, key for Part II.
These homotopies are defined locally,
on Zariski open sets, and hence so are individual trees,
but the amplitudes can be glued to global sections of a sheaf
by gauge independence; gluing is in Part II.
The main results of Part I are Theorems \ref{well-defined2}, \ref{theorem:ymgrnew}, \ref{theorem:opthom}.
\item[\dc{Part II comprising Sections
\ref{sec:variety}, \ref{sec:helicityone}, \ref{sec:helicity}, \ref{sec:recchar}.}]
The purpose of this part of the paper
is to derive recursion relations for tree scattering amplitudes $\{-,\ldots,-\}$
of the caricatural form, known as factorization of residues\footnote{%
So the residues factor into
products of amplitudes with fewer legs, that is,
compositions of minimal model brackets with fewer arguments.
The case of four legs,
that is, minimal model brackets with three arguments,
can involve an extra sum on the right.
For details see \eqref{eq:resrule}.},
\begin{equation}\label{eq:rr}
\Res_\px \{-,\ldots,-\} \;=\;\{-,\ldots,\{-,\ldots,-\}\}
\end{equation}
and to show that these recursions characterize the amplitudes uniquely,
given the homogeneity degrees of the amplitudes,
given permutation invariance,
and given the amplitudes with three legs, $\{-,-\}$,
which in turn admit a simple characterization using Lorentz invariance.
(This detaches tree scattering amplitudes from Feynman graphs,
 and indeed Part I enters Part II only to prove the existence of a solution
 to these recursions.
 One can go ahead and try to find other constructions without
 Feynman graphs that satisfy the recursions;
 literature along these lines includes \cite{csw,arkanietal}.
 Interesting as such approaches are,
 they are no replacement for Part I for those interested in studying 
 the connection between amplitudes
 and the partial differential field equations, namely the Maurer-Cartan equations,
 in more depth, a key motivation for this paper.)
 To make sense of \eqref{eq:rr}
 we study the complex algebraic variety of kinematically admissible momenta;
 we classify the irreducible codimension one subvarieties,
 labeled by prime divisors $\px$, along which
 amplitudes can have poles\footnote{%
  Poles can only occur when an internal momentum (a sum of a subset of momenta
  excluding trivial subsets and those equal to a single momentum)
  goes on-shell, but this is not always an irreducible variety.
  The irreducible components are labeled by
  height one prime ideals $\px$.
  For some helicity configurations, some potential poles are absent
  due to cancellations.};
 we clarify what we mean by the residue along $\px$;
 we introduce sheaves suitable for YM respectively GR amplitudes;
 we construct a codimension two subset where amplitudes need not be defined;
 and we obtain a variant of Hartogs extension for these sheaves\footnote{%
 Recall that the classical Hartogs extension theorem says that a holomorphic function
 of two complex variables $z,w$ defined on say $0 < |z|^2 + |w|^2 < 1$
    extends holomorphically to the origin,
 so there can be no singularities in codimension two.
 For singular varieties or sheaves,
 an analogous Hartogs phenomenon may or may not hold.
 It holds for the structure sheaf of complete intersections,
 sometimes said to have mild singularities,
 and we reduce our statement to this case.
 It can fail in the presence of wilder singularities, see Remark \ref{remark:hfail}.}.
 Hartogs is used to prove the recursive characterization.
 A hint of a proof of the recursive characterization,
 glossing over geometric details,
 is at the end of \cite{bcfw}.
 (This discussion at the end of \cite{bcfw}
  is separate from, and not to be confused with, the main result of that paper
  known as BCFW recursion, see below.)
 The recursive characterization implies qualitative properties
 that are not clear from Feynman graphs,
 such as the pole structure of the Parke-Taylor formulas.
 In our non-Lagrangian setup it is not immediate
 that the amplitudes $\{-,\ldots,-\}$, which are invariant under permuting the inputs,
 are also invariant under exchanging the output with an input
 relative to suitable bilinear forms,
 but we prove this also using Hartogs.
 The main result of Part II are Theorems
 \ref{theorem:pd4}, \ref{theorem:nlcm}, \ref{theorem:n3h}, \ref{theorem:unique}, \ref{theorem:exists}.
\end{description}
%%%%%%%%%%%%%%%%%%%%%%%%%%%%%%%%%%%%%%%%%%%%%%%%%%%%%%%%%%%%%%%%%%%%%%%%%%%%%%%%

The following things appear to be new:
the interpretation of on-shell tree scattering amplitudes as $L_\infty$ minimal model brackets
and the unified treatment of YM and GR made possible by this,
including a homological proof of gauge independence;
the purely trivalent tree Feynman graphs particularly in GR\footnote{%
By contrast, the Einstein-Hilbert Lagrangian for GR
is not polynomial and the associated Feynman rules, about Minkowski spacetime,
yield nodes of arbitrarily high degree.};
the classification of prime divisors along which amplitudes can have poles;
the construction of optimal homotopies using the homological perturbation lemma
from which the factorization of residues easily follows;
a detailed statement and proof of the recursive characterization, using a Hartogs extension argument.
\step
 
Let us discuss how
BCFW recursion \cite{bcfw} compares.
In its original form,
this effectively restricts amplitudes
to suitable $\C^2$ subspaces contained in the variety of kinematically admissible momenta
(there are enough such subspaces)
and by homogeneity one passes to the Riemann sphere $\mathbbm{P}(\C^2)$
hence one complex variable.
The north pole is a point of the form
$P = (0,\ldots,0,q,0,\ldots,0,-q,0,\ldots,0)$ with on-shell $q \in \C^4-0$.
One uses the familiar theorem from complex analysis by which
a meromorphic function on $\mathbbm{P}(\C^2)$
with simple poles is determined by the location and residues of the poles
up to an additive constant\footnote{%
Incidentally, the classical Hartogs extension theorem is one way to prove this.
By using the residue theorem one can produce explicit formulas.
}.
The residues come from restricting \eqref{eq:rr}
but one also needs control at the point at infinity $P$,
which is not immediate \cite{bcfw,arkanikaplan}.
By comparison, the recursive characterization in this paper does 
not require checking conditions at points such as $P$
because it exploits the Hartogs phenomenon.
It requires \eqref{eq:rr}
only away from a codimension two subset $Z$ of the variety;
$Z$ contains the entire singular locus and in particular $P \in Z$.
Roughly, the amplitudes are the unique sections given away from $Z$
that satisfy the recursions \eqref{eq:rr} away from $Z$.
No conditions need to be checked along $Z$.
This seems to be a simple and natural recursive characterization,
hence we give a rigorous account of this,
with Hartogs as a simple proof.
\step

\emph{Acknowledgments:}
We thank Thomas Willwacher for very useful comments on parts of this manuscript;
Eugene Trubowitz for studying $L_\infty$ algebras and homotopy transfer with us;
and Horst Kn\"orrer for his interest.

%--------------------------------------------------------------------------
\section{An informal overview}

The technical sections of this paper are, at least to some extent,
logically independent units.
This entails that the connections and the bigger picture may not always be clear.
The only two places where things are tied together is this technical but informal overview,
where we have taken the freedom to reorganize the material to simplify the discussion,
and then rigorously in Section \ref{sec:recchar}.
\step
{\bf Homological framework.}
    Both YM and GR about Minkowski spacetime admit
    a differential graded Lie algebra (dgLa) formulation.
    The partial differential classical field equations are
    the Maurer-Cartan equations
    \[
         du + \tfrac{1}{2}[u,u] = 0
    \]
    for the unknown $u$, an element of degree one in the dgLa.
    In particular, in the case of GR the MC-equations are
    equivalent to the vanishing of the Ricci curvature.
    Elements of degree zero act as infinitesimal gauge transformations.
    References are \cite{zeitlin,costello} for YM and \cite{rt} for GR.
    The only nonlinearity is the bracket,
    which will give rise to trivalent Feynman tree graphs.

    For amplitudes, we use a variant of these dgLa,
    set in complex momentum space, $k \in \C^4$.
    So $k$ is formally the variable appearing in the Fourier transform.
    This dgLa $\gx$ comes with an additional `momentum' grading,
    $\gx = \bigoplus_{k \in \C^4} \gx_k$, that is respected by the operations,
    $d \gx_k \subset \gx_k$ and $[\gx_{k_1},\gx_{k_2}] \subset \gx_{k_1+k_2}$:
    \[
      \graphbr
    \]
    We have $\gx_k \simeq \vecx$ for a
    finite-dimensional graded vector space $\vecx$.
    Relative to a basis of $\vecx$,
    the differential and bracket are
    arrays whose entries are polynomials in the momenta $k$ and $k_1,k_2$ respectively.
    Actually, they are first order polynomials, corresponding to first order
    partial differential operators.
 
    Elements of $\gx$
    may be interpreted as finite linear combination of plane waves,
    namely interpret $v \in \gx_k \simeq \vecx$
    as the plane wave $\R^4 \ni x \mapsto v e^{ikx}$.
    However, we use $\gx$ not because we are interested in finite linear combinations of plane waves,
    but because it is convenient for the definition of amplitudes.
\step
    {\bf The homology of YM and GR.}
    The differential $d$ describes the linearized problem.
    Its homology is a distinguishing feature of YM and GR.

    Let $\Q$ be the Minkowski square of $k$,
    so $\Q=0$ is the complex light cone.
    It is convenient to let $k$ be complex $2\times 2$ matrices
    and to set $\Q = \det k$.
    Both YM and GR have homology only on-shell, meaning on the cone,
     $\hx = \bigoplus_{k: \Q=0} \hx_k$ where $\hx_k = H(\gx_k)$,
    and only in degrees $i=1,2$. These split into helicities,
    \[
      \hx_k^i = \hx_k^{i,+} \oplus \hx_k^{i,-}
    \]
    The $\hx_k^{i,\pm}$ define rank one sheaves on the cone away from the origin,
    times a Lie algebra in the case of YM
    that we refer to as the internal Lie algebra.
    %------------------------------------------
\step
{\bf Amplitudes as minimal model brackets. Gauge independence.}
We define the amplitudes to be the $L_\infty$ minimal model brackets\footnote{%
These are all minimal model brackets given that
$\hx$ is concentrated in degrees one and two,
because the minimal model bracket $\hx^{\otimes n} \to \hx$ has degree $2-n$.}
\begin{equation} \label{eq:prelbr}
   \{-,\ldots,-\} \quad:
                  \quad
    \hx_{k_1}^1 \otimes \cdots \otimes \hx_{k_n}^1
    \to \hx_{k_1 + \ldots + k_n}^2
\end{equation}
with all internal lines off-shell, meaning
for all $J \subset \{1,\ldots,n\}$ with $1 < |J| < n$,
the momentum $k_J = \sum_{i \in J} k_i$ is not on the cone.
The minimal model brackets are defined as a sum of trivalent tree graphs,
implementing $L_\infty$ homotopy transfer
\cite{kontsevich,huebschmannstasheff,berglund,ksoibelman}.
A typical such tree is
\begin{equation}\label{eq:graphintro}
\graphmomentum
\end{equation}
standing for the nested composition
$\pm p[\hoff[\hoff[i-,i-],i-],\hoff[i-,i-]]$.
One needs the following objects to define trees,
informally:
\begin{itemize}
\item \emph{An off-shell homotopy for every internal line:}
A matrix $\hoff$ with entries rational
in $k$ such that $\hoff^2 = 0$ and $\hoff d + d \hoff = \one$.
Singular along $\Q = 0$.
Its evaluation at $k$ is a map $\hoff_k: \gx_k \to \gx_k$
of degree minus one.
\item \emph{An on-shell contraction for every external line:}
A contraction $(h,i,p)$
meaning
$hdh = h$, $h^2 = 0$, $ip = \one - dh-hd$, $pi = \one$.
Further $dhd=d$ along $\Q=0$.
They are matrices with entries rational in $k$,
regular including along $\Q = 0$.
Here $i_k : \hx_k \to \gx_k$ and $p_k : \gx_k \to \hx_k$
of degree zero.
\end{itemize}
Then \eqref{eq:prelbr}
is the sum of all trivalent trees decorated
with $i$ at every input,
$p$ at the output, $\hoff$ at every internal line,
and $[-,-]$ at each node.

The homotopies and contractions encode gauge choices.
But while individual trees depend on these choices,
we prove that the sum of all trees does not if all internal lines are off-shell.
The gauge independence is in the strong sense,
and not merely up to $L_\infty$ isomorphism
which holds for abstract reasons \cite{kontsevich}.

Gauge independence is a prerequisite to constructing global objects.
In fact the homotopies and contractions
are only required to be
regular on some open patch of momenta $k_1,\ldots,k_n$.
Gauge independence allows one to glue
these locally valid tree expansions to construct a global object.

The amplitudes are invariant under the group $S_n$
that permutes the $n$ inputs,
but invariance under $S_{n+1}$ is not by construction,
since our setting is without a Lagrangian
and it distinguishes inputs and outputs.
We prove invariance under $S_{n+1}$
at a late stage, using residues and the recursive characterization.

%---------------------------------------
\step
{\bf Optimal homotopies via the homological perturbation lemma.}
    The requirement $\hoff d + d \hoff = \one$
    means that the homotopy $\hoff$ is an inverse of $d$ 
    in a homological sense, and as such it is highly non-unique.
    Sticking to momentum conserving homotopies, $\hoff$ is a
    matrix whose entries are rational functions of $k$.
    However, we never write down explicit matrix entries.
    In fact the $\hoff$ that we use would not necessarily
    be convenient for explicit by-hand calculations, instead
    they distinguish themselves by their structure.

    A homotopy necessarily has a singularity near the cone, of type $1/\Q$.
    By an optimal homotopy we mean
    one where $1/\Q$ is multiplied by a matrix of the lowest possible rank.
    We show that every point on the cone
    has a Zariski open neighborhood on which there is a homotopy
    of the form
    \begin{equation}\label{eq:homx}
         \hoff \;=\; h + \frac{1}{\Q}i{\hprime}p
    \end{equation}
    with regular $h$, $i$, ${\hprime}$, $p$.
    Here ${\hprime}$ is a two-by-two matrix\footnote{%
      Actually a block matrix that is zero except for a two-by-two block
      going from homological degree two to one.
      Times the dimension of the internal Lie algebra in the case of YM.}
     so that $i{\hprime}p$ has the lowest possible rank.
    The restriction ${\hprime}|_{\Q = 0}$ is a canonical
    isomorphism $\hx^2 \to \hx^1$\footnote{%
    We describe in words how this map arises, see Section \ref{sec:opthom} for details.
    The derivative of $d$ transversal to $\Q=0$ induces a differential
    on the homology $\hx$ along $\Q=0$.
    It has a canonical normalization, using $\Q$.
    For YM and GR one obtains a complex $0 \to \hx^1 \to \hx^2 \to 0$
    at each point of the cone away from the origin.
    It is exact, hence induces a map $\hx^2 \to \hx^1$.
    }.
    Intuitively, $p$ receives and $i$ emits an on-shell particle.

    Our construction of $\hoff$ is conceptual,
    without reference to YM or GR.
    It is based on a nested application
    of the homological perturbation lemma.
%---------------------------------
\step
{\bf The momenta variety and the helicity sheaf for $N = n+1$ particles.} 
    The amplitudes are naturally sections of a sheaf on a variety
    that we refer to as the variety of kinematically admissible
    complex momenta,
    \[
        \big\{ (k_1,\ldots,k_N) \in (\C^4)^N \;\mid\; \Q_1 = \ldots = \Q_N = 0,
               \quad k_1 + \ldots + k_N = 0 \big\}
    \]
    a direct product of $N$ cones,
    intersected with a codimension 4 plane.
    We discuss its geometry.
    It has two irreducible components when $N=3$,
    but it is irreducible when $N \geq 4$.
    It  is a complete intersection and therefore
    has the property $S_2$ of Serre so that a variant of
    Hartogs extension holds for the structure sheaf.

    The amplitudes
    will be sections of a sheaf that is a product of $N$ sheaves,
    one for every input and output,
    times the sheaf of an effective divisor
    to accommodate $1/\Q$ singularities in the homotopies.
    The tree expansion yields sections on open patches.
    Gluing yields amplitudes on the complement of a subset $Z$
    of codimension two, a union of various sets
    where one does not have a valid tree expansion.
    Hartogs can limit the behavior along the undefined locus $Z$.
%---------------------------------
\step
{\bf Prime divisors and the factorization of residues.} 
        For every subset $J \subset \{1,\ldots,N\}$
        such that both $J$ and its complement $J^c$
       have at least two elements,
    define the internal momentum
    \[
        k_J = \textstyle\sum_{i \in J} k_i
    \]
    with Minkowski square $\Q_J$.
    Homotopies generally introduce $1/\Q_J$ singularities in amplitudes,
    known as collinear or multiparticle singularities.
    If $|J|, |J^c| \geq 3$
    then the codimension one subvariety $\Q_J = 0$ is irreducible,
    but in general it decomposes.
    Each component is a prime divisor $\px$
    and we classify them.
    Associated to each $\px$ is a
    residue condition for minimal model brackets, schematically\footnote{%
    Both sides are sections of a sheaf along the variety $\zl(\px)$,
    away from where this intersects $\zl(\qx)$
    for other such prime divisors $\qx \neq \px$.}
    \begin{equation}\label{eq:resc}
      \Res_\px \big\{-,\ldots,-\big\} \;=\;
      \big\{-,\ldots,-,{\hprimeJ}\big\{-,\ldots,-\big\}\big\}
    \end{equation} 
    There is actually a sum on the right hand side
    over all $J$ such that $\px$ is a minimal prime over $(\Q_J)$,
    but it always degenerates to a single term if $N > 4$.
    There are interesting exceptions in the case $N=4$,
    where the sum is over three terms.

    The residue condition can be thought of as a factorization.
    It follows immediately if, exploiting gauge independence,
    one uses optimal homotopies \eqref{eq:homx}.
    We illustrate this for a small tree and a $J$ for which
    $(\Q_J)$ is itself a prime divisor:
\[
\Res_{\Q_J}
\left(\;\graphfactor\;\right)
\;=\;
\graphfactored
\]
The subtree below ${\hprimeJ}$ is terminated by a $p$,
and the subtree above is entered with an $i$.
Hence, crucially, the subtrees
are again of the general form \eqref{eq:graphintro}.

    If we sum over all trees, then only those
    containing $k_J$ as an internal momentum contribute to this particular residue.
    On the right hand side we get a sum that factors into a double sum
    over all trees below and all trees above the low-rank operator ${\hprimeJ}$.
    They are reassembled into two brackets, giving \eqref{eq:resc}.

    We prove that these recursive conditions uniquely determine the amplitudes.
    This is what we refer to as the unique recursive characterization.
%----------------------------------------------------------
\step
\newcommand{\si}{S} 
{\bf Recursive characterization. Hartogs extension. Local cohomology.}
Let $X$ be the 
variety of kinematically admissible momenta,
$R$ its coordinate ring.
The amplitudes are only constructed away from a closed subset $Z \subset X$
of codimension two;
$Z$ is the union of all pairwise intersections of the zero loci of distinct prime divisors.
For every helicity configuration there is a sheaf $\sM$
associated to some finitely generated graded $R$-module $\M$\footnote{%
Here $\wt{\phantom{x}}$ is the module-to-sheaf functor,
defined by $\wt{M}(\{f \neq 0\}) = M_f$
for every nonzero ring element $f$,
where $M_f$ is module localization at the powers of $f$.
See Appendix \ref{appendix:m2s}.
}.
We refer to the sheaf that does not allow poles,
so the amplitudes are not in $\sM(X-Z)$ because they have poles,
but the difference of two amplitudes with identical residues is in $\sM(X-Z)$.
Consider the restriction map of graded $R$-modules
\begin{equation}\label{eq:rmrm}
\sM(X) \to \sM(X-Z)
\end{equation}
Its failure to be injective or
surjective is measured by respectively
the 0th and 1st local cohomology modules\footnote{%
The prototypical case is when $Z$ is a single point, but we have a bigger $Z$.
}.
Though $M$ is not locally free,
a variant of Hartogs extension shows that
\eqref{eq:rmrm} is an isomorphism.
The recursive characterization
follows since $\sM(X) = M$ is empty in the relevant degree,
for YM and GR\footnote{%
A more ambitious use of local cohomology \cite{localcohomology}
would be to derive detailed properties of the amplitudes
near various singular points of the variety.
Essentially one asks how constraining the structure of the sheaf itself is,
when combined with permutation symmetry and Lorentz invariance and homogeneity.
In this case, in \eqref{eq:rmrm}
one must replace $M$ by 
a module (just) big enough to accommodate poles,
so that the amplitude is a section on $X-Z$.}.

%% file: Preliminaries.tex
\section{Preliminaries} \label{sec:pre}

Here we collect a few definitions and facts that are used in several sections,
the reader may want to at least skim this section before moving on.
\step

{\bf Ground field.} All vector spaces and
algebras and varieties are over the complex numbers $\C$.
Tensor products of vector spaces are over $\C$.

\step
{\bf Differential. Homotopy. Contraction. Homotopy equivalence.}
On a $\Z$-graded vector space or module,
by a \emph{differential} we mean an endomorphism $d$ of degree one that satisfies $d^2 = 0$.
A space with a differential is called a complex.
We use the following terminology:
\begin{itemize}
\item A \emph{homotopy} for $d$ is an endomorphism $h$ of degree minus one
that satisfies
\[
  h^2 = 0 \qquad hdh = h
\]
Every homotopy yields three mutually orthogonal projections
\begin{equation}\label{eq:threep}
  dh \qquad\qquad (\pi=)\mathbbm{1}-dh-hd \qquad\qquad hd
\end{equation}
\item A \emph{contraction} for $d$
is what some authors call a strong deformation retract with side conditions.
Namely a triple $(h,i,p)$ where $h$ is a homotopy as above
and where $i,p$ are linear maps such that
\[
  pi = \mathbbm{1} \qquad ip = \mathbbm{1} - dh - hd
\]
Note that $hi = ph = 0$.
The codomain of $p$ and domain of $i$ is a second graded module,
as in the non-commutative \emph{contraction diagram}:
\begin{equation}\label{eq:hip}
\begin{tikzpicture}[baseline=(current  bounding  box.center)]
  \matrix (m) [matrix of math nodes, column sep = 32mm, minimum width = 4mm ]
  {
    \bullet & \bullet \\
  };
  \path[-stealth]
    (m-1-1) edge [out=15,in=180-15] node [above] {$p$} (m-1-2)
            edge [out=180+25,in=180-25,min distance=8mm] node
                 [left,xshift=-1mm] {$d$ \footnotesize differential} (m-1-1)
            edge [loop above,out=90+25,in=90-25,min distance=8mm] node
                 [above] {$h$ \footnotesize homotopy} (m-1-1)
    (m-1-2) edge [out=180+15,in=-15] node [below] {$i$} (m-1-1)
            edge [out=25,in=-25,min distance=8mm] node
                 [right,xshift=1mm] {$pdi$ \footnotesize differential} (m-1-2);
\end{tikzpicture}
\end{equation}
Using the differential $pdi$,
the maps $p$ and $i$ are a homotopy equivalence.
\item A \emph{homotopy equivalence} between two complexes $(C,d)$ and $(C',d')$
are maps $R \in \Hom^0(C,C')$, $L \in \Hom^0(C',C)$,
$\witness \in \End^{-1}(C)$, $\witness' \in \End^{-1}(C')$
where $R,L$ are chain maps and
$L R = \one - d\witness - \witness d$
and $RL = \one - d'\witness'-\witness'd'$. 
We usually consider all four maps to be part of the data.
Note that $R,L$ are quasi-isomorphisms.
The composition of two homotopy equivalences
$C \leftrightarrow C'$ and $C' \leftrightarrow C''$
is a homotopy equivalence $C \leftrightarrow C''$.
\end{itemize}

We often require that a homotopy satisfy $dhd = d$.
Then the images of the projections in \eqref{eq:threep} are
respectively $\image d$,
a complement of $\image d$ in $\ker d$,
and a complement of $\ker d$.
Conversely, for every choice of two such graded complements
there is a unique corresponding such homotopy\footnote{%
In particular every differential on a \emph{vector space} admits a homotopy with $dhd=d$.
On infinite-dimensional vector spaces,
the existence of complementary subspaces is not immediate,
but it is standard and convenient to adopt axioms that imply the existence.}.
For a contraction, $dhd=d$ is equivalent to $di = 0$ or $pd = 0$ or $pdi = 0$.
In this case the space on the right in \eqref{eq:hip}
is canonically isomorphic, via $i$ and $p$, to the homology of $d$.
\step

%------------------------------------------

{\bf Homological perturbation lemma.} 
Given a contraction and a perturbation of the differential,
the HPL produces a perturbed contraction.
Explicitly, if the perturbed differential is called $d'$,
and if we abbreviate $\delta = d'-d$, then
\begin{equation}\label{eq:hpl}
  h' = h(\mathbbm{1} + \delta h)^{-1}
  \qquad
  i' = (\mathbbm{1} + h \delta)^{-1} i
  \qquad
  p' = p(\mathbbm{1} + \delta h)^{-1}
\end{equation}
is the new contraction      
if $\delta$ is suitably small so that the inverses are defined.
The HPL keeps the spaces fixed and only perturbs the arrows in \eqref{eq:hip}.
Beware that $dhd=d$ does not imply $d'h'd'=d'$.
See \cite{crainic} for an exposition.
One may think of $h-h' = h'(d'-d)h$
as analogous to Hilbert's resolvent identity.
\step
{\bf Differential graded Lie algebra and MC-elements.}
A \emph{graded Lie algebra}
or gLa is a $\Z$-graded vector space $\gx$
with a bracket $[-,-] \in \Hom^0(\gx \otimes \gx,\gx)$
that respects the grading,
is graded antisymmetric,
and satisfies the graded Jacobi identity. Explicitly for all homogeneous elements,
the degree of $[x,y]$ is the sum of the degrees of $x$ and $y$ and
$[x,y] = -(-1)^{xy}[y,x]$ and
\[
[x,[y,z]]
+ (-1)^{x(y+z)}[y,[z,x]]
+ (-1)^{z(x+y)}[z,[x,y]] = 0
\]
A \emph{differential graded Lie algebra}
or dgLa $\gx$ is a gLa 
with $d \in \End^1(\gx)$ a differential, $d^2 = 0$,
compatible with the bracket in the sense of the Leibniz rule
$d[x,y] = [dx,y] + (-1)^x[x,dy]$.
The Maurer-Cartan set is
\[
  \MC(\gx) = \{u \in \gx^1 \mid du + \tfrac{1}{2}[u,u] = 0\}
\]
Formally, the Lie algebra $\gx^0$ acts on this set,
and $\MC(\gx)/{\sim}$
is the moduli space of interest,
a rigorous variant of which is the deformation functor \cite{kontsevich}.
%-----------------------------------------------------
\step
{\bf The Lie algebra of the Lorentz group and its representations.}
  Given a 4-dimensional complex vector space
  with a nondegenerate quadratic form $\Q$,
  the automorphism Lie algebra is $\sl \oplus \sl$
  with $\sl$ the complex Lie algebra of traceless $2 \times 2$ matrices.
  If the vector space is that of $2 \times 2$ complex matrices\footnote{%
  A minor clash of notations,
  the letter $d$ is also used for differentials. 
  }
  \[
    k = \begin{pmatrix}
      a & b \\ c & d
    \end{pmatrix}
  \]
  and $\Q = \det k = ad-bc$,
  then left-multiplication and right-multiplication by matrices with
  determinant one yield all automorphisms.
  Define right-multiplication with a transpose to get a left-action.
  At the Lie algebra level, $\sl \oplus \sl$.
  The finite-dimensional irreducible representations are
  \[
    \rep{p}{q} = S^{2p}\rep{\tfrac12}{0} \otimes S^{2q}\rep{0}{\tfrac12}
  \]
  where $p,q \geq 0$ are half-integers, $\rep{\tfrac12}{0} \simeq \C^2$
  and $\rep{0}{\tfrac12} \simeq \C^2$ are the fundamental representations of left and right $\sl$ respectively,
  and $S$ is the symmetric tensor product.
  So $\dim \rep{p}{q} = (2p+1)(2q+1)$.
  As $\sl \oplus \sl$ representations,
  \begin{equation}\label{eq:cg}
    \rep{p}{q} \otimes \rep{p'}{q'}
        \;\simeq\;
        \textstyle\bigoplus_{p'' = |p-p'|}^{p+p'}
        \bigoplus_{q''=|q-q'|}^{q+q'} \rep{p''}{q''}
  \end{equation}
  where $p''$ and $q''$ increase in steps of one.
  The Lie algebra $\sl \oplus \sl$
  is 
  the complexification
  of the real Lie algebra of the Lorentz group,
  as follows.
  On the real subspace of Hermitian $2 \times 2$ matrices,
  $\Q$ is a real quadratic form of signature
  ${+}{-}{-}{-}$
  whose automorphism Lie algebra
  is the real subalgebra of $\sl \oplus \sl$ of elements of the form
  $A \oplus \overline{A}$ where a bar means element-wise conjugation.
  
  %--------------------------------------------------------
\step
{\bf The Lorentz equivariant complexes $\cc_{\pm h}$.}
For every half-integer\footnote{%
  A minor clash of notations,
  the letter $h$ is also used for homotopies. 
  }
 $h \geq \tfrac{1}{2}$ called
helicity and for every
momentum $k \in \rep{\tfrac{1}{2}}{\tfrac{1}{2}} \simeq \C^4$
define complexes
\begin{equation}
  \label{eq:lec}
  \begin{aligned}
    \cc_{h} : & \qquad
    0 \to \rep{h}{0}
      \to \rep{h-\tfrac{1}{2}}{\tfrac{1}{2}}
      \to \rep{h-1}{0}
      \to 0\\
    \cc_{-h} : & \qquad
    0 \to \rep{0}{h}
      \to \rep{\tfrac{1}{2}}{h-\tfrac{1}{2}}
      \to \rep{0}{h-1}
      \to 0
  \end{aligned}
\end{equation}
  where, by definition,
  the three terms are in homological degrees $1$, $2$, $3$ respectively
  and where the last term is dropped when $h = \frac12$.
  The dependence on $k$ is implicit in the differential.
  By definition, the differential is
  linear in $k \in \rep{\tfrac{1}{2}}{\tfrac{1}{2}}$
  and it is the unique $\sl \oplus \sl$ equivariant map
  \[
    \rep{\tfrac{1}{2}}{\tfrac{1}{2}} \otimes \rep{h}{0}
      \to \rep{h-\tfrac{1}{2}}{\tfrac{1}{2}}
      \qquad
      \rep{\tfrac{1}{2}}{\tfrac{1}{2}} \otimes
      \rep{h-\tfrac{1}{2}}{\tfrac{1}{2}}
      \to \rep{h-1}{0}      
  \]
  for $\cc_{h}$, analogous for $\cc_{-h}$.
  The uniqueness is by \eqref{eq:cg} and is up
  to an irrelevant multiplicative constant.
  This is a differential because its square
  is an equivariant map
  $S^2\rep{\tfrac{1}{2}}{\tfrac{1}{2}} \otimes \rep{h}{0} \to \rep{h-1}{0}$
  that vanishes by
  $S^2\rep{\tfrac{1}{2}}{\tfrac{1}{2}} \simeq \rep{0}{0} \oplus \rep{1}{1}$
  and \eqref{eq:cg}.
  Explicitly,
  for $\cc_{\pm h}$ the first part of the differential is given by 
  \begin{subequations}\label{eq:diffk}
  \begin{equation}
     S^{2h}\C^2 \spl{}
        S^{2h-1}\C^2 \otimes \C^2
        \xrightarrow{\one \otimes k^\mp \eps}
        S^{2h-1}\C^2 \otimes \C^2
  \end{equation}
  and the second is given by
  \begin{multline}
        S^{2h-1}\C^2 \otimes \C^2
        \spl{\otimes\one}
        (S^{2h-2}\C^2 \otimes \C^2) \otimes \C^2
        \simeq
        S^{2h-2}\C^2 \otimes (\C^2 \otimes \C^2)\\
        \xrightarrow{\one \otimes (\eps k^{\pm} \eps)}
        S^{2h-2}\C^2
  \end{multline}
  \end{subequations}
  where
  $k = (\begin{smallmatrix} a & b \\ c & d \end{smallmatrix})$
  and {\recalleps} and $k^+ = k$ and $k^- = k^T$,
  splitting is defined by
  $z^{\otimes p} \mapsto z^{\otimes p-1} \otimes z$ for all $z \in \C^2$,
  and $\eps k^{\pm} \eps : \C^2 \otimes \C^2 \to \C$,
   $x \otimes y \mapsto x^T\eps k^\pm \eps y$\footnote{%
  The composition of the maps \eqref{eq:diffk} is zero by construction.
  To check it directly, it suffices to show that it annihilates
  all elements of the form $z^{\otimes 2h}$ with $z \in \C^2$ since they span
  $S^{2h}\C^2$.
  One obtains $(z^T \eps k^\pm \eps k^\mp \eps z) z^{\otimes (2h-2)}$
  which is zero since $\eps k^\pm \eps k^\mp \eps$ is antisymmetric.
  }.
  More explicitly still,
  there are bases for which the differential for $\cc_2$ is
  given by
  \begin{equation}\label{eq:ccm}
   \begin{pmatrix}
    a & c & 0 & 0 & 0 \\
    b & d & 0 & 0 & 0 \\
    0 & a & c & 0 & 0 \\
    0 & b & d & 0 & 0 \\
    0 & 0 & a & c & 0 \\
    0 & 0 & b & d & 0 \\
    0 & 0 & 0 & a & c \\
    0 & 0 & 0 & b & d
  \end{pmatrix}
  \qquad
  \text{and}
  \qquad
  \begin{pmatrix}
    b & -a & d & -c & 0 & 0 & 0 & 0\\
    0 & 0 & b & -a & d & -c & 0 & 0\\
    0 & 0 & 0 & 0 & b & -a & d & -c
  \end{pmatrix}
\end{equation}
  Switch $b$, $c$ for $\cc_{-2}$.
  Similar for $\cc_{\pm h}$.
  The homologies are in Lemmas \ref{lemma:sph}, \ref{lemma:sph2}.

%% file: MinimalModel.tex
\section{The $L_\infty$ minimal model is gauge independent} \label{sec:mm}

Homotopy transfer refers generally to the transfer of certain algebraic structures
through quasi-isomorphisms, see e.g.~\cite{lv}.
In the special case of a dgLa $\gx$
and a contraction to the homology $\hx$,
one obtains an $L_\infty$ algebra structure on $\hx$
called the $L_\infty$ minimal model,
unique up to $L_\infty$ isomorphisms.
There are explicit formulas for the $L_\infty$ minimal model brackets
as a sum of trees, see below.
We refer to the literature for proofs
that this actually defines the $L_\infty$ minimal model,
including formulas for suitable $L_\infty$ quasi-isomorphisms
 \cite{kontsevich,huebschmannstasheff,berglund,ksoibelman,lv}.
\step
We show that for a dgLa with a momentum grading,
and homotopies that respect the momentum grading,
the $L_\infty$ minimal model brackets defined using a tree expansion are independent of the homotopy
in the strong sense (not merely up to $L_\infty$ isomorphism)
when all internal lines are off-shell.
Namely, individual trees may depend on the homotopy,
but the sum of all trees does not.
The results are for all homological degrees,
so amplitudes are a special case.

\begin{definition}[Momentum grading] \label{def:momgr}
  Suppose $K$ is an Abelian group that we call momentum space.
  By a dgLa $\gx = \bigoplus_{i\in \Z} \gx^i$ with momentum grading we mean one that
  carries a compatible $K$-grading, with algebraic direct sum
  \[
        \gx = \textstyle\bigoplus_{k \in K} \gx_k
  \]
  Compatibility means that the $K$-grading
  respects the $\Z$-grading and that
  \[
    d \gx_k \subset \gx_k
    \qquad\qquad
    [\gx_{k_1},\gx_{k_2}] \subset \gx_{k_1+k_2}
  \]
  Then the homology of $\gx$ also decomposes,
  $\hx = \bigoplus_{k \in K} \hx_k$
  where $\hx_k$ is the homology of $\gx_k$.
  A momentum $k$ will be called on-shell if $\hx_k \neq 0$, off-shell if $\hx_k = 0$.
\end{definition}
%--------------------------------------------------------------
We define the minimal model brackets as a sum of trivalent trees.
Let $T_n$ be the set of tree graphs
  	with $n+1$ labeled leaves,
        $n-2$ unlabeled internal lines,
        $n-1$ internal nodes of degree 3 (known as trivalent or cubic).
The leaves $1,\ldots,n$ are called inputs, and $n+1$ the output.
Let $P_n$ be such trees with a planar embedding.
The map $P_n \to T_n$ that forgets the embedding
is surjective\footnote{%
Each fiber of $P_n \to T_n$ has $2^{n-1}$ elements,
and
 $|T_n| = (2n-3)!!$.}.

\begin{definition}[Trees] \label{def:matet}
For a dgLa $\gx$ and a homotopy $h$ that satisfies $dhd = d$,
  let $p: \gx \to \hx$
  and $i: \hx \to \gx$ be the induced contraction.
  For every $P \in P_n$ define $m_{P,h} \in \Hom^{2-n}(\hx^{\otimes n},\hx)$ as follows:
\begin{itemize}
\item Decorate each input leaf by $i$, the output leaf by $p$.
\item Decorate each internal line by $h$.
\item Decorate each node by $\llbracket-,-\rrbracket$.
Here $\llbracket x,y\rrbracket = (-1)^x [x,y]$ for all $x,y \in \gx$.
\item
Given $x_1 \otimes \cdots \otimes x_n \in \hx^{\otimes n}$
one inserts each $x_j$ at the input labeled $j$.
\item
Multiply by the sign needed to permute
$x_1,\ldots,x_n$ into place,
where an even (odd) $x_j$ is considered odd (even)
for the purpose of this permutation\footnote{%
Concretely, the permutation is given by reading off the input
labels of $P$ counter-clockwise,
starting just to the left of the output.
The sign is equal to the ordinary permutation signature for 
permuting only the $x_j$ with even degree.}.
\item Multiply by the sign
$(-1)^{x_{n-1} + x_{n-3} + x_{n-5} + \ldots}$\footnote{%
This sign, independent of $P$,
contributes to the simple Koszul sign rule in Definition \ref{def:mmb}.
}.
\end{itemize}
Then $m_{P,h}$ is independent of the planar embedding\footnote{%
To see this, note that the building block
$h\llbracket-,-\rrbracket$
is homogeneous of degree zero and graded symmetric
as a map $\gx[1]^{\otimes 2} \to \gx[1]$.
Here $\gx[1]$ is obtained from $\gx$ by shifting the degree by one.
}.
So for every $T \in T_n$ we can set $m_{T,h} = m_{P,h}$
where $P \in P_n$ is any planar embedding of $T$.
\end{definition}
\begin{example}
The set $P_n$ is in bijection with full parenthesizations of any permutation of 
the elements $1,\ldots,n$. With this understanding,
\begin{align*}
m_{(12),h}(x_1,x_2) & = p[ix_1,ix_2]\\
m_{((((12)3)4)5),h}(x_1,\ldots,x_5) & = 
   p[h[h[h[ix_1,ix_2],ix_3],ix_4],ix_5]\\
m_{(((12)3)(45)),h}(x_1,\ldots,x_5) & = 
(-1)^{1+x_1+x_2+x_3} p[h[h[ix_1,ix_2],ix_3],h[ix_4,ix_5]]
\end{align*}
Note that if all inputs have odd degree,
one always gets a plus sign.
\end{example}
%-------------------------------------------------------------------------------
\begin{definition}[The minimal model brackets associated to a homotopy] \label{def:mmb}
  For a dgLa $\gx$ and a homotopy $h$ as above,
   the $n$-slot minimal model bracket
  \[
    \{-,\ldots,-\}_h = \sum_{T \in T_n} m_{T,h}
  \]
  is in $\Hom^{2-n}(\hx^{\otimes n},\hx)$.
  It satisfies $\{\ldots,y,x,\ldots\}_h = -(-1)^{xy} \{\ldots,x,y,\ldots\}_h$.
  In this section, $\gx$ has a momentum grading
  and we assume $h\gx_k \subset \gx_k$ for all $k$.
\end{definition}
%---------------------------------------------------------------

One may want to take a more cavalier attitude towards signs.
By contrast, momentum conservation at each node is essential.
Pictorially, a bracket such as
 $\hx_{k_1} \otimes \cdots \otimes \hx_{k_5}
\to \hx_{k_1 + \ldots + k_5}$ uses
trees like the following where $h_k = h|_{\gx_k}$.
\[
\graphtech
\]
\begin{warning}[Discontinuous nature of $h$]
\label{warning:discontinuous}
The space $K$ has no topology
and no continuity in $k$ is assumed.
But even if $K = \C^4$ as in other sections,
even if the $\gx_k$ are finite-dimensional and isomorphic
and $d_k = d|_{\gx_k}$ is polynomial in $k$,
the homotopy $h_k$ must still be discontinuous because the projection
$i_kp_k = \one - d_kh_k - h_kd_k$
is zero off-shell but nonzero on-shell.
When we apply this section,
$h$ will be separately defined off-shell
and on-shell, see Section \ref{sec:pre}.
\end{warning}
%------------------------------------------------------------
\begin{theorem}[Gauge independence I] \label{theorem:well-defined}
  The bracket $\{-,\ldots,-\}_h$ is independent of the
  momentum-conserving homotopy $h$ when evaluated on
  \[
        \hx_{k_1} \otimes \cdots \otimes \hx_{k_n} \to \hx_{k_1 + \ldots + k_n}
  \]
  with all internal lines off-shell, meaning for all
  $(k_1,\ldots,k_n) \in K^n$
  such that $\hx_{k_J}=0$ for all subsets $J \subset \{1,\ldots,n\}$
  with $1 < |J| < n$ and $k_J = \sum_{i \in J} k_i$.
\end{theorem}
\begin{proof}
  By Theorem \ref{well-defined2},
  $M_h : \gx^{\otimes n} \to \gx$ is a chain map
  so $dM_h = M_h\dtot$ since $d\pi = \pi d = 0$,
  it induces the minimal model bracket $\hx^{\otimes n} \to \hx$ on homology,
  and $M_h$ and $M_{h'}$ are homotopy equivalent
  when internal lines off-shell.
  Therefore they induce the same map on homology.
\qed\end{proof}
\begin{theorem}[Gauge independence II] \label{well-defined2}
  Let $M_h \in \Hom^{2-n}(\gx^{\otimes n},\gx)$ be defined
  like the minimal model brackets but with
  input and output leaves decorated by
  $\pi = \one - dh - hd$, replacing $i$ and $p$.
  Let $\dtot$ be the differential on $\gx^{\otimes n}$, so
  \[
    \dtot = \textstyle\sum_{i=1}^n
    (\pm \one)^{\otimes (i-1)} \otimes d \otimes \one^{\otimes (n-i)}
  \]
  where $\pm \one$ is the sign map.
  Then for all homotopies $h$ and $h'$
  there exists a linear map $E: \gx^{\otimes n} \to \gx$
  (that can be given as a sum of trees
  built using only $h$ and $h'$ and $d$ and the bracket) such that
  \[
    M_{h'} - M_h \qquad \text{equals} \qquad dE + E\dtot
  \]
  when evaluated on
  $\gx_{k_1} \otimes \cdots \otimes \gx_{k_n} \to \gx_{k_1 + \ldots + k_n}$
  with all internal lines off-shell,
  where off-shell means homology-free just like in Theorem 
  \ref{theorem:well-defined}.
\end{theorem}
%----------------------------------------------------------------------------
The following lemma,
which is for any complex of vector spaces,
connects any two homotopies $h$ and $h'$ by three transformations.
We will use it to connect the homotopies by three curves
polynomial in a parameter $s$. Pictorially,
\newcommand{\dsmooth}[3]{\draw [black] plot [smooth, tension=1]
    coordinates {(#1) (#2) (#3)};}
\newcommand{\zpt}[2]{\draw [black,fill=black]
    (#1) circle (0.04) node [anchor=south,xshift=0,yshift=3] {#2};}
\begin{center}
  \begin{tikzpicture}[scale=2]
  \coordinate (a) at (0,0);
  \coordinate (ab) at (0.5,-0.4);
  \coordinate (b) at (1,0);
  \coordinate (bc) at (1.5,-0.4);
  \coordinate (c) at (2,0);
  \coordinate (cd) at (2.5,-0.4);
  \coordinate (d) at (3,0);
  \dsmooth{a}{ab}{b}
  \dsmooth{b}{bc}{c}
  \dsmooth{c}{cd}{d}
  \zpt{a}{$h$};
  \zpt{b}{$h_{\fixY}$};
  \zpt{c}{$h_{\fixX}$};
  \zpt{d}{$h_{\fixXY} = h'$};
\end{tikzpicture}
\end{center}
%-----------
\begin{lemma}[The ABC lemma]\label{lemma:abc}
Given is a complex with differential $d$.
In this lemma we only consider homotopies $h$ that satisfy $dhd=d$,
and we denote $\pi = \one - dh - hd$.
If $h$ is a homotopy then another homotopy $h'$ is given by
\[
\begingroup
\renewcommand{\arraystretch}{1.15}
\begin{array}{c|c|c|c}
& h' & \pi' & \text{constraints}\\
\hline
\ff{\fixY} & h(\one-\fixy \pi) & (\one+dh\fixy )\pi & d\fixy  = \pi \fixy  = \fixy d = \fixy h = 0\\
\ff{\fixX} & (\one+d\fixx d)h(\one-d\fixx d) & \pi & \pi \fixx  = h\fixx  = \fixx h = \fixx \pi = 0\\
\ff{\fixXY} & (\one-\pi \fixxy )h & \pi (\one+\fixxy hd) & h\fixxy  = d\fixxy  = \fixxy \pi = \fixxy d = 0
\end{array}
\endgroup
\]
for all $\fixy ,\fixxy  \in \End^0(\gx)$ or $\fixx  \in \End^{-2}(\gx)$ respectively,
subject to the constraints.
And any two homotopies $h$ and $h'$
are related by a composition of $\ff{\fixY}$, $\ff{\fixX}$, $\ff{\fixXY}$.
Namely,
if one sets $h_\fixY = \ff{\fixY}h$
and $h_\fixX = \ff{\fixX}h_\fixY$ and  $h_\fixXY = \ff{\fixXY}h_\fixX$
with
\begin{equation}\label{eq:abc}
\fixy  = -dh'\pi
\qquad
\fixx  = -h_\fixY dh'h_\fixY
\qquad
\fixxy  = -\pi_\fixX h'd
\end{equation}
then $h_\fixXY = h'$.
\end{lemma}
\begin{proof}
First the transformations in the table.
For $\ff{\fixXY}$
we have $(h')^2 = 0$
using $h^2 = 0$, $h\pi = 0$,
we have $h'dh' = h'$
using $hdh = h$, $d\pi = 0$,
we have $dh'd = d$
using $dhd = d$, $d\pi = 0$,
and $\pi' = \one - h'd-dh'$.
Similar for $\ff{\fixY}$, $\ff{\fixX}$.
Before proving the second part of the theorem,
we derive another characterization of $\ff{\fixY}$, $\ff{\fixX}$, $\ff{\fixXY}$.

Note that 
$\ff{\fixY}$ implies $hd = h'd$,
$\ff{\fixX}$ implies $\pi = \pi'$,
$\ff{\fixXY}$ implies $dh = dh'$.
Conversely,
for all $h$ and $h'$,
if $hd=h'd$ then they are related by $\ff{\fixY}$ using $\fixy  = -d h'\pi$,
if $\pi = \pi'$ then by $\ff{\fixX}$ using $\fixx  = -hdh'h$,
if $dh=dh'$ then by $\ff{\fixXY}$ using $\fixxy  = -\pi h'd$.
Say in the case $\ff{\fixY}$, the given $\fixy $ has degree zero,
satisfies the constraints in the table, and
$h(\one-\fixy \pi) = h(\one + dh'\pi)
= h + hdh'\pi
= h + h'dh'\pi = h + h'\pi
= h + h'(\one-dh-hd) = h'$
as required using $h'hd = h'h'd = 0$
and $h'dh = hdh = h$. 
Similar for $\ff{\fixX}$, $\ff{\fixXY}$.
The constraints make
$\fixy,\fixx,\fixxy$ unique.

Though one can proceed at the level of equations,
we switch to a geometric argument. 
Recall the bijection between
homotopies $h$ with $dhd=d$ and pairs $(X,Y)$ of graded subspaces
where $X$ is a complement of $\image d$ in $\ker d$,
$Y$ a complement of $\ker d$.
The bijection is established by
$X = \image \pi$, $Y = \image hd$.
The last paragraph shows that
$\ff{\fixY}$ connects any two homotopies with the same $Y$, 
$\ff{\fixX}$ those with the same $X$ and $Y \oplus \image d$,
$\ff{\fixXY}$ those with the same $X$ and $Y\oplus X$.

Now, given any two homotopies $h$, $h'$
corresponding to $(X,Y)$, $(X',Y')$ respectively,
define homotopies $h_\fixY$, $h_\fixX$, $h_\fixXY$ by
$X_\fixY = X_\fixX = X_\fixXY = X'$,
$Y_\fixY = Y$, $Y_\fixXY = Y'$,
and $Y_\fixX $ is given by $Y\oplus\image d = Y_\fixX \oplus\image d$ and
$Y_\fixX\oplus X'=Y'\oplus X'$.
There exists a unique such $Y_\fixX$ since $Y$ and $Y'$ are complements of
$\ker d = \image d\oplus X'$. 
By the last paragraph we have 
$h_\fixY = \ff{\fixY}h$,
$h_\fixX = \ff{\fixX}h_\fixY$,
$h_\fixXY = \ff{\fixXY}h_\fixX$
for some $\ff{\fixY}$, $\ff{\fixX}$, $\ff{\fixXY}$.
One can see that they are given by \eqref{eq:abc}.
By construction, $h_\fixXY=h'$.
\qed	
\end{proof}
%%%%%%%%%%%%%%%%%%%%%%%%%%%%%%%%%%%%%%%%%%%%%%%%%%%%%%%%%%%%%%%%%%%%%%%%%%%%%%%%%%%%%%%%%%%%%
%On[Assert];
%n=RandomInteger[{2,5}];
%m=RandomInteger[{2,5}];
%d=ConstantArray[0,{n+n+m,n+n+m}];
%d[[1;;n,n+1;;2n]]=IdentityMatrix[n];
%hX=Transpose[d];
%id=IdentityMatrix[n+n+m];
%zero=ConstantArray[0,{n+n+m,n+n+m}];
%R:=Module[{r,Q},
%Q=Array[r,{n+n+m,n+n+m}];
%Q=(Q/.First[Quiet[Solve[Flatten[Q.d==d.Q],Flatten[Q]]],Solve::svars])/.MapThread[Rule,{Flatten[Q],RandomInteger[{-10,10},(n+n+m)^2]}];
%If[Det[Q]===0,R,Q]];
%randomH:=Module[{r=R},r.hX.Inverse[r]];
%\[CapitalPi][h_]:=id-d.h-h.d;
%check[h_]:=Assert[d.h.d===d&&h.h===zero&&h.d.h===h];
%A[h_,a_]:=With[{hp=h.(id-a.\[CapitalPi][h])},check[h];Assert[\[CapitalPi][hp]===(id+d.h.a).\[CapitalPi][h]];hp];
%B[h_,b_]:=With[{hp=(id+d.b.d).h.(id-d.b.d)},check[h];Assert[\[CapitalPi][hp]===\[CapitalPi][h]];hp];
%CC[h_,c_]:=With[{hp=(id-\[CapitalPi][h].c).h},check[h];Assert[\[CapitalPi][hp]===\[CapitalPi][h].(id+c.h.d)];hp];
%T[h_,hp_]:=Module[{hA,hB,hC,a,b,c},
%a=-d.hp.\[CapitalPi][h];
%hA=A[h,a];
%b=-hA.d.hp.hA;
%hB=B[hA,b];
%c=-\[CapitalPi][hB].hp.d;
%hC=CC[hB,c];
%hC==hp]
%Table[T[randomH,randomH],{3}]//DeleteDuplicates
%%%%%%%%%%%%%%%%%%%%%%%%%%%%%%%%%%%%%%%%%%%%%%%%%%%%%%%%%%%%%%%%%%%%%%%%%%%%%%%%%%%%%%%%%%%%%%%%%%%
\begin{proof}[of Theorem \ref{well-defined2}]
It suffices to prove the theorem 
in the special cases $\ff{\fixY}$, $\ff{\fixX}$, $\ff{\fixXY}$ of Lemma \ref{lemma:abc}.
Since $h\gx_k \subset \gx_k$, \eqref{eq:abc} implies
$\fixy \gx_k, \fixx \gx_k, \fixxy \gx_k\subset \gx_k$.
The brackets are polynomial in $h$ and $\pi$,
so if we consider polynomial curves
$\fixy(s)$,
$\fixx(s)$, 
$\fixxy(s)$ then the brackets are polynomial in $s$.
It suffices to show that we get the desired result when differentiating 
with respect to $s$,
namely that $\dot{M} = d\dot{E} + \dot{E}\dtot$ for some $\dot{E}$
where a dot denotes a derivative at $s=0$\footnote{%
Setting $s=0$ is no loss of generality since
$\ff{\fixY}$, $\ff{\fixX}$, $\ff{\fixXY}$ satisfy group laws
in $\fixy$, $\fixx$, $\fixxy$
provided one makes minor modifications
for consistency with the constraints.
}:
\begin{itemize}
  \item $\ff{\fixY}$. Here $\dot{h} = -h\dot{\fixy }\pi$
    and $\dot{\pi} = dh\dot{\fixy } \pi$.
    Since all internal lines are off-shell hence annihilated by $\pi$,
    we effectively have $\dot{h} = 0$.
    So only inputs and outputs are affected.
    One can take $\dot{E} = h\dot{\fixy } M$ in this case.
 \item $\ff{\fixX}$. Here $\dot{h} = d\dot{\fixx }dh-hd\dot{\fixx }d
   = d\dot{\fixx } - \dot{\fixx }d$ and $\dot{\pi} = 0$.
   In particular, only internal lines are affected.
   Here $\dot{E} = 0$.
  \item $\ff{\fixXY}$. Here $\dot{h} = -\pi \dot{\fixxy } h$
    and $\dot{\pi} = \pi \dot{\fixxy } hd$.
    As in $\ff{\fixY}$,
    we effectively have $\dot{h} = 0$.
    Here $\dot{E} = M \dot{\fixXY}$ where $\dot{\fixXY}$
    is a graded symmetrization of $\dot{\fixxy } h \otimes \one^{\otimes n-1}$.
\end{itemize}
With this setup, it suffices to show 
separately for every nonempty $J \subset \{1,\ldots,n\}$
that the following infinitesimal variations,
affecting internal lines with momentum $k_J = \sum_{i \in J} k_i$,
yield zero after summation over all trees:
\begin{itemize}
  \item Variations of type $d\dot{\fixx }$ at an internal line if $1 < |J| < n$ 
	respectively $dh\dot{\fixy }\pi$ at the input if $|J| = 1$.
	Note the $d$ on the left.
  \item Variations of type $-\dot{\fixx }d$ at an internal line if $1 < |J| < n$,
    respectively $\pi \dot{\fixxy } hd$ at the output if $|J| = n$.
    Note the $d$ on the right.
\end{itemize}
All cases reduce to Lemma \ref{lemma:noobs}.
Given the results for $\dot{E}$,
one can see that $E$ has the claimed form.
For $\ff{\fixY}$ take $E=haM$,
for $\ff{\fixX}$ take $E=0$,
similar for $\ff{\fixXY}$.
\qed\end{proof}
%---------------------------------------------------------------------
\begin{lemma}[A cancellation] \label{lemma:noobs}
The map $\gx_{k_1} \otimes \cdots \otimes \gx_{k_n} \to \gx_{k_1 + \ldots + k_n}$
defined just like $M$ but with the output leaf decorated by $Nd$ (rather than $\pi$)
is identically zero
when all internal lines are off-shell.
Likewise if one input leaf is decorated by $dN$.
Here $N$ is any momentum conserving operator of degree $-1$
to guarantee that $Nd$ respectively $dN$ have degree zero\footnote{%
This implies that the fact
that $m_{T,h} = m_{P,h}$ is independent of the planar embedding
also holds with the decorations $Nd$ respectively $dN$.
Hence the map in this lemma is well-defined.}.
\end{lemma}
%-----------------------------------------------------------------------------
\begin{proof} 
When $d$ is the output, and all inputs are odd, then
this is the lemma in dgLa-based deformation theory
that says that \emph{obstructions are cocycles} \cite{gerstenhaber}.
In general, when the output is decorated by $Nd$, the proof is by
repeatedly moving occurrences of $d$ \emph{down the trees} using this algorithm:
\begin{equation}\label{eq:alg}
    \left.
    \begin{aligned}
      & \text{If $d$ hits a bracket as in $d[-,-]$, use the Leibniz rule.}\\[4pt]
      & \text{If $d$ hits an $h$,
        replace $dh = - hd - \pi + \one$.}\\[-2pt]
      & \text{Terms from $\pi$ vanish since $\pi$ annihilates off-shell momenta.}\\[-2pt]
      & \text{Terms from $\one$ are put in a basket, to be dealt with later.}\\[4pt]
      & \text{If $d$ hits an input $\pi$, use $d\pi = 0$.}
    \end{aligned}\;\;\;\right\}
\end{equation}
We must show that the terms in the basket add to zero.
These terms are one-to-one with $T_n'$,
the set of trees like $T_n$ but with a distinguished internal line.
The distinguished line is the one decorated by $\one$,
corresponding to a direct nesting of two brackets as in $[[-,-],-]$.
Define an equivalence relation on $T_n'$
that identifies trees that differ
only by a permutation of the four lines adjacent to the distinguished line.
Each equivalence class has three elements as in
\newcommand{\xblob}[2]{%
  \draw [black] (#1) node {#2};
}
\newcommand{\xjacobi}[3]{\wrap{1}{
  \coordinate (b) at (-0.5,0.5);
  \coordinate (c) at (-0.5,-0.5);
  \coordinate (d) at (1.5,-0.5);
  \coordinate (a) at (1.5,0.5);
  \coordinate (x) at (0,0);
  \coordinate (y) at (1,0);
  \haux{y}{a}{shorten >=9}{}{}
  \haux{x}{b}{shorten >=9}{}{}
  \haux{x}{c}{shorten >=9}{}{}
  \haux{y}{d}{shorten >=9}{}{}
  \haux{x}{y}{}{anchor=south}{$\one$}
  \xblob{a}{$t_R$}
  \xblob{b}{$#1$}
  \xblob{c}{$#2$}
  \xblob{d}{$#3$}
  \br{x}{}{}
  \br{y}{}{}
}}
\[
\xjacobi{t_A}{t_B}{t_C}
\qquad\qquad
\xjacobi{t_B}{t_C}{t_A}
\qquad\qquad
\xjacobi{t_C}{t_A}{t_B}
\]
where $t_R,t_A,t_B,t_C$ are subtrees
and we agree that
a planar embedding is chosen for each
inducing an embedding for the three terms,
and the output leaf is in $t_R$.
These three terms are 
$\sigma \llbracket \llbracket A,B \rrbracket,C\rrbracket$
and $\sigma' \llbracket \llbracket B,C\rrbracket ,A\rrbracket$
and $\sigma''\llbracket \llbracket C,A \rrbracket,B \rrbracket$,
inserted into the operator given by $t_R$.
Their relative signs are
\[
\sigma : \sigma' : \sigma''
= 1 :(-1)^{(A+1)(B+C)} : (-1)^{(C+1)(A+B)}
\]
and so they add to zero by the definition of $\llbracket-,-\rrbracket$
and the Jacobi identity.
The relative signs are due to 
the permutation sign in Definition \ref{def:matet},
since $A+1$, $B+1$, $C+1$ are equal
to the number of even elements entering $t_A$, $t_B$, $t_C$ respectively, mod 2.
No relative sign was produced by algorithm \eqref{eq:alg},
in particular the last application of the Leibniz rule is without sign
since $\one$ is left-adjacent to $t_R$. 

Analogous if $dN$ decorates an input leaf,
in this case one repeatedly moves occurrences of $d$
\emph{away from that input}, a modification of \eqref{eq:alg}.
Suppose that input is in $t_A$.
We get zero again since the relative signs are
\[
\sigma : \sigma' : \sigma''
= \underline{(-1)} : \underline{(-1)^{B+C+1}} (-1)^{A(B+C)} : \underline{(-1)^C} (-1)^{(C+1)(A+B+1)}
\]
because the number of even elements entering $t_A$ is now $A$ mod 2
due to the presence of the operator $N$,
and the underlined relative signs are introduced by the final application of the Leibniz rule.
\qed\end{proof}

%% file: YangMillsGeneralRelativity.tex
\section{YM and GR admit a homological formulation} \label{sec:ymgr}

We define two dgLa whose
associated Maurer-Cartan equations are the ordinary nonlinear classical field
equations of YM respectively GR about Minkowski spacetime $\R^4$.
In particular, for GR the solutions are Ricci-flat metrics.
\step
Logically, this section is organized around an existence theorem
that says that for $h=1,2$
there exist dgLa with nontrivial bracket
whose homology is globally isomorphic for $k \in \C^4-0$
to that of the complexes $\cc_{-h} \oplus \cc_{+h}$\footnote{%
Tensored with the internal Lie algebra in the YM case $h=1$.},
plus some Lorentz invariance and homogeneity.
The point is that the amplitudes\footnote{%
Which we recall are the minimal model brackets with all internal lines off-shell.}
are the same for all such dgLa,
by the unique recursive characterization in Section \ref{sec:recchar}.
That is,
the amplitudes will be seen to be insensitive to properties of the dgLa
other than those listed in the existence theorem below.
\step
This section is in part based on
\cite{zeitlin,costello} for YM and \cite{rt} for GR.
%%%%%%%%%%%%%%%%%%%%%%%%%%%%%%%%%%%%%%%%%%%%%%%%%%%%%%%%%%%%%%%%%%%%%%%%
\begin{theorem}[%
Existence of a dgLa for YM and GR about Minkowski] \label{theorem:ymgrnew}
Let $h=1$ or $h=2$, that we refer to as YM and GR respectively.
For $h=1$ suppose we are given a finite-dimensional non-Abelian Lie algebra $\ym$.
Then there exists a dgLa $\gx$ with the following properties:
\begin{itemize}
\item \emph{Momentum grading:}
  It has a $\C^4$ momentum grading as in 
  Definition \ref{def:momgr},
  \[
    \gx = \textstyle\bigoplus_{k \in \C^4} \gx_k
  \]
There is a graded vector space $\vecx$ of finite dimension $d_\vecx$
and isomorphisms $\gx_k \simeq \vecx$
such that the differential $\gx_k \to \gx_k$
and bracket
$\gx_{k_1} \otimes \gx_{k_2} \to \gx_{k_1+k_2}$
are given by arrays of size
$d_\vecx \times d_\vecx$
and
$d_\vecx \times d_\vecx \times d_\vecx$
with entries in $\C[k]$
and $\C[k_1,k_2]$ respectively\footnote{%
Of course, many entries vanish due to the $\Z$-grading.}.
\item \emph{Global structure of the homology:}
Recall the definition of the $\cc$-complexes in \eqref{eq:lec}.
There is a collection of isomorphisms, one for every $k \neq 0$,
between the homology of $\gx_k\simeq \vecx$ and the homology of
\begin{equation}\label{eq:hxx}
  \begin{aligned}
    \textnormal{YM\,:} & \qquad
    (\cc_{-1} \oplus \cc_1) \otimes \ym\\
    \textnormal{GR\,:} & \qquad
    \cc_{-2} \oplus \cc_2
  \end{aligned}
\end{equation}
regular in the sense that:%
{\addtocounter{footnote}{1}
 \footnotetext{%
   These homotopy equivalences on Zariski open sets
   need not coincide on overlaps,
   but they induce the same isomorphism on homology.
   The matrices are relative to a basis of $\vecx \simeq \gx_k$.
 }\addtocounter{footnote}{-1}}%
\begin{equation}\label{claim:qisos}
\parbox[b]{101mm}{Every $k \neq 0$
  has a Zariski open neighborhood on which
  this isomorphism is induced by a homotopy equivalence
  (Section \ref{sec:pre}) 
  given by four matrices whose entries
  are regular rational functions in $k$.\footnotemark{}}
\end{equation}
\item \emph{Homogeneity and Lorentz equivariance:}
Let $\hx_k=\hx_k^-\oplus \hx_k^+$ be the homology of
\eqref{eq:hxx}.
The minimal model bracket of $\gx$, viewed as a map
\[
\hx_{k_1}^1 \otimes \cdots \otimes \hx_{k_n}^1 \to \hx_{k_1 + \ldots + k_n}^2
\]
with $k_1,\ldots,k_n,k_1 + \ldots + k_n \neq 0$
and assuming all internal lines off-shell,
so that it is well-defined by Theorem \ref{theorem:well-defined},
satisfies:
\begin{itemize}
\item It is homogeneous of degree $3-2n$\footnote{%
Recall that the differential
on $\cc_{\pm h}$ is a matrix linear homogeneous in $k$,
so $\hx_{\lambda k} = \hx_k$ for all $\lambda \in \C^\times$.
With this understanding,
the bracket 
$\hx_{\lambda k_1} \otimes \cdots \otimes \hx_{\lambda k_n} \to \hx_{\lambda k_1 + \ldots + \lambda k_n}$
is equal to $\lambda^{3-2n}$ times the bracket
$\hx_{k_1} \otimes \cdots \otimes \hx_{k_n} \to \hx_{k_1 + \ldots + k_n}$.}.
\item For $n=2$ it is Lorentz invariant,
with $\ux$ the trivial representation.
For YM this $n=2$ bracket is proportional to the Lie bracket of $\ux$,
namely an antisymmetric map times the Lie bracket of $\ux$.
\item For $n=2$ and $\sigma = \pm$, the maps $\hx^{1,-\sigma} \otimes \hx^{1,\sigma}
\to \hx^{2,\sigma}$ do not identically vanish.
In particular the dgLa bracket is not trivial.
\end{itemize}
\end{itemize}
\end{theorem}
\newcommand{\ct}[1]{\textnormal{\footnotesize \##1}}
\begin{proof}
  The remainder of this section.
  We omit straightforward checks that the dgLa constructed below
  yield minimal model brackets that are homogeneous 
  and, for $n=2$, Lorentz invariant and nontrivial as claimed.
  For Lorentz invariance and homogeneity,
  note that the Lorentz group and $\C^\times$ act as dgLa automorphisms.
  A sloppy calculation of the degree of homogeneity is
  \[
   (-1) \cdot \ct{inputs} + 0 \cdot \ct{nodes}
   + (-1)\cdot \ct{internallines} + (+1) \cdot \ct{outputs} = 3-2n
  \]
  where inputs and output
  contribute the degree of homogeneity of the connecting morphisms
  for the short exact sequences
  \eqref{eq:ymses} respectively \eqref{eq:grses}.
\qed\end{proof}
%----------------------------------------------------------------

The elements of $\gx$ may be interpreted as finite linear combinations
of plane waves with complex momenta.
In fact, in this section we freely transition between
the momentum space dgLa $\gx$
and the position space or $C^\infty$-variant $\gx^\infty$:
\begin{quote}
    Namely $\gx^\infty = C^\infty \otimes \vecx$
    where $C^\infty = C^\infty(\R^4)$ are the smooth complex valued functions of $x \in \R^4$.
    This will be a dgLa whose differential and bracket are constant coefficient
    differential operators.
\end{quote}
One determines the other by requiring that
the map $\gx \to \gx^\infty$ given by
    \begin{equation}\label{eq:infalg}
      \gx_k \ni v \mapsto e^{ikx} \otimes v
    \end{equation}
be a dgLa morphism. Here $(x \mapsto e^{ikx}) \in C^\infty$ is a plane wave,
$kx$ the component-wise dot product.
Note that $\gx^\infty$ is a free $C^\infty$-module of rank equal to the dimension of $\vecx$,
but the differential and bracket will not be $C^\infty$-linear.
%----------------------------------------------
\begin{remark}[On constructing dgLa]
A \emph{dgca} is a differential graded commutative associative algebra over $\C$.
There are natural products $\dgca \otimes \dgca = \dgca$
and $\dgca \otimes \dgLa = \dgLa$
and if one forgets the algebra structure, they correspond to
the standard tensor product of complexes\footnote{%
Incidentally if $A$ is a {\dgca} then $A \otimes A$
can also be given the structure of a {\dgLa} by setting
$[u \otimes v, x \otimes y]
= (-1)^{y(v+x)} uy\otimes vx - (-1)^{u(v+x)} vx \otimes uy$
and $d(u \otimes v) = (du) \otimes v + (-1)^u u \otimes dv$.
If $\pi(u \otimes v) = (-1)^{uv} v \otimes u$
then $\pi[-,-] = [\pi-,\pi-]$, $\pi d = d\pi$,
so $\{u \mid \pi u = u\}$ is a sub-\dgLa.}.
\end{remark}

Let $\Omega$ be the
dgca of complex de Rham differential forms on $\R^4$
with the de Rham differential.
So $\Omega = C^\infty \otimes \wedge \C^4$
where the second factor $\wedge \C^4$ is the unital gca
freely generated in degree one by the symbols $dx^0,dx^1,dx^2,dx^3$.
As a representation of the Lorentz group,
 $\wedge \C^4 \simeq \wedge \rep{\frac12}{\frac12}$ is isomorphic to
\begin{equation}\label{eq:omrep}
0 \to \rep{0}{0} \to \rep{\tfrac12}{\tfrac12} \to
      \rep{1}{0} \oplus \rep{0}{1} \to \rep{\tfrac12}{\tfrac12} \to \rep{0}{0} \to 0
\end{equation}
where the arrows indicate the essentially
unique differential that depends linearly on $k \in \rep{\frac12}{\frac12}$
and is Lorentz equivariant,
which is nothing but the de Rham differential in momentum space. It is exact if $k \neq 0$.
Decompose $\Omega^2 = \Omega^2_+ \oplus \Omega^2_-$
where $\Omega^2_{\pm}$
is the $C^\infty$-submodule generated by all
$dx^0 dx^a \pm i dx^b dx^c$
with $a,b,c$ a cyclic permutation of $1,2,3$.
That is
\begin{equation}\label{eq:om2rep}
\Omega^2_+ \simeq C^\infty \otimes \rep{1}{0}
\qquad \qquad
\Omega^2_- \simeq C^\infty \otimes \rep{0}{1}
\end{equation}
Set $\Omega^{\leq 2}_{\pm} = \Omega^0 \oplus \Omega^1 \oplus \Omega^2_{\pm}$
and $\Omega_{\pm}^{\geq 2} = \Omega^2_{\pm} \oplus \Omega^3 \oplus \Omega^4$.
Note that $\Omega^{\geq 2}_{\pm}$
coincide with the complexes
$\cc_{\pm 1}$ up to a shift of the homological degree by one.
%---------------------------------------------------------------------------------
\begin{prop}[YM dgLa]
  Let $\C \oplus \C \eps$ be the dgca with 
  $\eps$ a symbol of degree $-1$,
  product given by $\eps^2 = 0$,
  and differential $z \oplus w\eps \mapsto w \oplus 0\eps$.
  Then the tensor product of dgca
  $(\C \oplus \C\eps) \otimes \Omega
  \simeq \Omega \oplus \eps\Omega$ has a dgca subquotient\footnote{%
    That is, the numerator is a sub-dgca, and the denominator
    is a dgca ideal in the numerator.}
  \begin{equation}\label{eq:ainf}
    \ax^\infty = \frac{\Omega \oplus \eps \Omega_+^{\geq 2}}{%
      \Omega_-^{\geq 2} \oplus \eps 0}
  \end{equation}
  As a $C^\infty$-module,
  $\ax^\infty \simeq \Omega^{\leq 2}_+
                 \oplus \eps \Omega_+^{\geq 2}$.
  The tensor product of this dgca with any finite-dimensional complex Lie algebra $\ym$
  yields a dgLa, with $\ym$ ungraded,
  \[
      \gx^\infty = \ax^\infty \otimes \ym
    \]
  The associated MC-equation yields ordinary YM with `internal Lie algebra' $\ym$.
\end{prop}
\newcommand{\ddr}{d_\Omega}
\begin{proof}
  The dgca axioms hold for $\C \oplus \C \eps$.
  In \eqref{eq:ainf}, the numerator is a subcomplex, and a subalgebra using
  $\Omega \Omega_+^2 \subset \smash{\Omega_+^{\geq 2}}$.
  The denominator is a subcomplex of the numerator,
  and an algebra ideal using 
  $\Omega_+^2 \Omega_-^2 = 0$.
  So $\ax^\infty$ is a dgca subquotient.
	Let $u \in (\gx^\infty)^1$,
	so $u = A \oplus \eps F$ with $A \in \Omega^1 \otimes \ym$
	and $F \in \Omega^2_+ \otimes \ym$.
	The MC-equations $du + \tfrac{1}{2}[u,u]=0$
  are 
  $\ddr A + F + \tfrac{1}{2} [A,A] = 0$ in $\Omega^2/\Omega^2_- \otimes \ym$
  and
  $\ddr F + [A,F] = 0$ in $\Omega^3 \otimes \ym$ with $\ddr$
  the de Rham differential and the bracket
  is the product of forms and the bracket in $\ym$.
  These are the YM equations\footnote{To recover the standard real form, let $\ym$ be a real 
  Lie algebra and take tensor products over the reals.
  Restrict $A$ to the real subspace of $\Omega^1\otimes \ym$.
  Note that $ \Omega_+^2\otimes\ym$ has no real structure, 
  but relative to a coordinate system we essentially have $F=E+iB$, the 
  electric and magnetic fields. 
  Then $dF+[A,F]=0$ encodes both the equation of motion 
  and the Bianchi identity.}.
  See also
  \emph{First order Yang-Mills theory}
  in Costello \cite{costello}.
\qed\end{proof}
%--------------------------------------------
Using the standard basis for $\Omega$ generated by $dx^0, dx^1, dx^2, dx^3$
and the basis for $\Omega_+^2$ given before, we get
$\gx^\infty = C^\infty \otimes \vecx$ with
\[
  \vecx \simeq (\C \oplus \C^7 \oplus \C^7 \oplus \C) \otimes \ym
\]
and with summands in homological degrees $0,1,2,3$ respectively.
The differential is a constant coefficient first order differential operator,
the bracket is bilinear over $C^\infty$
meaning it does not involve derivatives.
%---------------------------------------------------------------
\begin{prop}[YM homology] \label{prop:ymhom}
  Let $\ax = \bigoplus_{k \in \C^4} \ax_k$
  be the algebraic variant of $\ax^\infty$.
  Then for $k \neq 0$ there is a global canonical isomorphism
  between the homology of $\ax_k$
  and that of $\cc_{-1} \oplus \cc_1$.
  This extends trivially to 
  $\gx = \ax \otimes \ym$.
\end{prop}
\begin{proof}
  Abbreviate $I_{\pm} = \Omega_{\pm}^{\geq 2}$
  viewed as a complex with the de Rham differential.
  Let
  $\Omega \oplus \eps I_+$ be the complex with differential
  $a \oplus \eps b \mapsto (da + b) \oplus (-\eps db)$.
  There is a short exact sequence of complexes
  \begin{equation}\label{eq:ymses}
      0 \to I_- \oplus I_+ \to (\Omega \oplus \eps I_+)
      \oplus I_+ \to \ax^\infty \to 0
  \end{equation}
  The left term is a direct sum of complexes.
  The middle term has a differential that in $2 \times 2$
  block form is lower triangular
  with the lower left term
  $\Omega \oplus \eps I_+ \to I_+$, $a \oplus \eps b \mapsto b$,
  which yields a differential.
  The first map in the sequence is the direct sum of the two inclusion maps,
  hence a chain map. We get a short exact sequence 
  inducing the correct differential on
  $\ax^\infty \simeq (\Omega \oplus \eps I_+)/I_-$.
  
  Passing from $C^\infty$ to the algebraic level,
  analogous to \eqref{eq:infalg},
  we get a short exact sequence of complexes
  for every $k \in \C^4$.
  If $k \neq 0$ then the middle term
  is exact because if we reorganize
  as $(\eps I_+ \oplus I_+) \oplus \Omega$
  then the differential is lower block triangular
  as a $2 \times 2$ matrix
  and exact since both diagonal entries are if $k \neq 0$,
  namely $\eps I_+ \oplus I_+$ is exact
  since it is the mapping cone for the identity map,
  and $\Omega$ is exact if $k \neq 0$.
  So for $k \neq 0$,
  the associated long exact sequence in homology yields
  an isomorphism of the homology of $\ax_k$ with
  the homology of $I_- \oplus I_+$
  with a degree shift by one, which is $\cc_{-1} \oplus \cc_1$
  by \eqref{eq:omrep}, \eqref{eq:om2rep}.
\qed\end{proof}
%---------------------------------------------------------------------------

To define the dgLa for GR
we use Minkowski spacetime as auxiliary background structure.
This account obscures invariance properties,
and the more invariant
account \cite{rt} does not use $\Omega = C^\infty \otimes \wedge \C^4$ the way we do here.
We define directly the complex version,
but there is an obvious real structure at every step
so nothing is lost.
Consider the direct sum of Lie algebras $\gr = \C^4 \oplus \so$
which in particular is not the Lie algebra of the Poincare group.
We have the following Lie algebra representations:
\begin{itemize}
  \item The Abelian Lie algebra $\C^4$ acts on $\Omega$
    by differentiation on $C^\infty$ by 
 identifying $\C^4 \simeq \C \p_0 \oplus \C \p_1 \oplus \C \p_2 \oplus \C \p_3$
where $\p_0,\p_1,\p_2,\p_3$
are the usual partial derivatives,
and by acting trivially on $\wedge \C^4$.
\item The Lie algebra $\so \simeq \sl \oplus \sl$,
  which is the complexification of the Lie algebra of the Lorentz group,
  acts on $\Omega$ by acting trivially on $C^\infty$, and as the fundamental representation
on $dx^0,dx^1,dx^2,dx^3$ which extends uniquely to an action
as derivations of degree zero on $\wedge \C^4$.
\end{itemize}
They combine to an action of $\gr$ on $\Omega$
denoted $v \mapsto (\omega \mapsto v(\omega))$.
Then the tensor product of vector spaces
$\Omega \otimes \gr$
is a gLa with grading from $\Omega$ and 
bracket\footnote{%
This bracket is well-defined, bilinear,
and satisfies the graded Jacobi identity.
More generally, if $X$ is a graded commutative algebra, $P$ a Lie algebra 
and $P \to \mathrm{Der}^0(X)$ a Lie algebra map,
so a representation as $\C$-linear derivations of degree zero, 
then on $X \otimes_\C P$ we get a gLa bracket by setting
$[xp,x'p'] = (xp(x'))p' - (p'(x)x')p + (xx')[p,p']$.}
\begin{equation}\label{eq:br}
  \underbrace{[\omega \otimes v,\omega' \otimes v']}_{\text{bracket in $\Omega \otimes \gr$}}
      = (\omega v(\omega'))\otimes v' - (v'(\omega) \omega') \otimes v
      + (\omega \omega') \otimes \underbrace{[v,v']}_{\mathclap{\text{bracket in $\gr$}}}
    \end{equation}
for all $\omega,\omega' \in \Omega$ and $v,v' \in \gr$.
Note that $\Omega \otimes \gr$ is naturally a
module over $\Omega$ by left-multiplication,
but the bracket is not bilinear over $\Omega$
nor even $C^\infty$\footnote{%
  So the bracket involves differentiation.
The language of algebroids captures this structure.}.
Elements of degree one sometimes define a metric, as follows.
\begin{definition}[Associated frame and metric] \label{def:met}
    Given a real element of $\Omega^1 \otimes \gr$, drop the $\so$ part,
which leaves something of the form
$e_a^\mu dx^a \otimes \p_\mu$ with summation implicit,              
and $e_a^\mu \in C^\infty$. Consider the real vector fields
$e_a = e_a^\mu \p_\mu$.
If linearly independent,
we get a metric by
$g(e_a,e_b)_{a,b=0,1,2,3} = \textnormal{diag}(-1,1,1,1)$.
\end{definition}
An example is the real element of $\Omega^1 \otimes \gr$ given by
\eqref{eq:mink},
which is in $\MC(\Omega \otimes \gr)$ and whose associated metric is the Minkowski metric.
But not every solution to the Einstein equations of GR comes from a real 
nondegenerate element of $\MC(\Omega \otimes \gr)$,
so $\Omega \otimes \gr$ is not yet the right object.
\begin{prop}[GR dgLa]
  As a representation, $\wedge^2 \C^4 \otimes \so$
  contains exactly one copy of $\rep{2}{0} \oplus \rep{0}{2}$.
  Denote the span over $\Omega$ of this copy 
  by $I \subset \Omega^{\geq 2} \otimes \gr$.
  Then $I$ is an ideal of the gLa $\Omega \otimes \gr$ with bracket \eqref{eq:br}.
  Set
  \[
    \gx^\infty = \frac{\Omega \otimes \gr}{I}
  \]
  a gLa quotient. Then the element
  \begin{equation}\label{eq:mink}
     m = dx^0 \otimes \p_0 + dx^1 \otimes \p_1 + dx^2 \otimes \p_2 + dx^3 \otimes \p_3
   \end{equation}
  satisfies $[m,m] = 0$ in $\gx^\infty$,
  and the associated differential
  $d = [m,-]$ makes $\gx^\infty$ a dgLa. The MC-equation yields ordinary GR
  (Ricci-flat metrics)
  for real elements that define a nondegenerate frame as in
  Definition \ref{def:met}\footnote{%
  Since $m$ has a nondegenerate frame, so do all elements close to $m$.
  }.
\end{prop}
\begin{proof}[Sketch]
  As representations
  $\wedge^2 \C^4 \simeq \so \simeq \rep{1}{0} \oplus \rep{0}{1}$
  so by \eqref{eq:cg} their product contains
  $\rep{2}{0} \oplus \rep{0}{2}$ once,
  and $I$ is isomorphic to $C^\infty$ times
\begin{equation}\label{eq:icplx}
0 \to \rep{2}{0} \oplus \rep{0}{2} \to \rep{\tfrac32}{\tfrac12} \oplus \rep{\tfrac12}{\tfrac32}
                             \to \rep{1}{0} \oplus \rep{0}{1} \to 0
\end{equation}
The arrows, not used in this proof, indicate the essentially
unique Lorentz equivariant differential linear in $k \in \rep{\frac12}{\frac12}$,
equivalently $[m,-]$.
  By \eqref{eq:cg}, \eqref{eq:omrep} one sees that
  $I$ is in the kernel of 
  $\Omega \otimes \gr \to \End(\Omega)$,
  $\omega \otimes v \mapsto (\omega' \mapsto \omega v(\omega'))$.
  Hence $[\Omega \otimes \gr, I] \subset \Omega [\gr, I] \subset \Omega I \subset I$,
  an ideal.
  One gets ordinary GR, see \cite{rt}\footnote{%
    The paper \cite{rt} uses a more general variant with $\gr$ replaced by
  $\C^4 \oplus (\so \oplus \C)$. This generalization
  is not needed for this paper,
  but it can provide flexibility in other applications.}.
\qed\end{proof}
Using standard bases, we get
$\gx^\infty = C^\infty \otimes \vecx$ with
\[
      \vecx \simeq \C^{10} \oplus \C^{40} \oplus \C^{50} \oplus \C^{24} \oplus \C^4
\]
with summands in homological degrees $0,1,2,3,4$.
The stabilizer Lie algebra of $m$ given by
$\{ x\in (\gx^{\infty})^{0}\mid [x,m]=0 \}$ acts as dgLa automorphisms,
corresponding to infinitesimal translations and Lorentz transformations.
The differential and the bracket are constant coefficient
first order differential operators.
%----------------------------------------------------------
\begin{prop}[GR homology] \label{prop:grhom}
  Let $\gx = \bigoplus_{k \in \C^4} \gx_k$ be the algebraic variant of
  $\gx^\infty$.
  Then for $k \neq 0$ there is a global canonical isomorphism
  between the homology of $\gx_k$
  and that of $\cc_{-2} \oplus \cc_2$.
\end{prop}
\begin{proof}
  Recall that $m$ given by \eqref{eq:mink} satisfies $[m,m] = 0$
  both in $\Omega \otimes \gr$
  and in $\gx^\infty$,
  and therefore defines a differential on both by $[m,-]$,
  and by restriction on the ideal $I$.
  Therefore we have a short exact sequence of complexes
\begin{equation}\label{eq:grses}
0 \to I \to \Omega \otimes \gr \to \gx^\infty \to 0
\end{equation}

Passing from $C^\infty$ to the algebraic level,
see \eqref{eq:infalg}, we get a short exact
sequence of complexes for every $k \in \C^4$.
If $k \neq 0$ then the middle term is exact because if we reorganize as
$(\Omega \otimes \so) \oplus (\Omega \otimes \C^4)$ then the differential
is lower block triangular as a $2\times 2$ matrix and exact
because the diagonal entries are exact, namely
the differential on both $\Omega \otimes \so$ and $\Omega \otimes \C^4$
is simply the de-Rham differential on $\Omega$,
and $\Omega$ is exact when $k \neq 0$. So for $k \neq 0$,
the associated long exact sequence in homology yields an isomorphism of the homology
of $\gx_k$ with the homology of $I$ with a degree shift by one, which is
$\cc_{-2} \oplus \cc_2$ by \eqref{eq:icplx}.
\qed\end{proof}
We now construct the homotopy equivalence
required for Theorem \ref{theorem:ymgrnew}.
%------------------------------------------------------------------------
\newcommand{\leftmap}{\ell}
\newcommand{\rightmap}{r}
\newcommand{\zzc}{C}
\begin{lemma}[A zig-zag lemma with homotopy equivalence] \label{lemma:zz}
  Given is a short exact sequence
  of complexes of vector spaces\footnote{%
  Or modules over a unital commutative ring.}
  with exact middle term,
     \[
       0 \to \zzc'' \xrightarrow{\rightmap' } \zzc' \xrightarrow{\rightmap  } \zzc \to 0
     \]
  with differentials $d''$, $d'$, $d$.
  Suppose $h' \in \End(\zzc')$ is a homotopy that witnesses
  the exactness of $\zzc'$,
  so $(h')^2 = 0$ and $h'd' + d'h' = \mathbbm{1}$.
  Suppose $\leftmap'  \in \Hom(\zzc',\zzc'')$, $\leftmap  \in \Hom(\zzc,\zzc')$
  witness the exactness of the short exact sequence,
  $\rightmap' \leftmap' +\leftmap \rightmap  
  = \rightmap  \leftmap 
  = \leftmap' \rightmap' 
  = \mathbbm{1}$\footnote{Note that $\leftmap ,\leftmap' $ need not be chain maps.}.
  Then a homotopy equivalence $C'' \leftrightarrow C$, with a degree shift by one, 
  is given by $R = \rightmap  h'\rightmap'  \in \Hom(\zzc'',\zzc)$,
  $L = \leftmap' d'\leftmap  \in \Hom(\zzc,\zzc'')$,
  $\rightmap  h'\leftmap \in \End(\zzc)$,
   $\leftmap' h'\rightmap' \in \End(\zzc'')$,
  and $L$ induces the usual connecting morphism.
\end{lemma}
\begin{proof}
  For example,
    $Ld = \leftmap' d'\leftmap d
    = \leftmap' d'\leftmap d\rightmap  \leftmap 
    = \leftmap' d'\leftmap \rightmap  d'\leftmap 
    = \leftmap' d'(\one-\rightmap' \leftmap' )d'\leftmap 
    = -\leftmap' d'\rightmap' \leftmap' d'\leftmap 
    = -\leftmap' \rightmap' d''\leftmap' d'\leftmap  
    = -d''L$
  and analogously
  $Rd'' = -dR$.
\qed\end{proof}
%---------------------------------------------------------------------
\begin{lemma}[Homotopy equivalences in Theorem \ref{theorem:ymgrnew}] \label{lemma:qqqiso}
  The isomorphisms in Propositions \ref{prop:ymhom} and \ref{prop:grhom}
  have the property \eqref{claim:qisos}.
\end{lemma}
\begin{proof}
  Apply Lemma \ref{lemma:zz}
  to the short exact sequences \eqref{eq:ymses} and \eqref{eq:grses}.
  Here $\rightmap'$ and $\rightmap$ are matrices with constant
  meaning $k$-independent complex entries,
  hence $\leftmap' $ and $\leftmap $ can be taken to be constant matrices.
  The middle term's differential
  has polynomial entries,
  and it is exact at every $k \neq 0$.
  A homotopy exists pointwise since we are over $\C$.
  Pick a homotopy at the given point $k \neq 0$.
  Use the HPL \eqref{eq:hpl} to extend the homotopy
  to a neighborhood, with rational entries in $k$.
  See the proof of Theorem \ref{theorem:opthom}
   for an application of the HPL, with more details.
\qed\end{proof}
%-----------------------------------------------------------------------
\begin{remark}[Color-ordered amplitudes]
We expect that the $A_\infty$ minimal model for the dgca $\ax$
gives the so-called color-ordered amplitudes.
We have not pursued interesting relations known as Bern-Carrasco-Johansson or BCJ,
and Kawai-Lewellen-Tye or KLT relations, in the physics literature.
\end{remark}

%% file: OptimalHomotopies.tex
\section{Optimal homotopies} \label{sec:opthom}

We construct homotopies
for complexes $C$ of vector spaces
that depend on a parameter $k \in \C^\npar$.
We assume that the differential $d$ is a matrix
with entries in the polynomial ring
$\C[k] = \C[k_1,\ldots,k_\npar]$
and that there is homology only
along the zero locus of an irreducible polynomial $\Q \in \C[k]$.
Under some assumptions, we construct homotopies with
entries rational in $k$
that degenerate just as much as they have to along $\Q = 0$,
formula \eqref{eq:opthom} below.
The construction is an iterated
application of the homological perturbation lemma, HPL.

This section applies in particular to the complexes $\cc_{\pm h}$
that depend parametrically on the momentum $k \in \C^4$
for which the homology carrying
subvariety $\Q=0$ is the light cone, see Lemmas \ref{lemma:sph}, \ref{lemma:sphdgLa}.
Recall that homotopies play the role of
propagators in YM and GR, and they encode gauge choices.
\step

Consider this situation:
\vskip 4mm
\begin{center}
  \begin{tikzpicture}
  \draw [black] (0,0)
  to[bend angle = 8, bend left] (3,0.5)
  to[bend angle = 12, bend left] (3.5,1.5) node [anchor=west,yshift=-10] {$\Q=0$}
  to[bend angle = 8, bend right] (0.5,1)
  to[bend angle = 12, bend right] cycle;
  \draw [black,->] (1.85,0.85) -- (1.95,-0.5) node [anchor=west,xshift=3] {$\xi$};
  \draw [black,fill=blue] (1.85,0.85) circle (0.07) node [anchor=west,xshift=3] {$q$};
\end{tikzpicture}
\end{center}
Start from a smooth point $q$ of the homology carrying variety $\Q=0$.
To construct a homotopy in a neighborhood of $q$
we need an assumption that
we paraphrase as \emph{the homology disappears to first order
transversal to the variety at $q$}. To state this precisely,
pick any vector $\xi$ transversal to the variety at $q$.
Differentiate $d^2 = 0$ to find,
with a dot denoting a derivative along $\xi$,
\[
  \dot{d}d+d\dot{d}=0 \qquad 
  2\dot{d}^2 + \ddot{d}d + d\ddot{d} = 0
\]
The first implies that
$\dot{d}$ induces a linear map on the homology $\hx_q$ at $q$.
The second implies that this induced map is a differential.
The assumption will be that $\hx_q$,
as a complex with differential induced by $\dot{d}$, is exact.
\begin{definition}[Regular homology point] \label{def:regp}
  Let $C$ be a finite-dimensional vector space
  with a differential $d$ with entries in the polynomial ring $\C[k]$.
  Let $\hx_k$ be the homology at $k \in \C^\npar$.
  Suppose $\Q \in \C[k]$ is irreducible
  and that $\hx_k \neq 0$ implies $\Q(k) = 0$.
  We say that $q \in \C^\npar$ is a regular homology point if:
  \begin{itemize}
    \item $\Q(q) = 0$ and $q$ is a smooth point of the variety $\Q=0$.
    \item $\dim \hx_k = \dim \hx_q$ for all $k$ with $\Q(k) = 0$
      in a Zariski neighborhood of $q$.
    \item There exists a $\xi \in T_q\C^\npar$ such that,
      with a dot denoting a derivative along $\xi$ at $q$,
      we have $\dot{\Q} \neq 0$ and the
      differential on $\hx_q$ induced by $\dot{d}$ is exact.
  \end{itemize}
\end{definition}
\begin{example} Set $\npar=1$ and
    consider $C : 0 \to \C \xrightarrow{k^a} \C \to 0$ and $\Q = k$.
    Then $q=0$ is regular homology point if $a=1$, but not if $a \geq 2$.
\end{example}
\begin{theorem}[Optimal homotopy] \label{theorem:opthom}
  Given a regular homology point $q$,
  there exist matrices $h,i,p,{\dprime},{\hprime}$
  with entries in the field of fractions $\C(k)$
  whose denominators do not vanish at $q$,
  such that over $\C(k)$ we have:
  \begin{itemize}
    \item $h^2 = 0$,\;$hdh=h$,\;$ip = \one-dh-hd$,\;$pi = \one$,
    and at $q$ we have $dhd = d$.
    \item $pdi = \Q {\dprime}$,\;$({\dprime})^2 = ({\hprime})^2 = 0$,\;${\hprime}{\dprime} + {\dprime}{\hprime} = \one$.
  \end{itemize}
  Furthermore $\hoff^2 = 0$ and $\hoff d+d\hoff = \one$ where
      \begin{equation}\label{eq:opthom}
        \hoff = h + \frac{1}{\Q} i{\hprime}p
      \end{equation}
  One can freely choose $h$, $i$, $p$ at $q$
  provided they satisfy the first bullet at $q$.
  If $C$ is graded, 
  everything can be made compatible with the grading.
  \end{theorem}
  In summary,
  $h,i,p$ is a contraction regular at and near $q$;
  the contraction is to $\hx_q$;
  the induced differential is of the form $\Q {\dprime}$ with a regular differential
  ${\dprime}$ that is exact as witnessed by a regular ${\hprime}$;
  and $C$ is exact over $\C(k)$
  as witnessed by a homotopy $\hoff$ that is regular except for an explicit $1/\Q$.

The proof below is in two stages. With reference to \eqref{eq:hip}:
\begin{itemize}
\item {\bf 1st stage.} 
If $k$ is a point close to $q$, then the differential $d$ is a small perturbation of $d_q$.
Therefore starting from a contraction diagram at $q$, it extends to 
a Zariski open neighborhood using the HPL:
\[
\begin{tikzpicture}[baseline=(current  bounding  box.center)]
  \matrix (m) [matrix of math nodes, column sep = 17mm, minimum width = 7mm ]
  {
    C\vphantom{\hx_q} & \hx_q \\
  };
  \path[-stealth]
    (m-1-1) edge [out=15,in=180-15] node [above] {$p$} (m-1-2)
            edge [out=180+25,in=180-25,min distance=8mm] node
                 [left,xshift=-1mm] {$\mathllap{d}$} (m-1-1)
            edge [loop above,out=90+25,in=90-25,min distance=8mm] node
                 [above] {$h$} (m-1-1)
    (m-1-2) edge [out=180+15,in=-15] node [below] {$i$} (m-1-1)
            edge [out=25,in=-25,min distance=8mm] node
                 [right,xshift=1mm] {$\mathrlap{\Q {\dprime}}$} (m-1-2);
\end{tikzpicture}
\]
\item
{\bf 2nd stage.}
The factorization $\Q {\dprime}$ is proved using the Nullstellensatz,
and it is natural to study ${\dprime}$.
The assumptions imply that ${\dprime}$ is exact at $q$,
hence in a neighborhood by the HPL:
\[
\begin{tikzpicture}[baseline=(current  bounding  box.center)]
  \matrix (m) [matrix of math nodes, column sep = 17mm, minimum width = 7mm ]
  {
    \hx_q & 0\vphantom{\hx_q} \\
  };
  \path[-stealth]
    (m-1-1) edge [out=15,in=180-15] node [above] {$0$} (m-1-2)
            edge [out=180+25,in=180-25,min distance=8mm] node
                 [left,xshift=-1mm] {$\mathllap{{\dprime}}$} (m-1-1)
            edge [loop above,out=90+25,in=90-25,min distance=8mm] node
                 [above] {${\hprime}$} (m-1-1)
    (m-1-2) edge [out=180+15,in=-15] node [below] {$0$} (m-1-1)
            edge [out=25,in=-25,min distance=8mm] node
                 [right,xshift=1mm] {$\mathrlap{0}$} (m-1-2);
\end{tikzpicture}
\]
\end{itemize}
Contraction diagrams can be composed, and this gives $\hoff$.

  \begin{proof}
    Denote by $d_q$ the differential at $q$, a matrix with complex entries.
    Such a differential always admits a contraction $h_q,i_q,p_q$
    with $d_qh_qd_q = d_q$, given by matrices with complex entries.
    Set $\delta = d - d_q$ and
    \[
        h = h_q(\one + \delta h_q)^{-1}
        \qquad
        i = (\one + h_q \delta)^{-1} i_q
        \qquad
        p = p_q(\one + \delta h_q)^{-1}
    \]
    Since $\delta$ has entries in $\C[k]$ and vanishes at $q$, Cramer's rule
    implies that $h,i,p$ are matrices with entries in $\C(k)$
    with denominators that do not vanish at $q$.
    The HPL \eqref{eq:hpl} implies that
    $h,i,p$ is a homotopy of $d$ over $\C(k)$
    but not necessarily $dhd = d$, though this does hold at $q$ by construction.
    The homologies of $d$ and $pdi$ coincide at each point.
    The assumption that $k \mapsto \dim \hx_k$
    be constant along $\Q=0$ near $q$
    implies that $pdi$ vanishes along $\Q=0$
    because the dimension of the homology drops at points where $pdi \neq 0$.
    So Hilbert's Nullstellensatz and the irreducibility of $\Q$
    imply $pdi = \Q {\dprime}$ for some matrix ${\dprime}$ with entries in $\C(k)$
    whose denominators do not vanish at $q$.
    Necessarily $({\dprime})^2 = 0$ and we now show that ${\dprimeaux{q}}$ is exact.
    Differentiating $pdi = \Q {\dprime}$ and evaluating at $q$ yields
    \begin{equation}\label{eq:dotequiv}
      p_q\dot{d}i_q = \dot{\Q}{\dprimeaux{q}}
    \end{equation}
    using $p_qd_q = d_qi_q = 0$.
    The differential $p_q\dot{d}i_q$ is
    that on $\hx_q$ induced by $\dot{d}$
    which was assumed exact, so now $\dot{\Q} \neq 0$
    implies that ${\dprimeaux{q}}$ is exact.
    Hence it admits a homotopy ${\hprimeaux{q}}$ with
    $({\hprimeaux{q}})^2 = 0$ and ${\hprimeaux{q}}{\dprimeaux{q}} + {\dprimeaux{q}}{\hprimeaux{q}} = \one$.
    Set
    \[
        {\hprime} = h_q'(\one + ({\dprime}-{\dprimeaux{q}}) h_q')^{-1}
    \]
    By Cramer's rule,
    ${\hprime}$ has entries in $\C(k)$ with denominators that do not vanish at $q$,
    and the HPL \eqref{eq:hpl}
    implies $({\hprime})^2 = 0$ and ${\hprime}{\dprime}+{\dprime}{\hprime} = \one$.
    Define $\hoff$ by \eqref{eq:opthom}.
    It is now immediate that it has all the required properties.
 \qed\end{proof}
%------------------------------------------------------
\begin{lemma}[Canonical exact differential on homology] \label{lemma:uniqd}
	The map
	$i {\dprime} p$ induces an exact differential on $\hx_q$,
	at all regular homology points $q$ where it is regular.
	It is equivalently induced by
	$\dot{d}/\dot{\Q}$ for all transversal $\xi\in T_q\C^{\npar}$.
\end{lemma}
\begin{proof}
By \eqref{eq:dotequiv}, which holds for all $\xi $.
\qed\end{proof}
%---------------------------------------------------------
\begin{lemma}[Trivial homotopy away from $\Q=0$] \label{lemma:trivhom}
  If $\Q(q) \neq 0$ so that $\hx_q=0$, then there exists a matrix $h$
  with entries in $\C(k)$
  whose denominators do not vanish at $q$,
  such that over $\C(k)$ we have $h^2 = 0$ and $dh+hd = \one$.
\end{lemma}
\begin{proof} Again, an application of the HPL.
\qed\end{proof}
%------------------------------------------------------
\begin{lemma}[Regular homology point stable under homotopy equivalence] \label{lemma:rhpx}
Suppose $C$, $C'$ are two complexes with homology only along $\Q = 0$
with $\Q$ irreducible.
Suppose $\Q(q) = 0$ and there exists a homotopy equivalence
$C \leftrightarrow C'$ given by four matrices with entries in $\C(k)$, regular at $q$.
Then
$q$ is a regular homology point of $C$ 
if and only if
$q$ is a regular homology point of $C'$.
\end{lemma}
\begin{proof}
We work with matrices
whose entries are in $\C(k)$ with denominators that do not vanish at $q$.
If $q$ is a regular homology point for $C$,
apply Theorem \ref{theorem:opthom} to it.
Composition yields a homotopy equivalence $(\hx_q,\Q\dprime) \leftrightarrow (C',d')$
by matrices $R,L,\witness,\witness'$ with
$d'R = \Q R \dprime$
and $L d' = \Q \dprime L$
and $RL = \one - \witness'd'-d'\witness'$
and $LR = \one - \Q \witness\dprime - \Q \dprime \witness$.
Hence
$Ld' R = \Q \dprime (\one - \Q \witness \dprime - \Q \dprime \witness)$.
The derivative along $\xi \in T_q\C^\npar$ yields
$L\smash{\dot{d}'} R = \dot{\Q} \dprime$
since $L d' = d' R = \Q = 0$ at $q$.
Since $R$, $L$ are quasi-isomorphisms,
the differential that $\smash{\dot{d}'}$
induces on homology is isomorphic to $\dot{\Q} \dprime$ hence exact.
So $q$ is a regular homology point of $C'$.
\qed\end{proof}
%------------------------------------------------------
\begin{lemma}[Single particle homology,
              cf.~Lemma \ref{lemma:sph2}] \label{lemma:sph}
  The complex $\cc_{\pm h}$ has no homology if $\Q = ad-bc \neq 0$.
  If $\Q(q) = 0$ but $q \neq 0$ then
  \[
    \hx_q^i \simeq \begin{cases} \C & \text{if $i=1,2$}\\
      0 & \text{if $i \neq 1,2$}
    \end{cases}
  \]
  and every such point $q$ is a regular homology point, so Theorem \ref{theorem:opthom} applies.
\end{lemma}
\begin{proof}
It suffices to check this for one point
in every orbit of the automorphism group of the cone $\Q = 0$.
Convenient are $(a,b,c,d) = (1,0,0,1)$
as well as $(a,b,c,d) = (1,0,0,0)$ and $\xi = (0,0,0,1)$.
Use \eqref{eq:ccm}.
\qed\end{proof}
\newcommand{\sleft}{{\cc}}
\newcommand{\sright}{\vecx}
%---------------------------------------------------------------------------
\begin{lemma}[Single particle homology for the dgLa] \label{lemma:sphdgLa}
Let $h=1,2$.
Take any dgLa as in Theorem \ref{theorem:ymgrnew},
viewed as a vector space $\vecx$ with differential polynomial in $k$.
Use $\cc$ as a shorthand for \eqref{eq:hxx}.
Let $\Q = ad-bc$ and pick any $q \neq 0$ with $\Q(q) = 0$.
Then $q$ is a regular homology point of $\vecx$ and $\cc$,
and there exists a diagram of homotopy equivalences,
commutative up to homotopy equivalence,
by matrices with entries in $\C(k)$ whose denominators are nonzero at $q$:
\newcommand{\bax}{2pt}
\begin{equation}\label{eq:sqhdgLa}
\begin{tikzpicture}[baseline=(current  bounding  box.center)]
  \matrix (m) [matrix of math nodes, column sep = 35mm, row sep = 11mm, minimum width = 7mm, minimum height = 8mm ]
  {
    \sleft & \sright \\
    (\hx_q^\sleft,\Q \dprime^\sleft) & (\hx_q^{\sright}, \Q \dprime^{\sright})\\
  };
  \path[-stealth]
    (m-1-1) edge [transform canvas={yshift=\bax}] node [above] {$R$} (m-1-2)
    (m-1-2) edge [transform canvas={xshift=\bax}] node [right] {$p^{\sright}$} (m-2-2)
    (m-2-2) edge [transform canvas={yshift=-\bax}] node [below] {$L' = p^\sleft Li^{\sright}$} (m-2-1)
    (m-2-1) edge [transform canvas={xshift=-\bax}] node [left] {$i^\sleft$} (m-1-1)
    %--------------
    (m-1-2) edge [transform canvas={yshift=-\bax}] node [below] {$L$} (m-1-1)
    (m-2-2) edge [transform canvas={xshift=-\bax}] node [left] {$i^{\sright}$} (m-1-2)
    (m-2-1) edge [transform canvas={yshift=\bax}] node [above] {$R' = p^{\sright} Ri^\sleft$} (m-2-2)
    (m-1-1) edge [transform canvas={xshift=\bax}] node [right] {$p^\sleft$} (m-2-1);
\end{tikzpicture}
\end{equation}
Here:
\begin{itemize}
\item The homotopy equivalence $R,L$ is as in \eqref{claim:qisos}.
\item The vertical contractions are two separate applications of Theorem \ref{theorem:opthom}.
\item The homotopy equivalence $R',L'$ is defined by composition, as indicated.
\end{itemize}
The matrices $R',L'$ are also a homotopy equivalence of the exact complexes
\[
(\hx_q^\sleft,\dprime^\sleft) \leftrightarrow (\hx_q^{\sright}, \dprime^{\sright})
\]
and $\hprime^\sleft L' = L' \hprime^{\sright}$ and $R' \hprime^\sleft = \hprime^{\sright} R'$.
Along $\Q = 0$ we have $R'L' = L'R' = \one$.
\end{lemma}
%--------------------------------------------
\begin{proof}
The first part is by Lemmas \ref{lemma:rhpx}, \ref{lemma:sph}.
Since composing homotopy equivalences yields a homotopy equivalence,
$\smash{\Q \dprime^\sleft L'} = \smash{\Q L' \dprime^{\sright}}$,
$\smash{\Q R' \dprime^\sleft} = \smash{\Q \dprime^{\sright} R'}$.
Canceling $\Q$ we get
$\smash{\dprime^\sleft L'} = \smash{L' \dprime^{\sright}}$,
$\smash{R' \dprime^\sleft} = \smash{\dprime^{\sright} R'}$.
Since there is nothing outside degrees $1,2$ by Lemma \ref{lemma:sph},
conclude
$\smash{\hprime^\sleft L'} = \smash{L' \hprime^{\sright}}$,
$\smash{R' \hprime^\sleft} = \smash{\hprime^{\sright} R'}$.
The homotopy equivalence also yields matrices $\witness,\witness'$ such that
$L'R' = \one - \smash{\Q(\dprime^\sleft \witness + \witness \dprime^\sleft)}$,
$R'L' = \one - \smash{\Q(\dprime^{\sright} \witness' + \witness' \dprime^{\sright})}$
so $R'$, $L'$ are inverses when $\Q = 0$.
\qed\end{proof}

%% file: VarietyAdmissibleMomenta.tex
\section[The variety of kinematically admissible complex momenta]{%
         The variety of kinematically admissible \\ complex momenta}
         \label{sec:variety}

The variety for $N=n+1$ momenta
is a direct product of $N$ light cones,
intersected with a codimension 4 plane that
implements momentum conservation.
We discuss its geometry emphasizing the Hartogs phenomenon,
and list the irreducible codimension one subvarieties (prime divisors) 
where amplitudes can have residues.
By convention, variety does not mean irreducible variety.
\step
The complex light cone for the momentum $k = (a,b,c,d)$
is the 3-dimensional irreducible
quadratic variety $\zl(ad-bc) \subset \C^4$.
The singular locus is the origin and has codimension 3.
Viewed as a matrix with vanishing determinant,
\[
  k = \begin{pmatrix} a & b \\ c & d \end{pmatrix}
\]
each $k$ on the cone is non-uniquely a product $k = vw^T$
of vectors $v,w \in \C^2$ called spinors.
This is a surjective morphism of varieties
\begin{equation}\label{eq:cc}
  \C^4 \to \zl(ad-bc),\qquad (v,w) \mapsto k = vw^T
\end{equation}
The preimage of the origin is the union
of the 2-planes $v=0$ and $w=0$.
\vskip 6mm
\begin{center}
  \begin{tikzpicture}[scale=0.95]
    \draw [dotted] (-3,-1.5) rectangle (1,1.5);
\draw (-1,-1.5) node [anchor=north,yshift=-10,align=center]
    {\footnotesize 4-dimensional affine space $\C^4$ of `spinors',\\
     \footnotesize the two lines indicate the two 2-planes\\
     \footnotesize that are mapped to the origin of the cone};
\draw [very thick,blue] (-3,0) -- (1,0);
\draw [very thick,blue] (-1,-1.5) -- (-1,1.5);
\draw (5,-1.5) -- (7,1.5);
\draw (7,-1.5) -- (5,1.5);
\draw (6,-1.5) ellipse (1 and 0.2);
\draw (6,-1.5) node [anchor=north,yshift=-10,align=center]
   {\footnotesize real projection of $\zl(ad-bc)$,\\
    \footnotesize the origin is the singular locus};
\draw (6,1.5) ellipse (1 and 0.2);
\draw [black,fill=blue] (6,0) circle (0.07);
\end{tikzpicture}
\end{center}
\vskip 2mm
There is a corresponding map on coordinate rings.
Namely the variable replacement rule
$k \mapsto vw^T$
induces an injective map of $\C$-algebras
    \[
      \C[k]/(ad-bc) \to \C[v,w]
    \]
Here $k$ is a shorthand for four variables
and $v,w$ are shorthands for two variables each.
This map embeds a prototypical non-UFD into a UFD.
The image consists of precisely the elements that are invariant
under the algebra automorphisms
$v \mapsto \lambda v$ and
$w \mapsto \lambda^{-1}w$
for all $\lambda\in \C^\times$.
The next lemma, that we will not actually use,
introduces Hartogs extension as a theme.
\begin{lemma}[Hartogs extension for the cone] \label{lemma:hc}
  A section of the structure sheaf of $\zl(ad-bc)$
  on the complement of the origin $k=0$
  extends uniquely to a global section,
  hence is the restriction of a polynomial in $k$ to the cone.
\end{lemma}
\begin{proof}
  This yields a section
  on the complement of
  $\{v = 0\} \cup \{w=0\}$ which has codimension 2.
  By the classical Hartogs theorem,
  it is the restriction of a unique element of $\C[v,w]$.
  This element is necessarily invariant under $v \mapsto \lambda v$
  and $w \mapsto \lambda^{-1}w$ for all $\lambda \in \C^\times$,
  hence it comes from an element of $\C[k]$.
\qed\end{proof}
\begin{remark}\label{remark:hfail}
  A Hartogs-like statement such as Lemma \ref{lemma:hc} is not automatic
  for a singular variety.
  A well-known counterexample is the image of
  \[
    \C^2 \to \C^4
    \qquad
    (s,t) \mapsto (a,b,c,d) = (s^4,s^3t,st^3,t^4)
  \]
  This 2-dimensional subvariety of the cone $\zl(ad-bc)$
  does not satisfy Hartogs.
  As a case in point,
  the section of the structure sheaf
  given away from the origin by
  $b^2/a$, $ac/b$, $bd/c$, $c^2/d$
  (every point has a neighborhood where one of the four expressions
  is defined and they match on overlaps;
  in the parametrization this corresponds to the missing $s^2t^2$)
  does not extend to a global section.
  Note that the origin has codimension 2, and it is the singular locus.
\end{remark}

Set now
$k_i = (\begin{smallmatrix} a_i & b_i \\ c_i & d_i \end{smallmatrix})$
and $\Q_i = a_id_i-b_ic_i$
and introduce 2-component vectors $v_i$ and $w_i$ called spinors.
We write $\kr{N}$ respectively $\vwr{N}$
for the polynomial rings, each with $4N$ variables.
The following varieties restrict $N$ momenta to the cone
and impose momentum conservation.

\begin{definition}[Momenta variety]
For $N \geq 3$ set:
\begin{align*}
  I_N & = (k_1 + \ldots + k_N,\Q_1,\ldots,\Q_N)
  && \text{an ideal in $\kr{N}$}\\
  \JX_N & = (v_1w_1^T + \ldots + v_N w_N^T)
  && \text{an ideal in $\vwr{N}$}
\end{align*}
Then $k_i = v_iw_i^T$
defines a surjective morphism of varieties
$\zl(\JX_N) \to \zl(I_N)$.
\end{definition}
%-----------------------------------------------------------------
Counting equations,
$\zl(I_N)$ and $\zl(\JX_N)$
ought to have dimensions $3N-4$ and $4N-4$ respectively.
To make this rigorous one must show that the defining equations
are suitably independent,
a consequence of the next lemma.
\begin{lemma}[Complete intersection] \label{lemma:ci0}
  For $N \geq 3$ the given generators of $I_N$ and $\JX_N$ are regular sequences,
  so the quotients are complete intersections.
\end{lemma}
\begin{proof}
  Sufficient
  for nonzero $r_1,\ldots,r_n$
  in a polynomial ring to be a regular sequence\footnote{%
    In a Noetherian ring $R$,
    a sequence $r_1,\ldots,r_n \in R$ is called a regular sequence
    iff multiplication by $r_{i+1}$ is an injective map on $R/(r_1,\ldots,r_i)$ for all $i=0,\ldots,n-1$.
    If $R$ is graded and the $r_i$ are \emph{homogeneous} of positive degree,
    then the notion of a regular sequence is independent of the ordering,
    and replacing $r_i$ by $r_i+pr_j$ for some $j \neq i$ and $p \in R$
    maps regular sequences to regular sequences,
    like elementary row operations in linear algebra.} is that
  there exists a monomial order such that the leading monomials
  of $r_1,\ldots,r_n$ are pairwise coprime\footnote{%
  Coprimality implies that every subsequence $r_1,\ldots,r_i$ is
  a Gr\"obner basis for the ideal that it generates,
  by Buchberger's criterion.
  Hence the well-ordered sequence of monomials
  that are not divisible by the leading monomial
  of any of $r_1,\ldots,r_i$
  is a $\C$-basis of $R/(r_1,\ldots,r_i)$.
  In this basis, one can see that
  multiplication by $r_{i+1}$ is injective 
  from $R/(r_1,\ldots,r_i)$ to itself.}.
  For $\JX_3$ the degrevlex order\footnote{%
    Degrevlex, or degree reverse lexicographic,
    is a well-known admissible monomial order.
    We will mainly use: If four symbols satisfy $a,b,c > d$
    then degrevlex implies $ab > cd$.} with
  $v_1 > w_1 > v_2 > w_2' > v_3' > w_3 > v_1' > w_1' > v_2' > w_2 > v_3 > w_3'$
  works where,
  abusing notation,
  no dash or a dash indicate
  the first or second component respectively.
  The leading monomials are $v_1w_1$, $v_2w_2'$, $v_3'w_3$, $v_1'w_1'$.
  For $I_3$ the given generators are a regular sequence
  if and only if, after eliminating $k_3$ using momentum conservation,
  $(\Q_1,\Q_2, a_1d_2-b_1c_2+a_2d_1-b_2c_1)$ is
  a regular sequence in $\C[k_1,k_2]$.
  Use degrevlex with
  $a_1 > b_1 > c_1 > d_1 > a_2 > b_2 > c_2 > d_2$,
  the leading monomials are
  $b_1c_1$, $b_2c_2$, $a_2d_1$.
  Similar for $N > 3$.
\qed\end{proof}
%----------------------------------------------------------------
\begin{prop}[Hartogs extension for the momenta variety] \label{prop:htg}
If $N \geq 3$
and $X = \zl(I_N)$ or $X = \zl(\JX_N)$,
and if $Y \subset X$ is a Zariski closed subset of codimension $\geq 2$,
then restriction defines an isomorphism
\begin{equation}\label{eq:hartogs}
      \O_X(X) \simeq \O_X(X-Y)
\end{equation}
That is to say,
every section of the structure sheaf away from $Y$
extends uniquely to a global section,
and is therefore the restriction of a polynomial to $X$.
\end{prop}
\begin{proof}
A complete intersection
is $S_2$ which implies Hartogs.
More in detail,
a complete intersection,
meaning a polynomial ring modulo a regular sequence,
is a local complete intersection,
meaning all its local rings are complete intersections.
Hence it is a Gorenstein ring and a Cohen-Macaulay ring
by Theorems 21.3 and 18.1 of \cite{mat},
and it has the property $S_k$ of Serre for all $k$,
see Section 23 of \cite{mat}.
Equation \eqref{eq:hartogs}
now follows from Proposition 1.11 of \cite{hpaper}.
The elements of
$\O_X(X)$ are well-known to be restrictions of polynomials.
\qed\end{proof}
%------------------------------------------------------------------
\begin{prop}[Singular loci]\label{lemma:sloci}
  The singular locus of $\zl(I_3)$ has codimension one.
  For $N \geq 4$ the singular loci are:
  \begin{itemize}
    \item $\zl(I_N)$:
      Points where $k_i=0$ for at least one $i = 1\ldots N$,
      and points where all the $k_1,\ldots,k_N$
      are contained in one line in $\C^4$. Codimension 3.
    \item $\zl(\JX_N)$: Points where all the $v_1w_1^T,\ldots,v_Nw_N^T$
      are contained in one line in $\C^4$.
      Codimension $2N-5 \geq 3$.
  \end{itemize}
\end{prop}
\begin{proof}
  Consider say $\zl(I_N)$, $N\ge4$.
  Abbreviate $dQ_i = (d_i,-c_i,-b_i,a_i)$.
  The Jacobian of the list of generators of $I_N$ has the block structure
\[
\begin{pmatrix}
\one_{4 \times 4} & \ldots & \one_{4 \times 4}\\
dQ_1 & \ldots & 0\\
\vdots & \ddots & \vdots\\
0 & \ldots & dQ_N
\end{pmatrix}
\]
and is surjective iff $dQ_i \neq 0$ for all $i$
and $\sum_i \ker(d Q_i) = \C^4$,
hence iff $k_i \neq 0$ for all $i$
and at least two $k_i$ are linearly independent.
This is the non-singular locus by the Jacobian criterion and 
primality of $I_N$, so that $I_N$ is the ideal of $\zl(I_N)$, 
see Section I.5 of \cite{hbook} and Lemma \ref{lemma:irr} below.
Similar for $\zl(I_N')$.
\qed\end{proof}
%----------------------------------------------------------------
\begin{lemma}[Preimage] \label{lemma:pre2}
  If $N \geq 4$,
  then under $\zl(\JX_N) \to \zl(I_N)$
  the preimage of every Zariski closed subset of codimension $\geq 2$
  has codimension $\geq 2$.
\end{lemma}
\begin{proof}
  Let $S \subset \zl(\JX_N)$
  be the union of a) the singular locus,
  b) the preimage of the singular locus of $\zl(I_N)$ under the map $\zl(\JX_N) \to \zl(I_N)$,
  and c) all points where the Jacobian of this map does not have full rank $3N-4$.
  One checks that b) contains a) and c).
  So $S$ has codimension 2.
  Consider now the claim.
  If the preimage does not have codimension $\geq 2$,
  then it has an irreducible component with codimension $\leq 1$,
  this component has (since its singular locus is proper and closed) 
  at least one non-singular point not in $S$,
  but then by a local analysis its image has codimension $\leq 1$,
  a contradiction.
\qed\end{proof}
\begin{remark} \label{remark:pre2}
  Lemma \ref{lemma:pre2} fails for $\zl(\JX_3) \to \zl(I_3)$
  because $\{k_1=0\}$ has codimension 2 but its preimage
  $\{v_1 = 0\} \cup \{w_1 = 0\}$ only has codimension 1.
% n = 3;
% R0 = QQ[v_1..v_n,vv_1..vv_n,w_1..w_n,ww_1..ww_n];
% q = (x,y) -> sum(for i from 1 to n list x_i*y_i);
% R = R0/ideal(q(v,w),q(vv,w),q(v,ww),q(vv,ww));
% codim ideal(v_1,vv_1)
\end{remark} 
%----------------------------------------------------------------
We now decompose these varieties into irreducible components
and determine the prime ideal of each component.

\begin{prop}[Prime decomposition for $N=3$] \label{prop:p3}
  The ideal $I_3$ is radical but not prime,
  instead $I_3 = I^+ \cap I^-$
        for two prime ideals $I^{\pm} \subset \kr{3}$:
      \begin{align*}
        I^+ &= \big(k_1 + k_2 + k_3,\,
         \textnormal{all maximal minors of the
           $2 \times 6$ matrix $(k_1\,k_2\,k_3)$}\big)\\
        I^- &= \big(k_1 + k_2 + k_3,\,
         \textnormal{all maximal minors of the
           $2 \times 6$ matrix $(k_1^T\,k_2^T\,k_3^T)$}\big)
      \end{align*}
      So $\zl(I_3) = \zl(I^+) \cup \zl(I^-)$, 
      where $\zl(I^+)$ are all points of the form $k_i = vw_i^{T}$ with $w_1+w_2+w_3=0$,
      and $\zl(I^-)$ are all points of the form $k_i = v_iw^{T}$ with $v_1+v_2+v_3=0$.
     The ring $\kr{3}/I^\pm$ is Cohen-Macaulay. 
  \end{prop}
\begin{proof}
  The ring $\kr{3}/I^\pm$ is isomorphic to the coordinate ring of the determinantal variety
  of $2 \times 4$ matrices of rank $<2$, hence an integral domain and Cohen-Macaulay.
  Alternatively, use Macaulay2 \cite{m2} to check that
  $I^+$ is the kernel of $\kr{3} \to \C[v,w_{1\ldots 3}]/(w_1+w_2+w_3)$, $k_i \mapsto vw_i^T$
  which realizes $\kr{3}/I^+$ as a subring of an integral domain.
  All statements in this proposition are easily checked using computer tools based on Gr\"obner bases.
%  S = QQ[a_1..a_3,b_1..b_3,c_1..c_3,d_1..d_3];
%  k = i -> matrix{{a_i,b_i},{c_i,d_i}};
%  Iplus = ideal(k(1)+k(2)+k(3)) + minors(2,k(1)|k(2)|k(3));
%  T = QQ[v1,v2,w1_1..w1_3,w2_1..w2_3]/(w1_1+w1_2+w1_3,w2_1+w2_2+w2_3);
%  f = map(T,S,flatten(transpose(for i from 1 to 3 list {v1*w1_i,v1*w2_i,v2*w1_i,v2*w2_i})));
%  ker f == Iplus
\qed\end{proof}
%S = QQ[a_1..a_3,b_1..b_3,c_1..c_3,d_1..d_3];
%k = i -> matrix{{a_i,b_i},{c_i,d_i}};
%kt = i -> transpose(k(i));
%Q = i -> det(k(i));
%Iplus = ideal(k(1)+k(2)+k(3)) + minors(2,k(1)|k(2)|k(3));
%Iminus = ideal(k(1)+k(2)+k(3)) + minors(2,kt(1)|kt(2)|kt(3));
%intersect(Iplus,Iminus) == (ideal(k(1)+k(2)+k(3))+ideal(Q(1),Q(2),Q(3)))
     Note that\footnote{%
      For instance,
       $\tr(\eps k_1^T \eps k_2)$ 
       is in the intersection but not in the product.
      } $I^+ \cap I^- \neq I^+I^-$.
%------------------------------------------------------------
\begin{prop}[Primality for $N \geq 4$] \label{lemma:irr}
  If $N \geq 4$ then $I_N$ and $\JX_N$ are prime ideals.
  The quotient rings are normal Noetherian domains.
\end{prop}
\begin{proof}
  Primality by Lemma \ref{lemma:fib3}
  with $n = N$, $f = 0$.
  Normality by Serre's criterion, since the ring is $S_2$ and
  the singular locus has codimension $>1$.
\qed\end{proof}
%------------------------------------------------------------
%------------------------------------------------------------
\begin{lemma}[Injective ring map]  \label{lemma:injrmap}
For $N \geq 3$ the assignment $k_i \mapsto v_i w_i^T$ induces an injective ring map
$\kr{N}/I_N \to \vwr{N}/\JX_N$.
\end{lemma}
\begin{proof}
An element in the kernel of $\kr{N} \to \vwr{N}/\JX_N$
vanishes on $\zl(I_N)$,
since $\zl(\JX_N) \to \zl(I_N)$ is surjective.
Therefore it is in the radical $\sqrt{I_N}$
by the Nullstellensatz, but $\sqrt{I_N} = I_N$.
%  N = 6; -- an example
%  S = QQ[a_1..a_N,b_1..b_N,c_1..c_N,d_1..d_N];
%  I = ideal(join(sum for i from 1 to N list {a_i,b_i,c_i,d_i},
%                 for i from 1 to N list a_i*d_i-b_i*c_i));
%  T = QQ[v1_1..v1_N,w1_1..w1_N,v2_1..v2_N,w2_1..w2_N];
%  vw = i -> {v1_i*w1_i,v1_i*w2_i,v2_i*w1_i,v2_i*w2_i};
%  J = ideal(sum for i from 1 to N list vw(i));
%  f = map(T/J,S,flatten(transpose(for i from 1 to N list vw(i))));
%  ker f == I 
\qed\end{proof}
%-------------------------------------------------------------
\begin{definition}[Internal momenta] \label{def:imom}
  Given $N \geq 4$,
  an internal momentum is a subset $J
  \subset \{1,\ldots,N\}$ with $2 \leq |J| \leq N-2$ and we denote
  \begin{align*}
        k_J & = \textstyle\sum_{i \in J} k_i\\
        \Q_J &= \det k_J
  \end{align*}
  in $\kr{N}/I_N$.
  Note that $k_J = -k_{J^c}$ and $\Q_J = \Q_{J^c}$ and
  $\Q_J$ is nonzero on an open dense subset of $\zl(I_N)$.
  The integer $N$ is always taken from context.
\end{definition}
%---------------------------------------------------------------
\begin{definition}[Prime divisors] Set
  $\mathfrak{P}_3 = \{\}$. For $N \geq 4$ set
  \[
    \mathfrak{P}_N = \bigg\{
        \;\;\px\subset \kr{N}/I_N\;\;\bigg|\;
        \begin{aligned}
          &\textnormal{there is an internal $J$ such that $\px$ is a minimal}\\
          &\textnormal{prime ideal over the principal ideal $(\Q_J)$}
      \end{aligned}
      \bigg\}
  \]
\end{definition}
  This is a finite set by a theorem of Noether,
  and each $\px$ has height one by Krull's principal ideal theorem.
  So $\mathfrak{P}_N$
  is a finite set of prime divisors.
%--------------------------------------------------------------------
\begin{theorem}[Classification of prime divisors] \label{theorem:pd4}
  For $N = 4$ there are $|\mathfrak{P}_4|=8$
  prime divisors, the following
  and their inequivalent permutations\footnote{%
   So $\zl(\px_{++++})$ is given by all
   points of the form $k_i = vw_i^T$ with $w_1+w_2+w_3+w_4=0$.}:
  \[
        \begingroup
        \renewcommand{\arraystretch}{1.25}
        \begin{array}{c||c|c}
          \px & \textnormal{$\px$ is generated by the maximal minors of}
          & \textnormal{$\px$ is minimal over}\\
          \hline
          \px_{++++} & (k_1\,k_2\,k_3\,k_4) &
                        \textnormal{all three $(\Q_{ij})$}\\
          \px_{----} & (k_1^T\,k_2^T\,k_3^T\,k_4^T) & 
                        \textnormal{all three $(\Q_{ij})$}\\
          \px_{++--} & \textnormal{$(k_1\,k_2)$ and $(k_3^T\,k_4^T)$} &
                        (\Q_{12}) = (\Q_{34})
          \end{array}
          \endgroup
  \]
  If $N \geq 5$
  then every prime divisor
  $\px \in \mathfrak{P}_N$
  lies over a unique $(\Q_J)$:
  \begin{itemize}
    \item 
  If $J$ and $J^c$ both have at least three elements,
  then $(\Q_J)$ is itself prime.
    \item
  For $i \neq j$ there are exactly two minimal primes over
  $(\Q_{ij})$, namely $\px_{ij}^+$ and $\px_{ij}^-$
  generated by the maximal minors of
  $(k_i\,k_j)$ respectively $(k_i^T\,k_j^T)$.
 \end{itemize}
\end{theorem}
%----------------------------------
\begin{proof}
 For $N=4$ one can make a computer check.
 For $N \geq 5$ we write these varieties as fiber products,
 see Appendix \ref{sec:fib}:
 \begin{align*}
 \zl(\Q_J)& \simeq \zl(I_{|J| + 1}) \times_{\an{4}} \cone^{|J^c|}    \\
 \zl(\px_{ij}^{\pm}) &\simeq \zl(I^{\pm}) \times_{\an{4}} \cone^{N-2}
 \end{align*}
 Here the first fiber product is defined using
 \[ 
 \zl(I_{|J|+1})
 \xrightarrow{\;\;((k_i)_{i \in J},k) \mapsto k \;\;} 
 \an{4}
 \xleftarrow{\;\;(k_i)_{i \notin J} \mapsto \sum_{i \notin J} k_i \;\;}
 \cone^{|J^c|}
  \]
 and the second is analogous. 
 Irreducibility and primality follow from 
 Lemmas \ref{lemma:fib2}, \ref{lemma:fib3}
 given that $I_{|J|+1}$, $I^{\pm}$ are prime.
 We have used $|J|,|J^c|\ge3$.
 The $\smash{\zl(\px_{ij}^{\pm})}$ have equal codimension by symmetry,
 and their union is $\zl(\Q_{ij})$, hence
 their codimension is one. Hence $\smash{\px_{ij}^{\pm}}$ are
 minimal primes over $(\Q_{ij})$ and there are no others.
 We have listed all $\px \in \mathfrak{P}_N$. 
 There are no duplicates since the $\zl(\px)$ are pairwise different.
 So every $\px$
 lies over a unique $(\Q_J)$.
\qed\end{proof}
%------------------------------------------------------------
%-------------------------------------------------------------------
\begin{definition}[The $P$ and $Z$ subsets] \label{def:pz}
  For $N \geq 4$ set
  \begin{align*}
    P_N & = \textstyle\bigcup_{\px \in \mathfrak{P}_N} \zl(\px)
    &
    Z_N & =
    \textstyle\bigcup_{\px, \qx \in \mathfrak{P}_N:\;
                       \px \neq \qx} \zl(\px) \cap \zl(\qx)
   \end{align*}
   which are closed subsets of $\zl(I_N)$.
   Let
   \[
     P_3 = Z_3 = \{k_1 = 0\} \cup \{k_2 = 0\} \cup \{k_3 = 0\}
   \]
   viewed as closed subsets of one of $\zl(I^{\pm})$.
\end{definition}
   Note that $P_N$ is equivalently the
   union of all $\zl(\Q_J)$,
   but there is no analogous definition of $Z_N$.
   Note that $Z_3$ has codimension two.
%------------------------------------------------------------------
\begin{prop}[Good properties away from a codimension two subset]
  \label{prop:cig}
  Let $N \geq 4$.
  Then $Z_N \subset \zl(I_N)$ has codimension $\geq 2$.
  In the complement $Z_N^c$:
  \begin{itemize}
    \item 
      We have $k_i \neq 0$ for all $i = 1,\ldots,N$
      and $k_J\neq 0$ for all internal $J$.\\
      The $k_i \in \C^4$ are pairwise linearly independent.
    \item
      The variety $\zl(I_N)$ is smooth of dimension $3N-4$.
    \item $\Q_J$
      has nonzero derivative tangent to $\zl(I_N)$, for all internal $J$.\\
      In particular $\zl(\Q_J)$ is smooth and has codimension one.
    \item $P_N$ is a smooth codimension one subvariety.
      If a point in $\zl(\px)$ also lies in $\zl(\Q_J)$,
      then necessarily $\px \supset (\Q_J)$ and locally $\zl(\px) \simeq \zl(\Q_J)$.
  \end{itemize}
\end{prop}
%-------------------------------------------------------------------
\begin{proof}
	Every intersection of two distinct irreducible codimension one subvarieties
	has codimension $\ge2$, so $Z_N$ has codimension $\ge2$.
  Concerning $Z_N^c$,
  several claims follow from $\zl(\px_{ij}^+) \cap \zl(\px_{ij}^-)$
  which contains all points with $k_i = 0$, 
  all points with $k_{\{i,j\}}=k_i+k_j=0$,
  and all points where $k_i$ and $k_j$ are linearly dependent.
  Points where $k_J=0$ with $|J|,|J^c| \geq 3$
  are in $\zl(\Q_J) \cap \zl(\Q_{J\cup \{i\}})\subset Z_N$
  for all $i \notin J$.
  Smoothness of $\zl(I_N)$ by Lemma \ref{lemma:sloci}.
The tangent derivative of $\Q_J=\det k_J$ is nonzero
since $k_J \neq 0$
and the Jacobian
of $Z_N^c \to \C^4, (k_1,\ldots,k_N) \mapsto k_J$
has rank four, which one sees using pairwise linear independence
of the $k_i$. 
Every point of $P_N\cap Z_N^c$ lies in a unique $\zl(\px)$,
hence in $\zl(\Q_J)$ iff $\px \supset (\Q_J)$. 
\qed\end{proof}

%% file: HelicitySheafOneParticle.tex
\section{The helicity sheaf for one particle} \label{sec:helicityone}

Here
we discuss certain rank one sheaves on the cone $X = \zl(ad-bc)$
and show how they arise as the 
homology of the complexes $\cc_{\pm h}$.
\step
In this section we denote $R = \C[k]/(ad-bc)$ and
$k = (\begin{smallmatrix} a & b \\ c & d \end{smallmatrix})$.
It is convenient to set $k^+ = k$ and $k^- = k^T$.
For every half-integer $h \geq 0$ let
\begin{equation}\label{eq:smat}
        S^{2h} k^{+} = 
        \begin{pmatrix} 
        a^{2h}&\cdots &  b^{2h} \\[-1mm] 
         \vdots & \ddots & \vdots \\[1mm]
        c^{2h}&\cdots & d^{2h}  
        \end{pmatrix}
\end{equation}
be the $(2h+1)\times (2h+1)$ matrix with entries in $R$ 
where multiplication of an entry by 
$\smash{\frac{b}{a}}=\smash{\frac{d}{c}}$ gives the next entry to the right, 
multiplication by $\smash{\frac{c}{a}}=\smash{\frac{d}{b}}$ gives the next 
entry below.
It is the symmetrized  Kronecker product of matrices,
\[ 
S^{2h}k^{+} \in \End_R(S^{2h}R^2) 
\]
using
$S^{2h}R^2 \simeq R^{2h+1}$.
The matrix $S^{2h}k^{-}$ is the transpose.
%------------------------------------------------------------
\begin{remark} \label{remark:lor1}
The Lorentz group acts as graded ring automorphisms on $R$.
The degree subspaces of $R$ are irreducible,
$R \simeq \bigoplus_h \rep{h}{h}$
where a $\C$-basis of $\rep{h}{h}$ is given by the entries of \eqref{eq:smat}.
There is a category of graded $R$-modules whose objects
$M$ are also Lorentz modules
with Lorentz equivariant scalar multiplication $R \times M \to M$,
and whose morphisms are Lorentz equivariant.
In this category,
$S^{2h}k^+$ is the unique morphism
$R \otimes \rep{0}{h} \to R \otimes \rep{h}{0}$ of degree $2h$,
up to normalization.
Hence its image is also a Lorentz module,
and its degree subspaces are seen to be irreducible\footnote{%
For every half-integer $p \geq 0$
and every equivariant map
$\rep{p}{p} \otimes \rep{0}{h} \to \rep{p+h}{p+h} \otimes \rep{h}{0}$,
its image must be contained in the unique $\rep{p}{p+h}$ subspace on the right, using \eqref{eq:cg}.
},
$\image S^{2h} k^+
\simeq \bigoplus_p \rep{p}{p+h}$.
\end{remark}
%------------------------------------------------------------
\begin{lemma}[Locally free of rank one] \label{lemma:r1}
For every maximal ideal $\mx$ corresponding to a point in $X-0$,
the complement of the origin,
there exist
column vectors $i,r$ and row vectors $\ell,s$
with entries in $R_\mx$ such that
$S^{2h}k^\pm = is$ and $\ell i = sr = 1$.
The sheaf $(\image S^{2h}k^{\pm})\wt{\phantom{x}}|_{X-0}$ is locally
 free of rank one\footnote{%
They are isomorphic to the rank one sheaves associated to multiples of the
Weil divisors $(b,d)$ and $(c,d)$ on the cone.
These divisors are mutual inverses in, and generators of, the Weil divisor class group of the cone
which is $\Z$. We do not discuss this perspective.
}.
\end{lemma}
See Appendix \ref{appendix:m2s} for the module to sheaf functor $\wt{\phantom{x}}$.
\begin{proof}
Over $R_\mx$ such a
`$\textnormal{split monomorphism} \circ \textnormal{split epimorphism}$'
factorization always exists for $k$,
for example if $a \notin \mx$ then
$k = (\begin{smallmatrix} a \\ c \end{smallmatrix})(1,b/a)$
and
$(1/a,0) (\begin{smallmatrix} a \\ c \end{smallmatrix})
= (1,b/a) (\begin{smallmatrix} 1 \\ 0 \end{smallmatrix}) = 1$.
Hence one exists for $S^{2h}k^\pm$.
Such a factorization is also valid over $R_f$
where $f$ is a common multiple of the denominators.
The locally free rank one claim follows from
$(\image S^{2h} k^{\pm})_f \simeq \image ((S^{2h} k^{\pm})_f) \simeq R_f$
using respectively the exactness of localization and the factorization.
\qed\end{proof}
%-----------------------------------------------------------
\begin{remark}\label{remark:vbp}
As rank one vector bundles on $X-0$,
these are subbundles of a trivial bundle
with fibers at $k = vw^T$ given by $\C v^{\otimes 2h}$ and $\C w^{\otimes 2h}$
respectively.
The fiberwise pairing $v^{\otimes 2h} \otimes w^{\otimes 2h} \mapsto 1$
is independent of the factorization $k = vw^T$,
hence it is globally well-defined. This is the next lemma.
\end{remark}
\begin{lemma}[Inverse sheaf] \label{lemma:isx}
There is an isomorphism of sheaves
\begin{subequations}
\begin{equation}\label{eq:inv}
(\image S^{2h}k^+)\wt{\phantom{x}}|_{X-0}
\otimes (\image S^{2h}k^-)\wt{\phantom{x}}|_{X-0}
\;\simeq\; \O_X|_{X-0}
\end{equation}
It is induced by taking powers of the identity
\begin{equation}\label{eq:kr4}
k^+ \otimes k^-
\;=\;
\bigg(\begin{smallmatrix}
a \\ b \\ c \\ d
\end{smallmatrix}\bigg)
\begin{pmatrix}
a & c & b & d
\end{pmatrix}
\end{equation}
\end{subequations}
which holds over $R$,
using the matrix Kronecker product on the left hand side.
\end{lemma}
\begin{proof}
Set $K^\pm = S^{2h} k^{\pm}$ and $L = (2h+1)^2$.
By \eqref{eq:kr4},
$K^+ \otimes K^- = is$
where $i$ is a column vector and $s$ a row vector with entries in $R$.
Their entries are a $\C$-basis of the degree $2h$ subspace $R_{2h} \subset R$,
they have the same entries as the matrix \eqref{eq:smat}.
Since $s : R^L \to R_{\ge2h}$ is surjective, $i: R_{\ge2h} \to R^L$ is injective,
they induce an isomorphism $\phi: \image(K^+ \otimes K^-) \to R_{\geq 2h}$.
Define the canonical map
$\psi: \image K^+ \otimes \image K^- \to \image (K^+ \otimes K^-)$.
Then $\phi\circ \psi : \image K^+ \otimes \image K^- \to R_{\geq 2h}$
induces \eqref{eq:inv}.
In fact, every point in $X-0$
has an open neighborhood of the standard form $D_f$
such that $(R_{\geq 2h})_f \simeq R_f$,
and $\psi_f$ is an isomorphism
using factorizations of $K^\pm$ over $R_f$
as in Lemma \ref{lemma:r1}.
\qed\end{proof}
%------------------------------------------------------------------
\begin{lemma}[Single particle homology,
              cf.~Lemma \ref{lemma:sph}] \label{lemma:sph2}
View $\cc_{\pm h}$ as a complex of free $R$-modules,
so the differential $d$ is a matrix with entries in $R$.
\begin{itemize}
\item Local claim:
For every $\mx$ corresponding to a point
in $X-0$ there exists a contraction of the differential by matrices $h,i,p$
with entries in $R_\mx$, with $dhd = d$\footnote{%
So
$h^2 = 0$,
$hdh=h$,
$dhd=d$,
$pi = \one$,
$ip = \one-dh-hd$.
See \eqref{eq:hip}.
}.
We have $i = i^1 \oplus i^2$
and $p = p^1 \oplus p^2$ 
where $i^1$, $i^2$ are column vectors and $p^1$, $p^2$ are row vectors.
In particular
$H^j(\cc_{\pm h})\wt{\phantom{x}}|_{X-0}$
is locally free of rank one if $j=1,2$ and zero otherwise.
\item Global claim: There 
are isomorphisms of locally free rank one sheaves
\begin{equation}\label{eq:zzt}
  (\image S^{2h}k^\pm)\wt{\phantom{x}}|_{X-0}
  \;\simeq\;
  H^1(\cc_{\pm h})\wt{\phantom{x}}|_{X-0}
  \;\simeq\;
  H^2(\cc_{\pm h})\wt{\phantom{x}}|_{X-0}
\end{equation}
induced by
$\image S^{2h} k^{\pm} \hookrightarrow S^{2h}R^2 \simeq \cc_{\pm h}^1$
resp.~the differential in Lemma \ref{lemma:uniqd}.
\end{itemize}
\end{lemma}
\begin{proof}
The contraction is by Theorem \ref{theorem:opthom},
available by Lemma \ref{lemma:sph}.
The contraction is valid over some $R_f$,
hence
$H(\cc_h)_f\simeq H((\cc_h)_f) \simeq R_f\oplus R_f$
since localization is exact.
This implies the local claim.
The map $\image S^{2h} k^+ \hookrightarrow \cc_{h}^1$
induces a map
$\phi: \image S^{2h} k^+ \to H^1(\cc_{h})$,
by \eqref{eq:diffk} and since $k^- \eps k^+ = 0$ over $R$.
For all $\mx$ corresponding to a point in $X-0$
the map $\phi_\mx$ is an isomorphism.
To prove this use
a factorization
as in Lemma \ref{lemma:r1} and a contraction as in this lemma
to get a map $R_\mx \to R_\mx$
and check that it is an isomorphism,
which comes down to a statement over $R_\mx/\mx R_\mx \simeq \C$
namely checking that at every point $k = vw^T$
the element $v^{\otimes 2h}$ is nonzero in homology.
This proves the first isomorphism in \eqref{eq:zzt}.
For the second,
note that $i{\dprime}p$ induces an isomorphism
$H^1(\cc_{h})_\mx \to H^2(\cc_{h})_\mx$
whose inverse is induced by $i{\hprime}p$.
These matrices from Theorem \ref{theorem:opthom} have entries in $R_\mx$.
These local isomorphisms are induced by a global isomorphism on $X-0$
by Lemma \ref{lemma:uniqd}.
\qed\end{proof}

%% file: HelicitySheaf.tex
\section{The helicity sheaf} \label{sec:helicity}

We show that a certain rank one sheaf on $\zl(I_N)$
satisfies a variant of Hartogs extension if $N \geq 4$, while
for $N=3$ there is local cohomology along $Z_3$.
These results are used respectively to prove
the uniqueness of amplitudes
and to define the 2-to-1 amplitude,
in the next section.
\step
In this section,
$R$ denotes one of $R^\pm$ if $N=3$ respectively $R_N$ if $N \geq 4$, 
and $X$ denotes one of $X^\pm$ if $N=3$ respectively $X_N$ if $N \geq 4$.
Here,
\begin{align*}
  R^{\pm} & = \kr{3}/I^{\pm} &
  X^{\pm} & = \zl(I^{\pm})\\
  R_N & = \kr{N}/I_N &
  X_N & = \zl(I_N) 
\end{align*}
Note that $X$ is irreducible, $R$ is its coordinate ring and it is Cohen-Macaulay.
Set $k_i = (\begin{smallmatrix} a_i & b_i \\ c_i & d_i \end{smallmatrix})$
and define the matrices $S^{2h}k_i^{\pm}$ as in Section \ref{sec:helicityone}.
%------------------------------------------------
\begin{definition}[The helicity module $M$] \label{def:hmm}
  For every half-integer $h \geq 0$
  and every $N$-tuple of signs $\sigma \in \{-,+\}^N$ with $N \geq 3$:
  \begin{itemize}
    \item If $N \geq 4$, let $M_h^\sigma$
  be the finitely generated $R$-module that is the image of
  the following Kronecker product of matrices:
  \begin{equation}\label{eq:mat}
    S^{2h}k_1^{\sigma_1} \otimes \cdots \otimes S^{2h}k_N^{\sigma_N}
  \end{equation}
\item If $N=3$, make the same definition over $R^\pm$,
  denoted $M^{\pm,\sigma}_h$.
\end{itemize}
We often use the shorthand $M$.
Note that the module $M$ inherits a grading from the ambient free module,
and as such is generated in degree $2hN$.
\end{definition}
%------------------------------------------------
\begin{remark} \label{remark:lorN}
Analogous to Remark \ref{remark:lor1}, here
the Lorentz group acts as automorphisms on $R$ separately on each $k_i$,
preserving $k_1 + \ldots + k_N = 0$.
Hence $M$ is also a Lorentz module.
\end{remark}
%----------------------------------------------
One always has $\wt{M}(X) \simeq M$, see Appendix \ref{appendix:m2s}.
%----------------------------------------------
\begin{theorem}[Hartogs extension for the helicity sheaf] \label{theorem:nlcm}
  If $N \geq 4$ then for all Zariski closed subsets $Y \subset X$ of codimension $\geq 2$,
  the restriction map from $X$ to $X-Y$ induces an isomorphism
  \[
      M \simeq \wt{M}(X-Y)
  \]
\end{theorem}
The proof below contains elements of the proof of Lemma \ref{lemma:hc}.
\begin{proof}
Let $R' = \vwr{N}/\JX_N$.
Recall $\phi: R \to R'$, $k_i \mapsto v_iw_i^T$
the injective ring map in Lemma \ref{lemma:injrmap}.
Let $L = (2h+1)^N$
and let $K$ be the $L \times L$ matrix \eqref{eq:mat}.
The definition of $\phi$ implies a matrix factorization $\phi(K) = IS$
where $I$ is a column vector and $S$ is a row vector with entries in $R'$.
And $I: R' \to R'^L$ is injective since $R'$ is an integral domain.
Note that $\phi^L K = IS \phi^L$ hence the injectivity of $\phi$ and $I$
yield a well-defined injective $\phi$-linear map\footnote{%
By $\phi$-linear we mean $\alpha(rm) = \phi(r)\alpha(m)$ for all $r\in R$, $m\in M$.} 
$\alpha: M \to R'$, $Kx \mapsto S \phi^L x$.
The image of $\alpha$ is 
precisely the $\C$-subspace of all $y\in R'$ that
for all $(\lambda_1,\ldots,\lambda_N) \in (\C^\times)^{N}$
transform like 
\begin{equation}\label{eq:autoy} 
y \mapsto \smash{\lambda_1^{-2h\sigma_1 }\cdots \lambda_N^{-2h\sigma_N } y} 
\end{equation}
under the algebra automorphism of $R'$ given by
$v_i \mapsto \lambda_i v_i$ and $w_i \mapsto \lambda_i^{-1} w_i$.
Set $X' = \zl(\JX_N)$.
Let $Y'$ be the preimage of $Y$ under $X' \to X$.
We have the following commutative diagram,
where $\rho$ and $\rho'$ are restriction maps,
and $\beta = \wt{\alpha}(X-Y)$
is induced by $\alpha$ and is injective\footnote{%
An injective module map 
induces an injective map on sections
since $\wt{\phantom{x}}$ and
the sections functor are left-exact.
Apply this to $\alpha\in\Hom_{R}(M,{{}_RR'})$
and use Proposition II.5.2 (d) \cite{hbook}.}:
\begin{equation}\label{eq:cdg}
\begin{tikzpicture}[baseline=(current  bounding  box.center)]
  \matrix (m) [matrix of math nodes,
               row sep = 6mm,
               column sep = 32mm ]
  {
    M \simeq \wt{M}(X) &
    R' \simeq \wt{R'}(X')\\
    \wt{M}(X-Y) &
    \wt{R'}(X'-Y')\\
  };
  \path[-stealth]
    (m-1-1) edge node [above] {$\alpha$ \footnotesize injective} (m-1-2)
            edge node [left] {$\rho$} (m-2-1)
    (m-1-2) edge node [right] {$\rho'$ \footnotesize bijective} (m-2-2)
    (m-2-1) edge node [above] {$\beta$ \footnotesize injective} (m-2-2);
\end{tikzpicture}
\end{equation}
Crucially, $\rho'$ is bijective by
Proposition \ref{prop:htg} for $X'$
and because $Y'$ has codimension $\geq 2$
by Lemma \ref{lemma:pre2} and $N \geq 4$.
The diagram implies that $\rho$ is injective.
There are obvious $(\C^{\times})^{N}$-actions on the spaces in the diagram,
given on $R'$ as above, that make the diagram equivariant.
Geometrically, because $Y'$ is $(\C^{\times})^{N}$-invariant.
Hence elements in the image of $(\rho')^{-1} \beta$ 
transform like \eqref{eq:autoy} and are contained in the image of $\alpha$.
This implies that $\rho$ is also surjective.
\qed\end{proof}
%---------------------------------------------------------------
Lemma \ref{lemma:pre2}
is not available for $N = 3$ and in fact
then there is a different result as the next lemma shows.
It is worked out in detail since
it will give the 2-to-1 amplitude
for YM and GR.
The codimension two subset $Z_3$ is chosen with
this application in mind.
There is an analogous proposition for $X^-$. 
%---------------------------------------------------------------------
\begin{theorem}[Hartogs extension failure for the helicity sheaf in $N=3$] \label{theorem:n3h}
  For all $\sigma \in \{-,+\}^3$,
  restriction induces an injection of graded $R^+$-modules\footnote{%
  The cokernel is the 1st local cohomology module
  \cite{localcohomology}
  of this sheaf along $Z_3$.}
  \[
    M_h^{+,\sigma} 
    \hookrightarrow
    \wt{M_h^{+,\sigma}}(X^+-Z_3)
  \]
  The module on the left is generated in degree $6h$.
   The module on the right is generated
  in degree $6h$ if $\sigma = {+}{+}{+},\, {-}{-}{-}$
  and in degree $4h$ if $\sigma = {+}{+}{-},\, {+}{-}{-}$.
  As a Lorentz representation, see Remark \ref{remark:lorN},
  the degree $4h$ subspace is:
\begin{equation} \label{eq:loc3}
    {+}{+}{-}\;\;:\;\;
    S^{2h} (\rep{0}{\tfrac12} \oplus \rep{0}{\tfrac12})
    \qquad\qquad
    {+}{-}{-}\;\;:\;\;
    S^{2h} \rep{\tfrac12}{0}
\end{equation}
Analogous for permutations of $\sigma$.
\end{theorem}
\begin{proof}
  This proof follows closely that of Theorem \ref{theorem:nlcm}.
  Here $M = M_h^{+,\sigma}$ and $X = X^+$
  and $R' = \C[v,w_1,w_2,w_3]/(w_1+w_2+w_3)$ an integral domain
  and $\phi: R^+ \to R'$, $k_i \mapsto vw_i^T$ an injective ring map.
  The map $\alpha: M \to R'$,
  defined analogously to the one in the proof of Theorem \ref{theorem:nlcm},
  is injective.
  The image of $\alpha$ is the $\C$-subspace of $R'$ 
  spanned by all elements of the schematic form\footnote{%
  For example, $w_1^5$ are all monomials
  in the two components of $w_1$ of degree $5$.}
\begin{equation}\label{eq:gens3}
        v^{n_1+n_2+n_3+m_1^- + m_2^- + m_3^-}
        w_1^{n_1 + m_1^+}
        w_2^{n_2 + m_2^+}
        w_3^{n_3 + m_3^+}
\end{equation}
for all $n_1,n_2,n_3 \in \Z_{\geq 0}$.
Here $m_i^\pm = h(1\pm \sigma_i)\in\{0,2h\}$.
The preimage of $Y=Z_3$ under $X' \to X$ is
$Y' = \{v=0\} \cup \{w_1=0\} \cup \{w_2 = 0\} \cup \{w_3 = 0\}$
  and has codimension 2.
  We again have the commutative diagram \eqref{eq:cdg}, in particular
  Hartogs on $X' \simeq \C^6$ implies that $\rho'$ is bijective.
  The following are equal:
  \begin{itemize}
\item[1)]
  The set of all $y\in R'$ that for $\lambda \in \C^\times$ transform like
  $y \mapsto \lambda^{-2h(\sigma_1+\sigma_2+\sigma_3)}y$
  under the algebra automorphism
  $v \mapsto \lambda v$ and $w_i \mapsto \lambda^{-1} w_i$.
    \item[2)] The $\C$-subspace of $R'$ spanned by the elements \eqref{eq:gens3}
  but allowing all $n_1,n_2,n_3 \in \Z$
  that give four nonnegative exponents.
  That is, all elements of the form
  $v^{n - 2h(\sigma_1+\sigma_2+\sigma_3)} w^n$
  with $n \geq 0$ and
  $n \geq 2h(\sigma_1+\sigma_2+\sigma_3)$
  where $w$ is schematic for all components of $w_1,w_2,w_3$.
    \item[3)] The image of $(\rho')^{-1}\beta$.
\end{itemize}
  Here $\textnormal{1)} \subset \textnormal{2)}$ is clear,
  $\textnormal{2)} \subset \textnormal{3)}$
  is proved using the fact that on $X-Z_3$ one can locally factor\footnote{%
  For example, use $v^T = (a_1,c_1)$ and
    $w_i^T = (a_i/a_1,b_i/a_1)$
    over the ring $(R^+)_{a_1}$.
  } $k_i = v w_i^T$, and
  $\textnormal{3)} \subset \textnormal{1)}$ by $\C^\times$-equivariance
  of \eqref{eq:cdg}.
  Using 2), the lowest degree pieces in the image of $(\rho')^{-1}\beta$
  are all elements $w^{6h}$ for ${+}{+}{+}$, $w^{2h}$ for ${+}{+}{-}$,
  $v^{2h}$ for ${+}{-}{-}$, $v^{6h}$ for ${-}{-}{-}$.
  For ${+}{+}{-}$ and ${+}{-}{-}$ this
  is below the lowest degree of the image of $\alpha$,
  by $4h$ in $R'$-degree, or $2h$ in $R^+$-degree.
  Given the quotient by $w_1+w_2+w_3$,
  the Lorentz modules of all $w^{2h}$ respectively $v^{2h}$ are as claimed
  using the natural Lorentz action on $R'$,
  since \eqref{eq:cdg} is Lorentz equivariant
  using the Lorentz module structure on $M$ in Remark \ref{remark:lorN}.
\qed\end{proof}
%%%%%%%%%%%%%%%%%%%%%%%%%%%%%%%%%%%%%%%%%%%%%%%%%%%%%%%%%%%%%%%%%%%%%%
\begin{remark} \label{rem:linv}
Note that 
$S^{2h} (\rep{0}{\tfrac12} \oplus \rep{0}{\tfrac12})
\simeq 
\bigoplus_{m+n=2h}S^{m}\rep{0}{\tfrac12}\otimes S^{n}\rep{0}{\tfrac12}
$ contains the trivial representation exactly once if $2h$ is even,
never if $2h$ is odd.
The trivial representation never appears in $S^{2h} \rep{\tfrac12}{0} \simeq \rep{h}{0}$ if $h>0$.
\end{remark}
\begin{lemma}[2-to-1 amplitude to-be] \label{lemma:lig}
If $h \in \Z_{\geq 0}$ then $\wt{M_h^{+,++-}}(X^+-Z_3)$
contains a unique Lorentz invariant element of
degree $4h$, up to normalization. It is given
for any local factorization $k_i = v w_i^T$
with $w_1 + w_2 + w_3 = 0$ by
\begin{subequations}\label{eq:exfboth}
\begin{equation}\label{eq:exf}
(w_1^T \eps w_2)^h\;v^{\otimes 2h} \otimes v^{\otimes 2h} \otimes w_3^{\otimes 2h}
\end{equation}
where {\recalleps}.
Analogous element in $\widetilde{M_h^{-,--+}}(X^--Z_3)$ using $k_i = v_iw^T$:
\begin{equation}
(v_1^T \eps v_2)^h\;w^{\otimes 2h} \otimes w^{\otimes 2h} \otimes v_3^{\otimes 2h}
\end{equation}
Analogous for $\wt{M^{+,\sigma}}$ if $\sigma$ has two plus,
and for $\wt{M^{-,\sigma}}$ if $\sigma$ has two minus.
\end{subequations}
\end{lemma}
\begin{proof}
The expression is independent of the local factorization,
so we get a section.
It is Lorentz invariant. It is unique by Remark \ref{rem:linv}.
\qed\end{proof}
%------------------------------------------------------------------
The next statement is a strict aside, we will not actually use it.
Our use of base change along $\C[k_i]/\Q_i \to R$
will be straightforward.
\begin{lemma}[Flat morphism]\label{lemma:flatm}
For all $N \geq 3$ the map $(k_1,\ldots,k_N) \mapsto k_i$ from $X - \{k_i = 0\}$
to the complement of the origin of the cone is a flat morphism.
Hence the inverse image functor for quasi-coherent sheaves,
which at the level of modules corresponds to\footnote{%
See Proposition II.5.2 (e) \cite{hbook}.} $R \otimes_{\C[k_i]/\Q_i}-$,
is exact there.
\end{lemma}
\begin{proof}
If $N \geq 4$ then the ring map $\C[k_i]/\Q_i \to R$ is already flat,
in fact $R$ is free as
a $\C[k_i]/\Q_i$-module with countable basis\footnote{%
For $i=N$,
use the $\C$-basis of $R$ associated to a Gr\"obner basis of $I_N$ 
with leading monomials $a_1,b_1,c_1,d_1,d_2a_3,b_2c_2,\ldots,b_Nc_N$
the last from $b_Nc_N - a_Nd_N$ (cf.~Lemma \ref{lemma:ci0}).
The main point is that at the level of leading monomials
we have a product situation,
all but one only involving $k_1,\ldots,k_{N-1}$, the last one only involving $k_N$.
}%
\textsuperscript{,}\footnote{%
Flatness at ring level fails for $N=3$
due to relations such as $a_3c_1 - c_3a_1$ in $I^+$.
Starting from
$(a_3,c_3) \hookrightarrow \C[k_3]/\Q_3$,
if we apply $R^+ \otimes_{\C[k_3]/\Q_3} - $
we fail to get an injection.}.
If $N \geq 3$ by `miracle flatness', since $R$ is Cohen-Macaulay,
the cone is smooth away from the origin,
and all fibers have the same dimension since the automorphism group of the cone acts transitively.
For $N=3$ the fibers are actually $\simeq \C^2$.
\qed\end{proof}
%------------------------------------------------------------------
\begin{definition}[The helicity module $C$] \label{def:hmc}
  For every $N \geq 3$ and every half-integer $h \geq \frac12$
  and every $\sigma \in \{-,+\}^N$:
  \begin{itemize}
   \item If $N \geq 4$,
   define for every $i$ the complex
   $\cc_{i,\pm h} = R \otimes_{\C[k_i]/\Q_i} \cc_{\pm h}$ of free $R$-modules,
   where the second factor $\cc_{\pm h}$ is
   the complex \eqref{eq:diffk} with $k$ replaced by\footnote{%
   Actually use $-k_N$ for $i = N$.
   Since $\cc_h$ (and $\cc_{-h}$) is isomorphic to itself
   with $k$ replaced by $-k$ simply by sign reversal in homological degree 2,
   we gloss over this distinction.
   } $k_i$
   and viewed as a complex of free $\C[k_i]/\Q_i$-modules.
   Set
  \begin{equation}\label{eq:ccxx}
        C_h^\sigma = H^{3-N}\big(
             \underbrace{%
               \Hom_R\big(
               \cc_{1,-\sigma_1 h} \otimes_R \cdots \otimes_R \cc_{N-1,-\sigma_{N-1} h},
               \cc_{N,\sigma_N h}\big)
              }_{\substack{
                \textnormal{the internal hom of chain complexes}\\
                \textnormal{itself a complex of free $R$-modules}}}
            \big)
  \end{equation}
  the chain maps of homological degree $3-N$ modulo trivial chain maps.
\item If $N=3$, make the same definition over $R^\pm$, denoted $C_h^{\pm,\sigma}$.
\end{itemize}
\end{definition}
%------------------------------------------------------------
\begin{remark}
  As a matrix, the differential $d_i$ of $\cc_{i,\pm h}$ is
  that of $\cc_{\pm h}$
  but with entries reinterpreted in $R$.
  The differential of $\cc_{1,\pm h} \otimes_R \cdots \otimes_R \cc_{N-1,\pm h}$
  is given, using the Kronecker product of matrices, by
  \[
    \dtot = \textstyle\sum_{i=1}^{N-1}
               (\pm \one)^{\otimes (i-1)} \otimes
               d_i \otimes \one^{\otimes (N-1-i)}
  \]
  Then chain maps are matrices $A$ with entries in $R$
  with $A\dtot = d_N A$, modulo the trivial chain maps of the form $A =  B\dtot + d_NB$.
  The homological grading restricts all matrices
  to a specific block structure.
If $f \in R$, $f \neq 0$
then $\wt{C}(D_f) = C_f$
is the corresponding quotient space of matrices with entries in $R_f$.
\end{remark}
%------------------------------------------------------------------
\begin{prop}[The helicity sheaves match on an open set]
  \label{prop:mphom}
  Suppose $h \geq \frac12$ and $N \geq 3$. 
  Set $Y = \{k_1 = 0\} \cup \cdots \cup \{k_N = 0\}$.
  \begin{itemize}
  \item Local claim:
     For every maximal ideal $\mx$ corresponding to a point in $X-Y$,
     there is over $R_\mx$ a factorization of $S^{2h}k_i^\pm$ analogous
     to the one in Lemma \ref{lemma:r1},
     respectively a contraction of the differential of $\cc_{i,\pm h}$
     analogous to the one in Lemma \ref{lemma:sph2}.
     Both are obtained by applying $\C[k_i]/\Q_i \to R$.
     The sheaves $\wt{C}$ and $\wt{M}$ are locally free of rank one on $X-Y$.
  \item Global claim:
  There is an isomorphism of locally free rank one sheaves
  \begin{equation} \label{eq:cmiso}
            \wt{C_h^\sigma}|_{X-Y}
            \;\simeq\;
            \wt{M_h^\sigma}|_{X-Y}
  \end{equation}
  induced by
  \eqref{eq:inv} and \eqref{eq:zzt}.
  For $N=3$ there are additional signs.
\end{itemize}
 \end{prop}
\begin{proof}
  Recall that for maps $f,f'$ or complexes $E,E'$ of finite dimensional vector spaces
  one has $\image (f \otimes f') \simeq \image f \otimes \image f'$
  and $H(E \otimes E') \simeq H(E) \otimes H(E')$
  and $H(\Hom(E,E')) \simeq \Hom(H(E),H(E'))$.
  To prove this one can choose factorizations of $f,f'$ analogous to Lemma \ref{lemma:r1}
  and contractions for $E,E'$ analogous to Lemma \ref{lemma:sph2}
  and construct explicit factorizations of $f \otimes f'$
  and contractions of $E \otimes E'$ and $\Hom(E,E')$.
  Clearly, the same arguments apply over $R_\mx$,
  given factorizations and contractions over $R_\mx$.
  In particular $\wt{C}$ and $\wt{M}$ are locally free on $X-Y$.
  The fact that $\wt{C}$ has rank one is because
  homological degree $3-N$ implies that the only contribution
  comes from a product of degree 1 elements giving a degree 2 element.
  The global claim is now clear given \eqref{eq:inv}, \eqref{eq:zzt}.
\qed\end{proof}

%% file: RecursiveCharacterization.tex
\newcommand{\recam}{{recursive amplitude}}
\section{The recursive characterization of amplitudes} \label{sec:recchar}

This section combines all previous
sections to construct amplitudes.
We introduce the notion of a {\recam},
Definition \ref{def:reca}.
Then there is a uniqueness result, the recursive characterization in
Theorem \ref{theorem:unique}.
Finally the existence result, namely a proof that the minimal model brackets
for YM and GR are {\recam}s, Theorem \ref{theorem:exists}.
\step
This section uses notation from all previous sections.
Often implicit are an integer $N$,
a tuple of signs $\sigma \in \{-,+\}^N$, and the helicity $h$.
Often the length $N_\sigma = |\sigma|$ is understood to determine $N$.
In this section $h=1$ or $h=2$,
though some things immediately generalize to other $h$.
We use these shorthands:
\begin{center}
\begingroup
\renewcommand{\arraystretch}{1.2}
\begin{tabular}{c||l}
\emph{shorthand} & \emph{full name}\\
\hline
$R$ & $R_N = \kr{N}/I_N$ or one of $R^{\pm} = \kr{3}/I^\pm$ if $N=3$\\
$X$ & $X_N = \zl(I_N)$ or one of $X^\pm = \zl(I^\pm)$ if $N=3$\\
$\mathfrak{P}$ & $\mathfrak{P}_N$, see Theorem \ref{theorem:pd4}\\
$P$ & $P_N$, see Definition \ref{def:pz}\\
$Z$ & $Z_N$, see Definition \ref{def:pz}\\
$M$ & $M_h^\sigma$ or one of
      $M^{\pm,\sigma}_h$ if $N=3$, see Definition \ref{def:hmm}\\
$C$ & $C_h^\sigma$ or one of
      $C_h^{\pm,\sigma}$ if $N=3$, see Definition \ref{def:hmc}
\end{tabular}
\endgroup
\end{center}

\step
We start with a preliminary for the internal Lie algebra for YM.
\begin{definition}[Internal Lie algebra]
For YM, $\ym$, is a finite-dimensional non-Abelian Lie algebra
together with an invariant nondegenerate symmetric $\C$-bilinear form $\ym \otimes \ym \to \C$.
It could be the Killing form if $\ym$ is semisimple.
\end{definition}
To keep the notation uniform,
we introduce a corresponding but trivial object for GR.
Both are denoted $\ux$,
a finite-dimensional vector space\footnote{%
Viewed as a trivial representation of the Lorentz group.
}
together with a nondegenerate symmetric $\C$-bilinear form
$\ux \otimes \ux \to \C$
and an element of $\ux^{\otimes 3}$:
\begin{equation} \label{table:ucube}
\begingroup
\renewcommand{\arraystretch}{1.3}
\begin{tabular}{c||c|c|c}
& $\ux$ & $\ux \otimes \ux \to \C$ & $\ux^{\otimes 3}$\\
\hline
YM, $h=1$ & a Lie algebra & invariant form & Lie bracket \\
GR, $h=2$ & $\C$ & multiplication & $1^{\otimes 3}$
\end{tabular}
\endgroup
\end{equation}
The isomorphism $\ux \simeq \ux^\ast$ induced by $\ux \otimes \ux \to \C$
is implicitly used in this section.
With this understanding,
for YM the element in $\ux^{\otimes 3}$ is
the Lie algebra bracket $\wedge^2 \ux \to \ux$
and it is totally antisymmetric.
For GR, $\ux^{\otimes 3}$ is totally symmetric.
\step
Introduce the additional notation
\[
\mext = M \otimes \ux^{\otimes N}
\qquad\qquad
\cext = C \otimes \ux^{\otimes N}
\]
One must think of each factor $\ux$ as associated to one momentum, 
hence one factor in \eqref{eq:mat}, \eqref{eq:ccxx}.
For example, it is implicit that the action of the permutation group $S_N$ on sections
also permutes the $\ux$ factors.
For $\cext$ one might want to use
$\ux^\ast$ for the first $N-1$ factors,
but we are already using the $\ux \simeq \ux^\ast$ convention.
As $R$-modules,
$\mext$ and $\cext$
are direct sums of finitely many copies of $M$ and $C$.
\step
See Appendix \ref{appendix:m2s} for the module to sheaf functor $\wt{\phantom{x}}$.
\step

Recall that the minimal model brackets are defined when all
internal lines are off-shell, which is $X-P$.
But given how $1/\Q$ singularities appear in homotopies,
the minimal model brackets will be seen to extend uniquely to sections on $X-Z$, 
of a sheaf constructed using the effective divisor $D$ that 
allows first order poles along the prime divisors
in $\mathfrak{P}$\footnote{%
	$R$ is normal Noetherian domain, hence Weil divisors are defined since for every height one prime ideal,
	the local ring is a DVR.
	Set
	$R(D)= \{ f\in\Frac(R)\mid  D+\divv f \ge 0 \}$.
	Then $\O_X(D) = R(D)\wt{\phantom{x}}$ and if
	$D$ is effective then there is a canonical inclusion $\O_X\hookrightarrow \O_X(D)$.
}. That is, they will be seen to be in the image of the injective (see Remark \ref{rem:0thcom}) restriction map
\[ 
 \big(\wt{\mext} \otimes \O_X(D)\big)(X-Z)\; \hookrightarrow\; \wt{\mext} (X-P)
 \]
We will directly define and construct amplitudes on the left hand side.
This is suitable for the recursive characterization which is all about residues across $P$.
\begin{definition}[An effective divisor]
  For $N \geq 4$ define the Weil divisor
  \[
     D = \textstyle\sum_{\px \in \mathfrak{P}} \px
  \]
  on the normal Noetherian domain $R$.
  If $N = 3$ then set $D = 0$.
\end{definition}
\begin{remark}[No 0th local cohomology]\label{rem:0thcom}
Suppose $\mathcal{F}$ is an $\O_X$-sheaf, locally free on $X-Z$, for example $\O_X(D)$.
Then every section on $X-Z$ with support on $P-Z$ is zero.
Hence restriction $\mathcal{F}(X-Z)\to \mathcal{F}(X-P)$ is injective.
\end{remark}

To define what a {\recam} is,
we have to define residues and a fusion operation.
These are only used away from $Z$,
where the variety is smooth,
$D$ is smooth, and
$\wt{\mext}$ is locally free,
by Proposition \ref{prop:cig}.
So there is no point in being overly technical
when defining residues and fusion.
\step
The residue along $\px \in \mathfrak{P}$
is the failure of a section of
$\O_X(\px)$ to be a section of $\O_X$.
Defined this way,
the residue does not depend on the choice of an equation
that locally defines $\zl(\px)$,
but it requires introducing another sheaf:
\[
0 \to \O_X|_{X-Z} \to \O_X(\px)|_{X-Z} \to \nsheaf_\px \to 0
\]
This short exact sequence of sheaves defines $\nsheaf_\px$,
the normal sheaf along $\zl(\px) - Z$. 
It is a locally free sheaf of rank one on $\zl(\px)-Z$.
\begin{definition}[Residue] \label{def:residue}
If $\mathcal{F}$ is an $\O_X$-sheaf, locally free on $X-Z$, then
\[
\Res_\px \;:\; \big(\mathcal{F} \otimes \O_X(D)\big)(X-Z) \;\to\;
    \big(\mathcal{F}|_{\zl(\px)-Z} \otimes \nsheaf_\px\big)(\zl(\px) - Z)
\]
is defined in the obvious way.
The common kernel of all $\Res_\px$ is $\mathcal{F}(X-Z)$.
\end{definition}
%--------------------------------------------------------------------
\begin{remark} 
\label{remark:respppp}
Recall that $\px_{++++}$ is a minimal prime over all three $(\Q_{ij})$.
The relative residues of the $\Q_{ij}$ will play an interesting role. They are
\[
    \eps_{12}\eps_{34} \Res_{\px_{++++}} \frac{1}{\Q_{12}}
    \;=\;
    \eps_{31}\eps_{24} \Res_{\px_{++++}} \frac{1}{\Q_{13}}
    \;=\;
    \eps_{23}\eps_{14} \Res_{\px_{++++}} \frac{1}{\Q_{23}}
\]
where $k_i = v w_i^T$
and $w_1 + w_2 + w_3 + w_4 = 0$
and $\eps_{ij} = w_i^T \eps w_j$ and \recalleps\footnote{%
In fact denoting $k_i = v_i w_i^T$
then
$\eps_{12}\eps_{34} \Q_{13} = \eps_{31} \eps_{24} \Q_{12}$
holds in $\vwr{4}/\JX_4$.
}.
\end{remark}
%R = QQ[v1_1..v1_4,v2_1..v2_4,w1_1..w1_4,w2_1..w2_4];
%eW = (i,j) -> w1_i*w2_j-w1_j*w2_i;
%eV = (i,j) -> v1_i*v2_j-v1_j*v2_i;
%k = i -> matrix{{v1_i*w1_i,v1_i*w2_i},{v2_i*w1_i,v2_i*w2_i}};
%Q = (i,j) -> det(k(i)+k(j))-det(k(i))-det(k(j));
%-- Q = (i,j) -> eV(i,j)*eW(i,j); -- same
%I = ideal flatten entries sum for i from 1 to 4 list k(i);
%S = R/I;
%eW(1,2)*eW(3,4)*Q(1,3)-eW(3,1)*eW(2,4)*Q(1,2)
%eW(1,2)*eW(3,4)*Q(2,3)-eW(2,3)*eW(1,4)*Q(1,2)
\step

Let $N \geq 4$, let $J$ be an internal momentum and
$\px \supset (\Q_J)$ a prime divisor.
Fusion is a bilinear operation that multiplies two amplitudes,
for $|J| + 1$ and $1+|J^c|$ particles
respectively,
and produces a new amplitude
for $N$ particles that is however only defined along a fiber product.
Here are the details:
\begin{itemize}
\item 
Geometrically, fusion is based on the isomorphism of varieties
\[
        X_{|J|+1} \times_{\an{4}} X_{1+|J^c|}
        \;\;\simeq \;\; \zl(\px)
\]
The fiber product and its identification with $\zl(\px) \subset X_N$ are defined by
%%%%%%%%%%%%%%%%%
\begin{equation}\label{eq:fibfuse}
\begin{tikzpicture}[baseline=(current  bounding  box.center)]
\matrix (m) [matrix of math nodes, column sep = 25mm, row sep = 8mm,minimum width=3mm,minimum height=3mm ]
{
	\zl(\px)\vphantom{_{1+|J^{c}|}} &  X_{1+|J^{c}|}\vphantom{(\px)}  \\
	X_{|J|+1}\vphantom{^{4}} & \an{4}\vphantom{_{|J|+1}}   \\
};
\path[-stealth]
(m-1-1) edge [ ]  (m-1-2)
(m-1-1) edge [ ]  (m-2-1)
(m-2-1) edge [ ] node [above] {\footnotesize{$  ((k_i)_{i \in J},k) \mapsto k   $}} (m-2-2)
(m-1-2) edge [ ] node [right] {\footnotesize $ (-k,(k_i)_{i \notin J}) \mapsto k   $} (m-2-2);
\end{tikzpicture}
\end{equation}
This is for $|J|, |J^c| \geq 3$, otherwise see Remark \ref{rem:special}.
Both factors restrict the
distinguished momentum $k = (a,b,c,d)$ to
the cone $\zl(ad-bc) \subset \an{4}$.
\end{itemize}
%---------------------------------------------------
Given an amplitude on each factor,
pull them back along the projection morphisms
$X_{|J|+1} \leftarrow \zl(\px) \to X_{1+|J^c|}$
and take their tensor product,
then reorder the $\image S^{2h} k^\pm_i \otimes \ux$ factors.
The two superfluous $\image S^{2h} k^\pm \otimes \ux$
associated to the distinguished momentum $k$,
one from each amplitude,
are annihilated:
\begin{itemize}
\item The $\image S^{2h} k^\pm$ factors are annihilated using \eqref{eq:inv}.
So fusion is only defined if one is $\image S^{2h}k^+$
and the other is $\image S^{2h}k^{-}$.
Effectively one cancels
\begin{equation}\label{eq:cancelks}
S^{2h}\bigg(\begin{smallmatrix} a \\ b \\ c \\ d \end{smallmatrix}\bigg)
\end{equation}
Fiberwise $v^{\otimes 2h} \otimes w^{\otimes 2h} \mapsto 1$ at $k = vw^T$,
in the sense of Remark \ref{remark:vbp}.
\item The $\ux$ factors
are annihilated using $\ym \otimes \ym \to \C$.
\end{itemize}
%----------------------------------------------------
\begin{definition}[Fusion]
Let $N_\sigma \geq 4$, let $J$ be an internal momentum and
$\px \supset (\Q_J)$ a prime divisor.
The above discussion defines a $\C$-bilinear map
\[
    \fuse{\px,J} \;:\;
\wt{\mext^{(\sigma_J,\zeta)}}(X_{|J|+1}-P) \times 
\wt{\mext^{(-\zeta,\sigma_{J^c})}}(X_{1+|J^c|}-P)\\
\;\to\;
\wt{\mext^{\sigma}}|_{\zl(\px)-Z}(\zl(\px) - Z)
\]
with 
$J$ and $J^c$ sorted in ascending order, 
$\sigma_J = (\sigma_i)_{i \in J}$, $\sigma_{J^c} = (\sigma_i)_{i \notin J}$, $\zeta = \pm$.
\end{definition}

Since the output is taken in $\zl(\px) - Z$,
and by the definition of $Z$,
it follows immediately that 
the input only requires sections on $X - P$.
When we use fusion in \eqref{eq:resrule},
restriction of the arguments from $X-Z$ to $X-P$ is implicit.
%------------------------------------------
\begin{remark}[Special rules] \label{rem:special}
	If $|J| = 2$ replace the lower left factor in \eqref{eq:fibfuse} by one of $X^\pm$,
	 if $|J^c| = 2$ replace the upper right factor by one of $X^\pm$.
	Which of $X^\pm$
	is determined by requiring that the result be in $\zl(\px)$.
	So if $J = \{1,2\}$:
	\begin{itemize} 
		\item For $\px = \px_{++++}$ use $X^+ \times_{\an{4}} X^+$.
		\item For $\px = \px_{++--}$ use $X^+ \times_{\an{4}} X^-$.
		\item For $N\ge5$ and $\px = \px_{12}^+$ use $X^+ \times_{\an{4}} X$.
	\end{itemize}	
	All other cases are analogous.
\end{remark}
%----------------------------------------------------
\begin{definition}[Recursive amplitudes] \label{def:reca}
For $h=1$ or $h=2$,
by a {\recam} we mean a collection of objects $(B^\sigma)$,
one for every $\sigma \in \{-,+\}^{N_\sigma}$
for all $N_\sigma \geq 3$, such that:
\begin{itemize}
\item If $N_\sigma = 3$
then $B^\sigma = B^{+,\sigma} \oplus B^{-,\sigma}$ where
\[
B^{\pm,\sigma} \in \wt{\mext^{\pm,\sigma}}(X^\pm-Z_3)
\]
Nonzero are only
$B^{+,\sigma}$ with $\sigma$ any permutation of ${+}{+}{-}$,
and
$B^{-,\sigma}$ with $\sigma$ any permutation of ${-}{-}{+}$.
They are up to normalization the elements in Lemma \ref{lemma:lig} 
times the given element in $\ux^{\otimes 3}$ in \eqref{table:ucube}.
Furthermore
\begin{equation}\label{eq:nonzero3}
B^{+,+-+}\neq 0
\qquad
B^{-,-+-}\neq 0
\end{equation}
\item If $N_\sigma \geq 4$ then
\[
B^{\sigma} \in (\wt{\mext^\sigma} \otimes \O_X(D))(X-Z)
\]
is a homogeneous element 
whose degree is $2h + 2(N_\sigma-3)$
lower than the generators of $\mext^\sigma$.
Furthermore,
for all prime divisors $\px \in \mathfrak{P}$:\footnote{%
The condition $N_\sigma \notin J$ has a double purpose.
First, it excludes duplicates due to
$\Q_J = \Q_{J^c}$.
Using permutation invariance, 
the summand is invariant under replacing $J$ by $J^c$ and
	$\zeta$ by $-\zeta$.
	The distinguished momentum $k$ changes sign,
	but there is no sign ambiguity 
	when we cancel \eqref{eq:cancelks}
	since $h$ is an integer.
Second, it is required for pre-amplitudes in Remark \ref{rem:pre}.}
\begin{equation}\label{eq:resrule}
\Res_\px B^\sigma
= \sum_{\substack{J:\;\px \supset (\Q_J)\\ N_\sigma \notin J}} \sum_{\zeta = \pm}
\Big(B^{(\sigma_J,\zeta)} \fuse{\px,J} B^{(-\zeta,\sigma_{J^c})}\Big)
\otimes \Res_\px \frac{1}{\Q_J}
\end{equation}
If $|J|=2$ or $|J^c|=2$ one must observe the special rules in Remark \ref{rem:special}.
\item For all $N\ge3$,
the collection $(B^\sigma)_{N_\sigma=N}$ is $S_N$ permutation invariant.
\end{itemize}
\end{definition}
%%%%%%%%%%%%%%%%%%%%%%%%%%%%%%%%%%%%%%%%%%%%%%%%%%%%%%%%%5
\begin{remark}[Normalization] \label{rem:nzo}
If $(B^\sigma)$ is a recursive amplitude
then so is 
\[ \big(\lambda^{N_\sigma-2}\mu^{N_{\sigma}^{+}-N_{\sigma}^{-}}B^\sigma\big) \] 
for all $\lambda,\mu \in \C^\times$,
where $N_{\sigma}^{\pm}$ is the number of plus and minus 
signs in $\sigma$.
Up to such transformations, any two recursive amplitudes have
identical $(B^{\sigma})_{N_{\sigma}=3}$.
\end{remark}
\begin{remark}[Degree of homogeneity]
Using the grading on $M$ in Definition \ref{def:hmm},
amplitudes have degree 
$d_N = 4+(N-1)(2h-2)$.
This is compatible with \eqref{eq:resrule},
because fusion cancels \eqref{eq:cancelks}
and is homogeneous of degree $-2h$,
and $d_N = d_{|J|+1} + d_{1+|J^c|}-2h-2$.
Incidentally, the last equation has a unique solution for every $d_3$, and
$d_3=4h$ is the lowest allowed by Lemma \ref{lemma:lig}.
\end{remark}
%------------------------------------------
\begin{theorem}[Recursive characterization; uniqueness] \label{theorem:unique}
Given $h$, $\ux$. Then
any two {\recam}s are equal, up to normalization as in Remark \ref{rem:nzo}.
\end{theorem}
\begin{proof}
Let $\delta B^\sigma$ be the discrepancy.
It vanishes for $N_\sigma = 3$,
and by induction for $N_\sigma \geq 4$.
In fact, by the induction hypothesis,
and since the common kernel of all $\Res_\px$ are regular elements,
we have $\delta B^\sigma \in \wt{\mext}(X-Z)$.
Since $Z$ has codimension $\geq 2$,
Theorem \ref{theorem:nlcm} implies $\delta B^\sigma \in \mext$.
But $\delta B^\sigma$ is homogeneous of degree lower than
the generators of $\mext$, hence zero.
\qed\end{proof}
%------------------------------------------
\begin{lemma}[The impossibility of a single residue] \label{lemma:isr}
For all $N_\sigma \geq 4$ and all $\px \in \mathfrak{P}$,
a homogeneous element of $(\wt{\mext^\sigma} \otimes \O_X(\px))(X-Z)$
of degree at least $3$ 
lower than the generators of $\mext^\sigma$ is zero.
\end{lemma}
\begin{proof}
Let $B$ be such an element. Every $\px$ is minimal over some $(\Q_J)$,
so $\Q_J B \in \wt{\mext}(X-Z) \simeq \mext$, using Theorem \ref{theorem:nlcm}.
But $\Q_J B$ is homogeneous of degree lower than the
generators of $\mext$, hence zero. By Remark \ref{rem:0thcom}, $B = 0$.
\qed\end{proof}
%-----------------------------------------------
%---------------------------------------------
\begin{prop}[Helicity violation and qualitative properties] \label{corollary:hcqp}
Every {\recam} satisfies:
\begin{itemize}
\item
If $N_\sigma \geq 4$ 
then $B^\sigma = 0$ if $\sigma$ contains fewer than two $+$ or two $-$.
\item
If $N_\sigma \geq 5$ then $B^{--++\cdots+}$,
known as MHV amplitudes\footnote{%
For YM
 these are explicitly given, color-ordered, by the
 Parke-Taylor formula.
 Quote from \cite{parketaylor}:
\emph{The Feynman diagrams for [\ldots] gluon scattering
contain propagators $(p_i+p_j+p_k)^2$, $(p_i+p_j+p_k+p_m)^2$, \ldots
These propagators must cancel for [our equation] to be correct.}},
can have nonzero residue only along the $\px_{ij}^+$.
The residue along $\px_{12}^+$ is also zero.
\end{itemize}
\end{prop}
\begin{proof}
By \eqref{eq:resrule} the amplitude $B^{-+++}$
can have residue only along $\px_{++++}$,
so it is zero by Lemma \ref{lemma:isr} since its degree is $2h+2$ below the generators of $\mext$.
This implies the first claim for $N_\sigma = 4$,
for all $N_\sigma$ by induction using \eqref{eq:resrule}. 
The second claim is also by \eqref{eq:resrule}, 
since here the sum over $J$, $\zeta$ always degenerates to at most one term
and fusion vanishes iff both fusion factors vanish.
\qed\end{proof}
%-------------------------------------------
In Theorem \ref{theorem:exists} below we construct recursive amplitudes 
using the minimal model brackets. Hence $S_{N-1}$ permutation invariance 
is by construction, but by the next argument this 
 implies $S_{N}$ permutation invariance.
\begin{remark}[Recursive pre-amplitudes are recursive amplitudes]\label{rem:pre}
Define a recursive pre-amplitude just like a recursive 
amplitude in Definition \ref{def:reca}, but
with $S_N$ permutation invariance relaxed 
to $S_{N-1}$ permutation invariance in the first $N-1$ factors,
indicated by a vertical bar below.
In particular, \eqref{eq:resrule} remains in force verbatim; 
the fact that $N$ is always the rightmost element of $J^{c}$ is now critical.
For $N=3$, permutation invariance
can only fail due to a mismatch of normalizations.
For the
$(B^{+,\sigma})_{N_\sigma = 3}$
define normalization constants $n_\pm \in \C$ by
\[
B^{++|-} = n_- B^{++-}
\qquad
B^{+-|+} = n_+ B^{+-+}
\]
where the $B$ on the left are those of a given pre-amplitude,
those on the right are $S_3$ permutation invariant and nonzero 
and used as a reference.
By \eqref{eq:nonzero3} for the pre-amplitude we have $n_+ \neq 0$.
We must show $n_- = n_+$.
Abbreviate 
\begin{align*}
A&=\big(B^{-++}  \fuse{\px,12} B^{-++}\big)\otimes \Res_{\px}\tfrac{1}{\Q_{12}}\\
A'&=\big(B^{-++} \fuse{\px,13} B^{-++}\big)\otimes \Res_{\px}\tfrac{1}{\Q_{13}}\\
A''&=\big(B^{++-} \fuse{\px,23} B^{+-+}\big)\otimes \Res_{\px}\tfrac{1}{\Q_{23}}
\end{align*}
where $\px = \px_{++++}$.
By \eqref{eq:resrule} for the pre-amplitude and Lemma \ref{lemma:isr} we have 
$B^{-++|+}=0$,
hence $n_+^2 A+n_+^2 A'+n_-n_+ A''=0$ by \eqref{eq:resrule}.
On the other hand we have 
$A+A'+A''=0$ by direct calculation, using 
the relative residues in Remark \ref{remark:respppp}, 
the expressions in Lemma \ref{lemma:lig}, the given
$U \in \ux^{\otimes 3}$, and:
\begin{itemize}
	\item For YM, the Jacobi identity for the Lie algebra bracket $U$.
	\item For GR, 
	$\eps_{12}\eps_{34} + \eps_{23}\eps_{14} + \eps_{31}\eps_{24} = 0$ and $U=1$.
\end{itemize}
Since $n_+\neq 0$ and $A'' \neq 0$ we conclude that $n_- = n_+$.
This establishes $S_3$ permutation invariance for $(B^{+,\sigma})_{N_\sigma = 3}$,
the case $(B^{-,\sigma})_{N_\sigma = 3}$ is analogous.
Now $S_{N}$ permutation invariance for $N\ge4$ is by induction using \eqref{eq:resrule}.
\end{remark}
%-------------------------------------------
\newcommand{\xblobx}[2]{%
  \draw [black] (#1) node {#2};}
\newcommand{\plusminus}[2]{\draw [shorten >=\xshort,shorten <=\xshort]
	(x) -- (y) node[midway,xshift=-7,yshift=-0] {#1}  node[midway,xshift=-2,yshift=6] {#2} ;}
\newcommand{\brr}[1]{\draw [black] (#1) circle (0.08); \draw [black] (#1) circle (0.04);}
\newcommand{\xres}[5]{\wrap{1}{
  \coordinate (a) at (-1,0);
  \coordinate (b) at (0,0);
  \coordinate (c) at (1,0);
  \coordinate (d) at (0,2.05);
  \coordinate (x) at (-0.5,0.66);
  \coordinate (y) at (0,1.33);
  \haux{x}{a}{shorten >=9}{}{}
  \haux{x}{b}{shorten >=9}{}{}
  \haux{y}{c}{shorten >=9}{}{}
  \haux{y}{d}{shorten >=9}{}{}
  \plusminus{$#4$}{$#5$}
  \xblobx{a}{$#1$}
  \xblobx{b}{$#2$}
  \xblobx{c}{$#3$}
  \xblobx{d}{$+_4$}
  \brr{x}
  \brr{y}
}}
\newcommand{\xrat}[3]{\wrap{1}{
  \coordinate (a) at (-0.5,0);
  \coordinate (b) at (0.5,0);
  \coordinate (c) at (0,1.15);
  \coordinate (x) at (0,0.5);
  \haux{x}{a}{shorten >=8}{}{}
  \haux{x}{b}{shorten >=8}{}{}
  \haux{x}{c}{shorten >=8}{}{}
  \xblobx{a}{$#1$}
  \xblobx{b}{$#2$}
  \xblobx{c}{$#3$}
  \brr{x}
}}
%------------------------------------------
The upcoming theorem 
finally connects amplitudes and minimal model brackets.
It can be viewed as an existence theorem for {\recam}s,
but beware that if existence is the only goal,
there can be more direct constructions,
perhaps using \cite{csw} or \cite{arkanietal},
but this is not a direction we pursue.
%----------------------------------------------
\begin{definition}[Homotopy data]
Given $h$, $\ux$.
Take any dgLa as in Theorem \ref{theorem:ymgrnew},
viewed as a vector space $\vecx$ with differential polynomial in $k$.
Denote by $d_i$ respectively $d_J$ the differential on $\vecx$,
as a function of $k_i$ respectively $k_J = \sum_{i \in J} k_i$.
Homotopy data at a maximal ideal $\mx \subset R$ corresponding to
a point $q = (q_1,\ldots,q_{N}) \in X-Z$ consists of:
 \begin{itemize}
	\item \emph{Off-shell homotopy for internal tree lines:}
	For every $J \subset \mbox{\{1,\ldots,N-1\}}$ with $1 < |J| < N-1$
	a matrix $\hoff_J$ with entries depending only on $k_J$
	that satisfy $(\hoff_J)^2 = 0$ and $\hoff_Jd_J + d_J\hoff_J = \one$:
	\begin{itemize}
		\item If $Q_{J}(q) \neq 0$, 
		use a trivial homotopy as in Lemma \ref{lemma:trivhom}.
		\item If $Q_{J}(q) = 0$,
		use an optimal homotopy based on Lemma \ref{lemma:sphdgLa}.
	\end{itemize}
	Its entries are in $R_\mx$,
	except for the explicit $1/\Q_J$ if $\Q_{J}(q)=0$.
	\item \emph{On-shell contraction for external tree lines:}
	For every $i=1\ldots N$ a contraction
	by matrices $(h_i,i_i,p_i)$ as in Lemma \ref{lemma:sphdgLa},
	depending only on $k_i$. 
	Its entries are in $R_\mx$. It satisfies
	$h_id_ih_i = h_i$, $(h_i)^2 = 0$, $i_ip_i = \one - h_id_i - h_id_i$, $p_ii_i = \one$,
	and (since $\Q_i=0$)
	$d_ih_id_i = d_i$.
\end{itemize}
\end{definition}
Every
invocation of Lemma \ref{lemma:sphdgLa} yields a diagram \eqref{eq:sqhdgLa}. 
Hence we can freely transition between
$\vecx$ and $\cc$ in the next theorem.
Here $\cc$ is a shorthand for \eqref{eq:hxx}.
%----------------------------------------------
\begin{theorem}[Amplitudes from minimal model brackets; existence] \label{theorem:exists}
  Given $h$, $\ux$.
  Take any dgLa as in Theorem \ref{theorem:ymgrnew}.
  Then there exists a {\recam} $(B^\sigma)$ with the following property:
  \begin{quote}
  For every $\sigma$ with $N_\sigma \geq 3$\footnote{%
  And if $N_\sigma = 3$, an additional sign to pick one of $X = X^\pm$.},
  every maximal  $\mx \subset R$ corresponding to
  a point in $X-Z$, and every choice of homotopy data at $\mx$,
 the chain map in Theorem \ref{well-defined2} induces, via \eqref{eq:cmiso},
  the section $B^\sigma$ over $R_\mx$.
	\end{quote}
\end{theorem}
%---------------------------------------------------
\begin{proof}
By Remark \ref{rem:pre} it suffices to construct a recursive pre-amplitude.
By \eqref{eq:cmiso}
it suffices to construct corresponding sections $(b^\sigma)$
that instead of $\mext$ use $\cext$.
In any given tree at most one off-shell homotopy with a $1/\Q_J$ appears\footnote{%
If $N_\sigma = 4$ then three off-shell homotopies
can have a $1/\Q_J$,
but never in the same tree.}
and this goes into $\O_X(D)$.
The chain map in Theorem \ref{well-defined2}\footnote{%
   One must be careful with the signs of the momenta.
   But $\cc_h$ (respectively $\cc_{-h}$) is isomorphic to itself
   with $k$ replaced by $-k$ simply by sign reversal in homological degree 2.
   Hence we gloss over this distinction just as we did in the definition of $C$.
}
induces a
$b_\mx \in (\wt{\cext} \otimes \O_X(D))(D_{f_\mx})$,
where $f_\mx \notin \mx$
is a common multiple of all denominators
excluding the $1/\Q_J$ already taken care of.
Here $D_{f_\mx} = \{f_\mx \neq 0\}$.
Given homotopy data at two maximal ideals $\mx$, $\nx$ then
on $D_{f_\mx} \cap D_{f_\nx} = D_{f_\mx f_\nx}$
the elements $b_\mx$, $b_\nx$
differ by a trivial chain map by Theorem \ref{well-defined2}\footnote{%
There is a mismatch of the current setting with Theorem \ref{well-defined2},
which requires a possibly discontinuous contraction globally on $\C^4$,
whereas here we have one for every $i$, $J$.
Theorem \ref{well-defined2} extends to this case
 by continuity, or by
inspection of the proof of Theorem \ref{well-defined2}.}\textsuperscript{,}\footnote{%
Theorem \ref{well-defined2}
requires that internal lines be off-shell,
which fails along the subvariety $P-Z$.
We addressed this earlier.
}.
Gluing yields a section $b$ on $X-Z$,
and all choices of homotopy data yield the same $b$.

The homogeneity in Theorem \ref{theorem:ymgrnew}
means that the $(b^\sigma)$ have $R$-degree zero
if the module for chain maps in \eqref{eq:ccxx}
is taken to be\footnote{By convention,
$\cc^1 \simeq R^{2h+1}$, $\cc^2 \simeq R^{4h}$, $\cc^3 \simeq R^{2h-1}$
as graded $R$-modules, so they are generated in degree zero.
We denote by $\cc^1[2]$ the shifted module
generated in degree $-2$, etc.
}
\[
\Hom(\cc^1[2] \otimes \cdots \otimes \cc^1[2],\cc^2[3])
\]
The two isomorphisms in \eqref{eq:zzt} have degree zero using
$(\image S^{2h} k)[j]$ and $\cc^1[j]$ and $\cc^2[j-1]$ for all $j$\footnote{%
Note that the differential in Lemma \ref{lemma:uniqd}
has the same degree as $d/\Q$.}.
The isomorphism \eqref{eq:inv} is induced by a degree zero map
$(\image S^{2h}k^\pm)[2h-2] \to ((\image S^{2h} k^\mp)[2])^\ast$, with $\ast$ the dual.
Combining, the $(B^\sigma)$ have degree zero
using the grading $M[4+(N-1)(2h-2)]$, as claimed.

One can now check all properties of a recursive pre-amplitude.
The $N = 3$ pre-amplitudes are as required,
compatible with Lemma \ref{lemma:lig}, 
since the degrees match
and by Lorentz invariance in Theorem \ref{theorem:ymgrnew}.
In particular we have \eqref{eq:nonzero3} by Theorem \ref{theorem:ymgrnew}.
For $N \geq 4$, the residue \eqref{eq:resrule}
at $\px$
is checked at a given maximal $\mx \supset \px$ by choosing homotopy data and using:
\begin{itemize}
\item The structure of optimal homotopies \eqref{eq:opthom}\footnote{%
The distinction between, for example,
$\hprime^\vecx$ and $\hprime^{\cc}$ is insignificant using \eqref{eq:sqhdgLa}.
}\textsuperscript{,}\footnote{%
	Note that $i{\hprime}p$ induces a map on homology inverse to $i{\dprime}p$
	in Lemma \ref{lemma:uniqd}. 
	We get \eqref{eq:resrule} by construction 
	since \eqref{eq:cmiso} is defined using
	\eqref{eq:inv} and \eqref{eq:zzt}, and the latter is induced by $i{\dprime}p$.
}.
\item Tree combinatorics\footnote{%
There is at most one $1/\Q_J$ in any given tree.
The sum over all trees containing $J$ as an internal line is a  
double sum over all trees below and all trees above $J$.
The internal line $J$ is decorated by the optimal homotopy.
}.
\end{itemize}
The sum over $\zeta = \pm$ in \eqref{eq:resrule}
is because $\cc$ is a direct sum of two complexes.
Permutation invariance in all but the last factor is by construction.
\qed\end{proof}
It would be interesting to study qualitative properties of the amplitudes
near $Z$ in codimension $\geq 2$.
Perhaps local cohomology calculations
yield interesting constraints
only based on the structure of the sheaves.
The case where one $k_i$ approaches zero
is known as the soft gluon respectively soft graviton limit.

%% file: ModuleToSheaf.tex
\section{On the module to sheaf functor} \label{appendix:m2s}

Throughout this appendix:
\begin{quote}
\emph{$R$ is 
the coordinate ring of an irreducible affine $\C$-variety $X$.
Hence
 $R$ is a reduced affine $\C$-algebra, Noetherian, and an integral domain.}
\end{quote}
For every nonzero $f \in R$, the localization $R_f$ inherits these properties
and it is the coordinate ring of the irreducible subvariety
$D_f = \{x \in X \mid f(x) \neq 0\}$.
Every localization of $R$ is a subring of the
 field of fractions $\Frac(R)$.
%---------------------------------------------------------------------
\begin{lemma}[The module to sheaf functor] \label{lemma:m2s}
For every finitely generated $R$-module $M$
there exists a unique $\O_X$-sheaf $\wt{M}$
with $\wt{M}(D_f) = M_f$ for all nonzero $f\in R$,
and canonical restriction maps.
Also $\wt{R} = \O_X$,
and $\wt{\phantom{x}}$ is an exact functor from finitely generated modules to 
coherent sheaves, and
\begin{align*}
\wt{M} \otimes_{\O_X} \wt{N} & = (M \otimes_R N)\wt{\phantom{x}}\\
\sheafHom_{\O_X}(\wt{M},\wt{N}) & = \Hom_R(M,N)\wt{\phantom{x}}
\end{align*}
We also have $\wt{M}|_{D_f} = \wt{M_f}$.
\end{lemma}
\begin{proof}
See Chapter I, Section 1.3 \cite{ega}
or Chapter II, Section 5 \cite{hbook}.
\qed\end{proof}
%--------------------------------------------------------
The second half of Lemma \ref{lemma:m2s}
can alternatively be taken to define
the tensor product,
the $\sheafHom$, and via exactness
the kernel and cokernel and image of a morphism,
for such sheaves.
These operations are local
in the sense that they commute with restriction to $D_f$,
say $(M \otimes_R N)_f = M_f \otimes_{R_f} N_f$ implies
\[
(\wt{M} \otimes_{\O_X} \wt{N})|_{D_f} = \wt{M}|_{D_f} \otimes_{\O_{D_f}} \wt{N}|_{D_f}
\]
%-----------------------------------------------------------

A sheaf $\wt{M}$ is locally free of rank one
if every point has an open neighborhood $D_f$ such that $\wt{M}|_{D_f} \simeq \O_X|_{D_f}$,
that is $M_f \simeq R_f$.
Local freeness on an open subset
is preserved by $\sheafHom$ and by the tensor product of sheaves.
It is not preserved say by cokernels,
cf.~the skyscraper sheaf.

%% file: FiberProduct.tex
\section{Irreducibility of a fiber product} \label{sec:fib}

The goal of this appendix is Lemma \ref{lemma:fib3}.
We give two proofs,
a longer proof via Lemmas \ref{lemma:fib} and \ref{lemma:fib2}
that conveys a geometric picture, and a direct proof.

\begin{lemma} \label{lemma:fib}
Suppose $X,Y$ are irreducible affine $\C$-varieties, and
\[
        X \xrightarrow{\;\;f\;\;} \an{m} \xleftarrow{\;\;g\;\;} Y
\]
are morphisms of varieties. Suppose
$Y$ is Cohen-Macaulay,
$g$ is surjective,
every fiber of $g$ has dimension $\dim Y - m$,
and there is an open dense subset $U \subset X$
      such that the fiber $g^{-1}(f(x)) \subset Y$
      is irreducible for all $x \in U$.
Then
\[
X \times_\an{m} Y = \{(x,y) \in X \times Y \mid f(x) = g(y) \}
\]
is an irreducible variety.
\end{lemma}
\begin{proof}
The morphism $g$ is flat
since $Y$ is Cohen-Macaulay,
$\an{m}$ is smooth,
and the fibers of $g$ have dimension $\dim Y - m$;
see Exercise III.10.9 in \cite{hbook}.
This is sometimes referred to as `miracle flatness'.
The morphism $g$ is locally of finite presentation,
because the coordinate ring of $Y$
admits a finite presentation
as (via $g$) a $\C[z_1,\ldots,z_m]$-algebra. 
Flat and locally of finite presentation
implies that every base change of $g$ is open,
by \cite{ega} IV.2, 2.4.6.
So base change by $f$
yields an open map $X \leftarrow X \times_\an{m} Y$.
By topology, % e.g. Stacks Lemma 5.8.12,
if $X \leftarrow A$ is continuous and open,
if $X$ is irreducible,
and if there is an open dense subset
of $X$ such that the corresponding
fibers are irreducible, then $A$ is irreducible.
\qed\end{proof}
%-------------------------------------------------------------------
In the following,
each $k_i$ has four components and each $v_i$ and $w_i$ has two components.
We denote $k_i = (a_i,b_i,c_i,d_i)$ and
$\Q_i = a_id_i - b_ic_i$.
\begin{lemma}[Irreducibility] \label{lemma:fib2}
Suppose $X$ is an irreducible affine $\C$-variety,
and $f: X \to \an{4}$ a morphism of varieties.
Let $n \geq 3$ be an integer
and if $n=3$ then demand that
$f$ not be identically zero.
Consider
\[
    X
    \xrightarrow{\;\;f\;\;}
    \an{4}
    \xleftarrow{\;\;(k_{1\ldots n}) \mapsto k_1 + \ldots + k_n\;\;}
    \cone^n
\]
where $\cone = \zl(ad-bc) \subset \an{4}$.
Then the fiber product
\[
    X \times_{\an{4}} \cone^n
\]
is irreducible.
Same if the right map is
$\an{4} \xleftarrow{(v_{1\ldots n},w_{1\ldots n})
\mapsto v_1w_1^T + \ldots + v_nw_n^T} \an{4n}$.
\end{lemma}
\begin{proof}
Clearly
$\cone^n$ is irreducible,
a complete intersection of dimension $3n$,
hence Cohen-Macaulay.
The fibers of $g: (k_{1\ldots n}) \mapsto k_1 + \ldots + k_n$
have dimension\footnote{%
Note that for $n = 2$, which is excluded,
one fiber has dimension 3 instead of 2.
} $3n-4$ because
$(k_1 + \ldots + k_n - \vn,\Q_1,\ldots,\Q_n)$ is
for all $\vn \in \an{4}$
a regular sequence
in $\C[k_{1\ldots n}]$
of length $n+4$, see Lemma \ref{lemma:ci0}.
Let $A_n$ be the assertion that
the fiber $g^{-1}(\vn)$ is irreducible for all $\vn \in \an{4}$,
except $\vn=0$ if $n=3$.
Note that $A_n$ implies
 Lemma \ref{lemma:fib2} for that $n$,
by Lemma \ref{lemma:fib}.
The claim $A_3$ is checked using the computer;
it suffices to check it for two $\vn \neq 0$,
one in every orbit of the automorphism group of $\cone$.
% Macaulay 2. The result could depend on the ground field.
%
%N = 3;
%S = QQ[a_1..a_N,b_1..b_N,c_1..c_N,d_1..d_N];
%k = i -> {a_i,b_i,c_i,d_i};
%Q = i -> a_i*d_i-b_i*c_i;
%J = s -> ideal(join(for i from 1 to N list Q(i),(sum for i from 1 to N list k(i))-s));
%isPrime J({0,0,0,0}) -- true if N>=4
%isPrime J({1,0,0,0}) -- true
%isPrime J({1,0,0,1}) -- true
%
%N = 3;
%T = QQ[v1_1..v1_N,w1_1..w1_N,v2_1..v2_N,w2_1..w2_N];
%vw = i -> {v1_i*w1_i,v1_i*w2_i,v2_i*w1_i,v2_i*w2_i};
%J = s -> ideal((sum for i from 1 to N list vw(i))-s);
%isPrime J({0,0,0,0}) -- true if N>=4
%isPrime J({1,0,0,0}) -- true
%isPrime J({1,0,0,1}) -- true
%
We prove $A_n$, $n \geq 4$ by induction. Namely the fiber
$k_1 + \ldots + k_n = \vn$ is equal to
$\cone \times_{\an{4}} \cone^{n-1}$ for
\begin{equation}\label{eq:ccc}
\cone
\xrightarrow{\;\;
    k_n \mapsto -k_n + \vn
\;\;}
\an{4}
\xleftarrow{\;\;(k_{1 \ldots n-1}) \mapsto k_1 + \ldots + k_{n-1}\;\;}
 \cone^{n-1}
\end{equation}
and is irreducible because $A_{n-1}$
allows us to use Lemma \ref{lemma:fib2}
for $n-1$.
\qed\end{proof}
%------------------------------------------------
\begin{lemma}[Primality] \label{lemma:fib3}
Suppose $\px \subset \C[x_{1\ldots d}]$ is a prime ideal
for some integer $d \geq 1$.
Let $f \in \C[x_{1\ldots d}]^4$ be four polynomials.
Let $n \geq 3$ be an integer
and if $n=3$ then demand that at least 
one of the four components of $f$ not be in $\px$. 
Then $\qx \subset \C[x_{1\ldots d},k_{1\ldots n}]$ given by
\[
\qx = (\px,k_1+\ldots + k_n - f, \Q_1,\ldots,\Q_n)
\]
is a prime ideal.
Same for $\qx = (\px,v_1w_1^T + \ldots + v_nw_n^T - f) \subset
\C[x_{1\ldots d},v_{1\ldots n},w_{1\ldots n}]$.
\end{lemma}
\begin{proof}[Using Lemma \ref{lemma:fib2}]
Set
$\ax_\vn =
(k_1 + \ldots + k_n - \vn,\Q_1,\ldots,\Q_n) \subset \C[k_{1\ldots n}]$.
Let $A_n$ be the assertion that $\ax_\vn$
is prime for all $\vn \in \an{4}$,
except $\vn=0$ if $n=3$.
Check $A_3$ using the computer.
We show that $A_n$ implies
Lemma \ref{lemma:fib3} for that $n$;
the $A_n$
then follow for $n \geq 4$ by induction by a direct analog of \eqref{eq:ccc}.
By Lemma \ref{lemma:fib2},
$\zl(\qx)$ is irreducible,
so $\sqrt{\qx}$ is prime by the Nullstellensatz,
so it suffices to check that $\sqrt{\qx} = \qx$.
So for all $p \in \C[x_{1\ldots d},k_{1\ldots n}]$
with $p^m \in \qx$, $m \geq 2$
we must show $p \in \qx$.
Denote evaluation
at $x \in \zl(\px)$ by $p_x \in \C[k_{1\ldots n}]$.
Then $(p_x)^m \in \ax_{f(x)}$.
But $A_n$,
and the assumption on $f$ if $n=3$,
 imply that $\ax_{f(x)}$ is prime
for all $x$ in an open dense $U \subset \zl(\px)$,
so $x \in U$ implies $p_x \in \ax_{f(x)}$.
There exist a monomial order for $k_{1\ldots n}$ and
a Gr\"obner basis $G_\vn$ for $\ax_\vn$ that depends
polynomially on $\vn$,
with leading monomials independent of $\vn$;
use degrevlex and coprimality as in Lemma \ref{lemma:ci0}.
For $x \in U$,
Gr\"obner reduction of $p_x$ using $G_{f(x)}$ yields zero.
So Gr\"obner reduction with $x$ symbolic 
returns a $P \in \C[x_{1\ldots d}][k_{1\ldots n}]$
that vanishes if $x \in U$
and by continuity if $x \in \zl(\px)$.
Hence $P \in \px[k_{1\ldots n}]$ since $\px$ is prime.
But $p-P \in \qx$ by construction and therefore $p \in \qx$.
\qed\end{proof}
%n=3;
%k[i_]:={d[i],c[i],b[i],a[i]};
%vars=Flatten[Table[k[i],{i,1,n}]];
%Q[i_]:=a[i]*d[i]-b[i]*c[i];
%ideal[s_]:=Flatten[{Sum[k[i],{i,1,n}]-s,Table[Q[i],{i,1,n}]}];
%G[s_]:=GroebnerBasis[ideal[s],vars,MonomialOrder->DegreeReverseLexicographic]
%L[s_]:=Map[MonomialList[#,vars,DegreeReverseLexicographic]&,G[s]];
%s4=Array[s,4];
%G[s4]//TableForm
%Map[First,L[s4]]//TableForm
\begin{proof}[Without using Lemmas \ref{lemma:fib}, \ref{lemma:fib2}]
We modify the above proof
to show directly that $p_1p_2 \in \qx$ implies $p_1 \in \qx$ or $p_2 \in \qx$.
If $x \in \zl(\px)$ then $p_{1x}p_{2x} \in \ax_{f(x)}$.
Set $U_i = \{x \in U \mid p_{ix} \in \ax_{f(x)}\}$.
Primality of $\ax_{f(x)}$ for $x \in U$ implies $U_1 \cup U_2 = U$.
Let $P_i \in \C[x_{1\ldots d}][k_{1\ldots n}]$
be the result of Gr\"obner reduction applied to $p_i$.
Then $p_i-P_i \in \qx$,
and $P_i = 0$ if $x \in U_i$.
Expand $P_i = \sum_M P_{iM} M$ where $P_{iM} \in \C[x_{1\ldots d}]$
and $M$ runs over the monomials in $k_{1\ldots n}$.
Then $P_{1M_1}P_{2M_2} = 0$ for $x \in U = U_1 \cup U_2$ for all $M_1$ and $M_2$,
therefore $P_{1M_1}P_{2M_2} \in \px$ using primality of $\px$,
hence $P_{1M_1} \in \px$ or $P_{2M_2} \in \px$.
If $P_{1M} \in \px$ for all $M$ then $p_1 \in \qx$ and we are done.
If $P_{1M} \notin \px$ for one $M$ then $P_{2M} \in \px$ for all $M$ and
then $p_2 \in \qx$.
\qed\end{proof}
%-------------------------------------------------------------------